\documentclass[11pt]{article} 

\usepackage[margin=1in]{geometry}
\usepackage{float}
\usepackage{amssymb,amsfonts,amsmath}
\usepackage{enumerate}
\usepackage{enumitem}
\usepackage{caption}
\usepackage{bbm}
\usepackage{amsthm}
\usepackage{color}
\usepackage[colorlinks=true, linkcolor=red, urlcolor=blue, citecolor=blue]{hyperref}
\usepackage{comment}

\newcommand{\C}{\mathbb{C}}	                
\newcommand{\R}{\mathbb{R}}                     
\newcommand{\Z}{\mathbb{Z}}

\newcommand{\poly}{\text{poly}}

\newcommand{\pr}[1]{\mathbf{\normalfont Pr}\normalfont\lbrack #1 \rbrack} 
\newcommand{\ex}[1]{\mathbb{E}\normalfont\lbrack #1 \rbrack}
\newcommand{\bpr}[1]{\mathbf{\normalfont Pr}\normalfont \Big[#1 \Big]} 
\newcommand{\bex}[1]{\mathbb{E}\normalfont \Big[#1 \Big]}

\newcommand{\eps}{\epsilon}

\newcommand{\ttx}[1]{\texttt{#1}}

\newcommand{\ab}{\allowbreak}

\theoremstyle{theorem}
\newtheorem{theorem}{Theorem}
\theoremstyle{lemma}
\newtheorem{lemma}{Lemma}
\theoremstyle{corollary}
\newtheorem{corollary}{Corollary}
\theoremstyle{fact}
\newtheorem{fact}{Fact}
\theoremstyle{proposition}
\newtheorem{proposition}{Proposition}
\theoremstyle{definition}
\newtheorem{definition}{Definition}

\newtheorem{remark}{Remark}
\newtheorem{claim}{Claim}

\title{Perfect $L_p$ Sampling in a Data Stream \footnote{A preliminary version of this work appeared in FOCS 2018.}}
\author{
	Rajesh Jayaram\\
	Carnegie Mellon University\\
	\texttt{rkjayara@cs.cmu.edu}
	\and
	David P. Woodruff\\
	Carnegie Mellon University \\
	\texttt{dwoodruf@cs.cmu.edu}
	\thanks{The authors thank the partial support by the National Science Foundation under Grant No. CCF-1815840.}
}
\date{}
\begin{document}
	\maketitle
	
	\begin{abstract}
		In this paper, we resolve the one-pass space complexity of {\it perfect} $L_p$ sampling for $p \in (0,2)$ in a stream. 
Given a stream of updates (insertions and deletions) to the coordinates of an underlying vector $f \in \R^n$, a perfect $L_p$ sampler must output an index $i$ with probability $|f_i|^p/\|f\|_p^p$, and is allowed to fail with some probability $\delta$. So far, for $p > 0$ no algorithm has been shown to solve the problem exactly using $\text{poly}( \log n)$-bits of space. In 2010, Monemizadeh and Woodruff introduced an {\it approximate $L_p$ sampler}, which outputs $i$ with probability $(1 \pm \nu)|f_i|^p /\|f\|_p^p$, using space polynomial in $\nu^{-1}$ and $\log(n)$. The space complexity was later reduced by Jowhari, Sa\u{g}lam, and Tardos to roughly $O(\nu^{-p} \log^2 n \log \delta^{-1})$ for $p \in (0,2)$, which matches the $\Omega(\log^2 n \log \delta^{-1})$ lower bound in terms of $n$ and $\delta$, but is loose in terms of $\nu$.

				Given these nearly tight bounds, it is perhaps surprising that no lower bound exists in terms of $\nu$---not even a bound of $\Omega(\nu^{-1})$ is known. In this paper, we explain this phenomenon by demonstrating the existence of an $O(\log^2 n \log \delta^{-1})$-bit {\it perfect} $L_p$ sampler for $p \in (0,2)$. This shows that $\nu$ need not factor into the space of an $L_p$ sampler, which closes the complexity of the problem for this range of $p$. For $p=2$, our bound is $O(\log^3 n \log \delta^{-1})$-bits, which matches the prior best known upper bound of $O(\nu^{-2}\log^3n \log \delta^{-1})$, but has no dependence on $\nu$. For $p<2$, our bound holds in the random oracle model, matching the lower bounds in that model. However, we show that our algorithm can be derandomized with only a $O((\log \log n)^2)$ blow-up in the space (and no blow-up for $p=2$). Our derandomization technique is quite general, and can be used to derandomize a large class of linear sketches, including the more accurate count-sketch variant of Minton and Price \cite{minton2014improved}, resolving an open question in that paper. 
		
		Finally, we show that a $(1\pm\eps)$ relative error estimate of the frequency $f_i$ of the sampled index $i$ can be obtained using an additional $O( \eps^{-p} \log n)$-bits of space for $p < 2$, and $O( \eps^{-2} \log^2 n)$ bits for $p=2$, which was possible before only by running the prior algorithms with $\nu = \eps$.

	\end{abstract}

\thispagestyle{empty}
\clearpage
\tableofcontents
\thispagestyle{empty}
\clearpage
\setcounter{page}{1}

\section{Introduction}

The streaming model of computation has become increasingly important for the analysis of massive datasets, where the sheer size of the input imposes stringent restrictions on the resources available to algorithms. 
Examples of such datasets include internet traffic logs, sensor networks, financial transaction data, database logs, and scientific data streams (such as huge experiments in particle physics, genomics, and astronomy). 
Given their prevalence, there is a large body of literature devoted to designing extremely efficient one-pass algorithms for analyzing data streams. We refer the reader to \cite{babcock2002models, muthukrishnan2005data} for surveys of these algorithms and their applications.  

More recently, the technique of sampling has proven to be tremendously powerful for the analysis of data streams. Substantial literature has been devoted to the study of sampling for problems in big data  \cite{muthukrishnan2005data,DBLP:books/sp/16/Haas16, cohen2015stream,cohen2009stream, DBLP:journals/jcss/CohenDKLT14,DBLP:journals/pvldb/CohenCD11, DBLP:journals/tocs/EstanV03,DBLP:conf/sigmod/GibbonsM98,DBLP:books/lib/Knuth98,DBLP:journals/pvldb/MankuM12,DBLP:journals/toms/Vitter85, cohen2012don, DBLP:journals/vldb/GemullaLH08,  DBLP:conf/vldb/GemullaLH06}, with applications to network traffic analysis \cite{thottan2010anomaly,huang2007communication,gilbert2001quicksand,mai2006sampled,duffield2004sampling}, databases \cite{olken1993random, DBLP:books/sp/16/Haas16, haas1996selectivity, haas1992sequential,lipton1990practical,lipton1995query},  distributed computation \cite{woodruff2016distributed, cormode2010optimal, cormode2012continuous,tirthapura2011optimal}, and low-rank approximation \cite{woodruff2016distributed, frieze2004fast,deshpande2006adaptive}. While several models for sampling in data streams have been proposed \cite{babcock2002sampling, andoni2010streaming, cormode2010optimal}, one of the most widely studied are the $L_p$ samplers introduced in \cite{monemizadeh20101}. Roughly speaking, given a vector $f \in \R^n$, the goal of an $L_p$ sampler is to return an index $i \in \{1,2,\dots , n\}$ with probability $|f_i|^p/\|f\|_p^p$. In the data stream setting, the vector $f$ is given by a sequence of updates (insertions or deletions) to its coordinates of the form $f_i \leftarrow f_i + \Delta$, where $\Delta$ can either be positive or negative. A $1$-pass $L_p$ sampler must return an index given only one pass through the updates of the stream.

Since their introduction, $L_p$ samplers have been utilized to develop alternative algorithms for important streaming problems, such as the heavy hitters problem, $L_p$ estimation, cascaded norm estimation, and finding duplicates in data streams \cite{andoni2010streaming, monemizadeh20101, Jowhari:2011, braverman2012optimal}. For the case of $p=1$ and insertion only streams, where the updates to $f$ are strictly positive, the problem is easily solved using $O(\log n)$ bits of space with the well-known reservoir sampling algorithm \cite{vitter1985random}. When deletions to the stream are allowed or when $p \neq 1$, however, the problem is more complicated. In fact, the question of whether such samplers even exist was posed by Cormode, Murthukrishnan, and Rozenbaum in \cite{cormode2005summarizing}. Later on, Monemizadeh and Woodruff demonstrated that, if one permits the sampler to be \textit{approximately} correct, such samplers are indeed possible \cite{monemizadeh20101}. We formally state the guarantee given by an approximate $L_p$ sampler below.
 
\begin{definition}\label{def:sampler}
Let $f \in \R^n$ and $\nu \in [0,1)$. For $p > 0$, an \textit{approximate $L_p$ sampler} with $\nu$-relative error is an algorithm which returns an index $i \in \{1,2,\dots,n\}$ such that for every $j \in\{1,2,\dots,n\}$ 
	\[\pr{i = j}=	\frac{|f_j|^p}{\|f\|_p^p}(1 \pm \nu) + O(n^{-c})	\]
	Where $c \geq 1$ is some arbitrarily large constant. For $p=0$, the problem is to return $j$ with probability $(1 \pm \nu)\max\{1,|f_j|\}/|\{j\: : \: f_j \neq 0 \}| +O(n^{-c})$, If $\nu = 0$, then the sampler is said to be \textit{perfect}.  An $L_p$ sampler is allowed to output \texttt{FAIL} with some probability $\delta$. However, in this case it must not output any index.
\end{definition}
 
The one-pass approximate $L_p$ sampler introduced in \cite{monemizadeh20101} requires $\poly (\nu^{-1}, \log n)$ space, albeit with rather large exponents. 
Later on, in \cite{andoni2010streaming}, the complexity was reduced significantly to $O(\nu^{-p}\log^3(n) \log(1/\delta))$-bits\footnote{We note that previous works \cite{Jowhari:2011, kapralov2017optimal} have cited the sampler of \cite{andoni2010streaming} as using $O(\log^3(n))$-bits of space, however the space bound given in their paper is in \textit{machine words}, and is therefore a $O(\log^4(n))$ bit bound with $\delta = 1/\poly(n)$. In order to obtain an $O(\log^3(n)\log(1/\delta))$ bit sampler, their algorithm must be modified to use fewer repetitions.}
 for $p \in [1,2]$, using a technique known as \textit{precision sampling}. Roughly, the technique of precision sampling consists of scaling the coordinates $f_i$ by random variable coefficients $1/t_i$ as the updates arrive, resulting in a new stream vector $z \in \R^n$ with $z_i = f_i/t_i$. The algorithm then searches for all $z_i$ which cross a certain threshold $T$. Observe that if $t_i = u_i^{1/p}$ where $u_i$ is uniform on $[0,1]$, then the probability that $f_i/t_i \geq T$ is precisely $ \pr{u_i < |f_i|^p / T^p} = |f_i|^p/T^p$. By running an $L_p$ estimation algorithm to obtain $T \in [\frac{1}{2}\|f\|_p, \frac{3}{2}\|f\|_p]$, an $L_p$ sampler can then return any $i$ with $z_i \geq T$ as its output. These heavy coordinates can be found using any of the well-known $\eta$-heavy hitters algorithms for a sufficiently small precision $\eta$.  

\begin{figure}[]
	\begin{center}
		\begin{tabular}{| c | c| c| c | } 
			\hline
			$L_p$ sampling upper bound (bits) & $p$ range &  Notes  & Citation \\ \hline
			$O(\log^3(n))$ & $p=0$ & perfect $L_0$ sampler, $\delta = 1/\poly(n)$ & \cite{frahling2008sampling} \\ \hline
			$O(\log^2(n) \log(1/\delta_2))	$ &  $p =0 $ & perfect $L_0$ sampler& \cite{Jowhari:2011} \\ \hline
			$\poly \log(\nu^{-1},n)$	& $p \in [0,2]$ & $\delta = 1/\poly(n)$ & \cite{monemizadeh20101} \\ \hline
			$O(\nu^{-p}\log^3(n)\log(1/\delta))$ & $p \in [1,2]$ &  $(1\pm\nu)$-relative error & \cite{andoni2010streaming} \\ \hline
			$O(\nu^{-\max\{1,p\}}\log^2(n)\log(1/\delta)) $& $ p \in (0,2) \setminus \{1\} $ &$(1\pm\nu)$-relative error &\cite{Jowhari:2011} \\ \hline
				$O(\nu^{-1}\log(\nu^{-1}) \log^2(n) \log(1/\delta))	$ &  $p =1 $ & $(1\pm\nu)$-relative error& \cite{Jowhari:2011} \\ \hline
				$O(\log^2(n) \log(1/\delta) )	$ &  $p \in (0,2) $ & \begin{tabular}{c}
				perfect $L_p$ sampler, \\ random oracle model,\\ matches lower bound
			\end{tabular}& This work \\ \hline
		
			$O(\log^2(n) \log(1/\delta)(\log\log n)^2 )	$ &  $p \in (0,2) $ & 
				\begin{tabular}{c}
					perfect $L_p$ sampler
				\end{tabular}
		& This work \\ \hline
			$O(\log^3(n) \log(1/\delta))	$ &  $p =2  $ & perfect $L_2$ sampler & This work \\ \hline
			$O(\log^3(n) )	$ &  $p \in (0,2)  $ & $\delta = 1/\poly(n)$ & This work \\ \hline
		
		\end{tabular} 
	\end{center} \caption{Evolution of one pass $L_p$ sampling upper bounds, with the best known lower bound of $\Omega(\log^2(n)\log(1/\delta))$ for $p \geq 0$ \cite{kapralov2017optimal} (see also \cite{Jowhari:2011} for a lower bound for constant $\delta$). } \label{fig:results}
\end{figure} 

 Using a tighter analysis of this technique with the same scaling variables $t_i = u_i^{1/p}$, Jowhari, Sa\u{g}lam, and Tardos reduced the space complexity of $L_p$ sampling for $p<2$ to $O(\nu^{-\max\{1,p\}} \log^2(n) \allowbreak \log(1/\delta))$-bits for $p \in (0,2) \setminus \{1\}$, and $O(\nu^{-1}\log(\nu^{-1})\allowbreak \log^2(n)\log(1/\delta))$ bits of space for $p=1$ \cite{Jowhari:2011}. Roughly speaking, their improvements result from a more careful consideration of the precision $\eta$ needed to determine when a $z_i$ crosses the threshold, which they do via the tighter tail-error guarantee of the well-known count-sketch heavy hitters algorithm \cite{Charikar:2002}. In addition, they give an $O(\log^2(n)\log(1/\delta))$ perfect $L_0$ sampler, and demonstrated an $\Omega(\log^2(n))$-bit lower bound for $L_p$ samplers for any $p \geq 0$. Recently, this lower bound was extended to $\Omega(\log^2(n)\log(1/\delta))$ \cite{kapralov2017optimal} bits, which closes the complexity of  the problem for $p=0$.
 
 For $p \in (0,2)$, this means that the  upper and lower bounds for $L_p$ samplers are tight in terms of $n, \delta$, but a gap exists in the dependency on $\nu$. Being the case, it would seem natural to search for an $\Omega(\nu^{-p}\log^2(n)\log(1/\delta))$ lower bound to close the complexity of the problem. It is perhaps surprising, therefore, that no lower bound in terms of $\nu$ exists -- not even an $\Omega(\nu^{-1})$ bound is known. This poses the question of whether the $\Omega(\log^2(n)\log(1/\delta))$ lower bound is in fact correct.

\subsection{Our Contributions}
In this paper, we explain the phenomenon of the lack of an $\Omega(\nu^{-1})$ lower bound by showing that $\nu$ need not enter the space complexity of an $L_p$ sampler at all. In other words, we demonstrate the existence of \textit{perfect $L_p$ samplers} using $O(\log^2(n)\log(1/\delta) (\log\log n)^2)$-bits of space for $p \in (0,2)$, thus resolving the space complexity of the problem up to $\log \log(n)$ terms\footnote{A previous version of this work claimed $O(\log^2(n)\log(1/\delta))$ bits of space for $p<2$, but contained an error in the derandomization. Thus, this bound only held in the random oracle model. In the present version we correct this derandomization using a slightly different algorithm, albeit with a $(\log \log n)^2$ blow-up in the space. The algorithm from the previous version can be found in Appendix \ref{sec:appendix}, along with a new analysis of its derandomization which allows it to run in $O(\log^2(n) (\log \log (n))^2)$-bits of space.}. In the random oracle model, where we are given random access to an arbitrarily long tape of random bits which do not count against the space of the algorithm, our upper bound is $O(\log^2(n) \log(1/\delta))$, which matches the lower bound in the random oracle model. For $p=2$, our space is $O(\log^3(n)\log(1/\delta))$-bits, which matches the best known upper bounds in terms of $n,\delta$, yet again has no dependency on $\nu$. In addition, for $p<2$ and the high probabiltiy regime of $\delta < 1/n$, we obtain a $O(\log^3(n))$-bit perfect $L_p$ sampler, which also tightly matches the lower bound without paying the extra $(\log\log n)^2$ factor. A summary of the prior upper bounds for $L_p$ sampling, along with the contributions of this work, is given in Figure \ref{fig:results}.

In addition to outputting a perfect sample $i$ from the stream, for $p \in (0,2)$ we also show that, conditioned on an index being output, given an additional additive $O(\min\{\eps^{-2},\eps^{-p} \log(\frac{1}{\delta_2}) \}\ab\log(n)\ab\log(1/\delta_2))$-bits we can provide a $(1 \pm \eps)$ approximation of the frequency $|f_i|$ with probability $1-\delta_2$. This separates the space dependence on $\log^2(n)$ and $\eps$ for frequency approximation, allowing us to obtain a $(1 \pm \eps)$ approximation of $|f_i|$ in $O(\log^2(n) + \eps^{-p}\log(n))$ bits of space with constant probability, whereas before this required $O(\eps^{-p}\log^2(n))$ bits of space. For $p=2$, our bound is $O(\eps^{-2}\log^2(n)\log(1/\delta_2))$, which still improves upon the prior best known bounds for estimating the frequency by an $O(\log(n))$-factor. Finally, we show an $\Omega(\eps^{-p}\log(n) \log(1/\delta_2))$ bits of space lower bound for producing the $(1 \pm \eps)$ estimate (conditioned on an index being returned).

\subsection{Applications}
Since their introduction, it has been observed that $L_p$ samplers can be used as a building block in algorithms for many important streaming problems, such as finding heavy hitters, $L_p$-norm estimation, cascaded norm estimation, and finding duplicates in data streams \cite{andoni2010streaming, monemizadeh20101, Jowhari:2011, braverman2012optimal}.
$L_p$ samplers, particularly for $p=1$, are often used as a black-box subroutine to design representative histograms of $f$ on which more complicated algorithms are run \cite{gibbons1997fast, gibbons1998new,olken1993random,gilbert2002summarize,huang2007communication,cormode2005summarizing}. For these black-box applications, the only property of the samplers needed is the distribution of their samples. Samplers with relative error are statistically biased and, in the analysis of more complicated algorithms built upon such samplers, this bias and its propagation over multiple samples must be accounted for and bounded. The analysis and development of such algorithms would be simplified dramatically, therefore, with the assumptions that the samples were truly uniform (i.e., from a perfect $L_1$ sampler). In this case, no error terms or variational distance need be accounted for. Our results show that such an assumption is possible without affecting the space complexity of the sampler.

Note that in Definition \ref{def:sampler}, we allow a perfect sampler to have $n^{-c+1}$ variation distance to the true $L_p$ distribution. We note that this definition is in line with prior work, observing that even the perfect $L_0$ sampler of \cite{Jowhari:2011} incurs such an error from derandomizing with Nisan's PRG. Nevertheless, this error will never be detected if the sampler is run polynomially many times in the course of constructing a histogram, and such a sampler is therefore statistically indistinguishable from a truly uniform sampler and can be used as a black box.

Another motivation for utilizing perfect $L_p$ samplers comes from applications in privacy.  Here $f \in \mathbb{R}^n$ is some underlying dataset, and we would like to reveal a sample $i \in [n]$ drawn from the $L_p$ distribution over $f$ to some external party without revealing too much global information about $f$ itself. Using an approximate $L_p$ sampler introduces a $(1 \pm\nu)$ multiplicative bias into the sampling probabilities, and this bias can depend on global properties of the data. For instance, such a sampler might bias the sampling probabilities of a large set $S$ of coordinates by a $(1+\nu)$ factor if a certain global property $P$ holds for $f$, and may instead bias them by $(1 - \nu)$ if a disjoint property $P'$ holds. Using only a small number of samples, an adversary would then be able to distinguish whether $P$ or $P'$ holds by determining how these coordinates were biased. On the other hand, the bias in the samples produced by a perfect $L_p$ sampler is polynomially small, and thus the leakage of global information could be substantially smaller when using one, though one would need to formally define a notion of leakage and privacy for the given application.  

\subsection{Our Techniques}
Our main algorithm is inspired by the precision sampling technique used in prior works \cite{andoni2010streaming, Jowhari:2011}, but with some marked differences. To describe how our sampler achieves the improvements mentioned above, we begin by observing that all $L_p$ sampling algorithms since \cite{andoni2010streaming} have adhered to the same algorithmic template (shown in Figure \ref{fig:template}). This template employs the classic \textit{count-sketch} algorithm of \cite{charikar2002finding} as a subroutine, which is easily introduced. For $k \in \mathbb{N}$, let $[k]$ denote the set $\{1,2,\dots,k\}$. Given a precision parameter $\eta$, count-sketch selects pairwise independent hash functions $h_j:[n]\to [6/\eta^2]$ and $g_j:[n] \to \{1,-1\}$, for $j=1,2,\dots,d$ where $d = \Theta(\log(n))$. Then for all $i \in [d], \:j \in [6/\eta^2]$, it computes the following linear function $A_{i,j} = \sum_{k \in [n], h_i(k) = j} g_i(k) f_k$, and outputs an approximation $y$ of $f$ given by $y_k = \text{median}_{i \in [d]} \{g_i(k) A_{i,h_i(k)}\}$. We will discuss the estimation guarantee of count-sketch at a later point.

The algorithmic template is as follows. First, perform some linear transformation on the input vector $f$ to obtain a new vector $z$. Next, run an instance $A$ of count-sketch on $z$ to obtain the estimate $y$. Finally, run some statistical test on $y$. If the test fails, then output \ttx{FAIL}, otherwise output the index of the largest coordinate (in magnitude) of $y$. We first describe how the sampler of \cite{Jowhari:2011} implements the steps in this template. Afterwards we describe the different implementation decisions made in our algorithm that allow it to overcome the limitations of prior approaches. 

\begin{center}
	
	\begin{figure}
	\fbox{\parbox{\textwidth}{ \texttt{Input:} $f \in \R^n$\\ \texttt{Output:} a sampled index $i^* \in [n]$
			\begin{enumerate}[topsep=0pt,itemsep=-1ex,partopsep=1ex,parsep=1ex] 
				\item Perform a linear transformation on $f$ to obtain $z$.
				\item Run instance $A$ of count-sketch on $z$ to obtain the estimate $y$.
				\item Find $i^* = \arg \max_i |y_i|$. Then run a statistical test on $y$ to decide whether to output $i^*$ or \texttt{FAIL}.
			\end{enumerate}
	}}\caption{Algorithmic Template for $L_p$ sampling} \label{fig:template}
\end{figure}

\end{center}

\paragraph*{Prior Algorithms.} The samplers of \cite{Jowhari:2011, andoni2010streaming} utilize the technique known as \textit{precision sampling}, which employs the following linear transformation. The algorithms first generate random variables $(t_1,\dots,t_n)$ with limited independence, where each $t_i \sim \text{Uniform}[0,1]$. Next, each coordinate $f_i$ is scaled by the coefficient $1/t_i^{1/p}$ to obtain the transformed vector $z \in \R^n$ given by $z_i = f_i/t_i^{1/p}$, thus completing Step $1$ of Figure \ref{fig:template}. For simplicity, we now restrict to the case of $p=1$ and the algorithm of \cite{Jowhari:2011}. The goal of the algorithm is then to return an item $z_i$ that crosses the threshold $|z_i |> \nu^{-1}R$, where $R = \Theta(\|f\|_1)$ is a constant factor approximation of the $L_1$. Note the probability that this occurs is proportional to $\nu|f_i|/\|f\|_1$.

Next, implementing the second step of Figure \ref{fig:template}, the vector $z$ is hashed into count-sketch to find an item that has crossed the threshold. Using the stronger \textit{tail-guarantee} of count-sketch, the estimate vector $y$ satisfies $\|y - z\|_\infty \leq \sqrt{\eta}\|z_{\text{tail}(1/\eta)}\|_2$, where $z_{\text{tail}(1/\eta)}$ is $z$ with the $1/\eta$ largest coordinates (in magnitude) set to $0$. 
Now the algorithm runs into trouble when it incorrectly identifies $z_i$ as crossing the threshold when it has not, or vice-versa. However, if the tail error $\sqrt{\eta}\|z_{\text{tail}(1/\eta)}\|_2$ is at most $O(\|f\|_1)$, then since $t_i$ is a uniform variable the probability that $z_i$ is close enough to the threshold to be misidentified is $O(\nu)$, which results in at most $(1 \pm \nu)$ relative error in the sampling probabilities. Thus it will suffice to have $\sqrt{\eta}\|z_{\text{tail}(1/\eta)}\|_2 = O(\|f\|_1)$ with probability $1-\nu$.
To show that this is the case, consider the level sets $I_k = \{z_i  \; |\;  z_i \in (\frac{\|f\|_p}{2^{(k+1)/p}}, \frac{\|f\|_p}{2^{k/p}})\}$, and note $\ex{|I_k|} = 2^k$.  We observe here that results of \cite{Jowhari:2011} can be partially attributed to the fact that for $p<2$, the total contribution $\Theta(\frac{\|f\|_p^2}{2^{2k/p}}|I_k|)$ of the level sets to $\|z\|_2^2$ decreases geometrically with $k$, and so with constant probability we have $\|z\|_2 = O(\|f\|_p)$. Moreover, if one removes the top $\log(1/\nu)$ largest items, the contribution of the remaining items to the $L_2$ is $O(\|f\|_1)$ with probability $1-\nu$.  So taking $\eta = \log(1/\nu)$, the tail error from count-sketch has the desired size. Since the tail error does not include the $1/\eta$ largest coordinates, this holds even conditioned on a fixed value $t_{i^*}$ of the maximizer. 

Now with probability $\nu$ the guarantee on the error from the prior paragraph does not hold, and in this case one \textit{cannot} still output an index $i$, as this would result in a $\nu$-\textit{additive} error sampler. Thus, as in Step $3$ of Figure \ref{fig:template}, the algorithm must implement a statistical test to check that the guarantee holds. To do this, using the values of the largest $1/\eta$ coordinates of $y$, they produce an estimate of the tail-error and output \texttt{FAIL} if it is too large. Otherwise, the item $i^* = \arg \max_i |y_i|$ is output if $|y_{i^*}| > \nu^{-1}R$. The whole algorithm is run $O(\nu^{-1}\log(1/\delta))$ times so that an index is output with probability $1-\delta$.

\paragraph*{Our Algorithm.}
Our first observation is that, in order to obtain a truly perfect sampler, one needs to use different scaling variables $t_i$. Notice that the approach of scaling by inverse uniform variables and returning a coordinate which reaches a certain threshold $T$ faces the obvious issue of what to return when more than one of the variables $|z_i|$ crosses $T$. This is solved by simply outputting the maximum of all such coordinates. However, the probability of an index becoming the maximum \textit{and} reaching a threshold is drawn from an entirely different distribution, and for uniform variables $t_i$ this distribution does not appear to be the correct one. To overcome this, we must use a distribution where the maximum index $i$ of the variables $(|f_1t_2^{-1/p}|, |f_2t_2^{-1/p}|,\dots,|f_nt_n^{-1/p}|)$ is drawn \textit{exactly} according to the $L_p$ distribution $|f_i|^p/\|f\|_p^p$. We observe that the distribution of exponential random variables has precisely this property, and thus to implement Step $1$ of Figure \ref{fig:template} we set $z_i = f_i/t_i^{1/p}$ where $t_i$ is an exponential random variable. We remark that exponential variables have been used in the past, such as for $F_p$ moment estimation, $p > 2$, in \cite{andoni2010streaming} and regression in \cite{DBLP:journals/corr/abs-1305-5580}. However it appears that their applicability to sampling has never before been exploited.

Next, we carry out the count-sketch step by hashing our vector $z$ into a count-sketch data structure $A$. Because we are only interested in the maximizer of $z$, we develop a modified version of count-sketch, called \textit{count-max}. Instead of producing an estimate $y$ such that $\|y - z\|_\infty$ is small, count-max simple checks, for each $i \in [n]$, how many times $z_i$ hashed into the largest bucket (in absolute value) of a row of $A$. If this number is at least a $4/5$-fraction of the total number of rows, count-max declares that $z_i$ is the maximizer of $z$. We show that with high probability, count-max never incorrectly declares an item to be the maximizer, and moreover if $z_i > 20 (\sum_{j \neq i} z_j^2)^{1/2}$, then count-max will declare $i$ to be the maximizer. Using the \textit{min-stability} property of exponential random variables, we can show that the maximum item $|z_{i^*}| = \max\{|z_i|\}$ is distributed as $\|f\|_p/E^{1/p}$, where $E$ is another exponential random variable. Thus $|z_{i^*}| = \Omega(\|f\|_p)$ with constant probability. Using a more general analysis of the $L_2$ norm of the level sets $I_k$, we can show that $(\sum_{j \neq i^*} z_j^2)^{1/2}= O(\|f\|_p)$.   If all these events occur together (with sufficiently large constants), count-max will correctly determine the coordinate $i^* = \arg \max_i \{|z_i|\}$. However, just as in \cite{Jowhari:2011}, we cannot output an index anyway if these conditions do not hold, so we will need to run a statistical test to ensure that they do. 

\paragraph*{The Statistical Test.} 

To implement Step $3$ of the template, our algorithm simply tests whether count-max declares any coordinate $i \in [n]$ to be the maximizer, and we output \texttt{FAIL} if it does not. This approach guarantees that we correctly output the maximizer conditioned on not failing.
The primary technical challenge will be to show that, conditioned on $i = \arg \max_i \{|z_i|\}$. for some $i$, the probability of failing the statistical test \textit{does not depend on $i$}. In other words, conditioning on $|z_i|$ being the maximum does not change the failure probability. Let $z_{D(k)}$ be the $k$-th order statistic of $z$ (i.e., $|z_{D(1)}| \geq |z_{D(2)}| \geq \dots \geq |z_{D(n)}|)$. Here the $D(k)$'s are known as \textit{anti-ranks}. To analyze the conditional dependence, we must first obtain a closed form for $z_{D(k)}$ which separates the dependencies on $k$ and $D(k)$. Hypothetically, if $z_{D(k)}$ depended only on $k$, then our statistical test would be completely independent of $D(1)$, in which case we could safely fail whenever such an event occurred. Of course, in reality this is not the case. Consider the vector $f= (100n,1,1,1,\dots,1) \in \R^{n}$ and $p=1$. Clearly we expect $z_1$ to be the maximizer, and moreover we expect a gap of $\Theta(n)$ between $z_{1}$ and $z_{D(2)}$. On the other hand, if you were told that $D(1) \neq 1$, it is tempting to think that $z_{D(1)}$ just \textit{barely} beat out $z_1$ for its spot as the max, and so $z_1$ would not be far behind. Indeed, this intuition would be correct, and one can show that the probability that $z_{D(1)} - z_{D(2)} > n$ conditioned on $D(1) = i$ changes by an additive constant depending on whether or not $i=1$. Conditioned on this gap being smaller or larger, we are more or less likely (respectively) to output \texttt{FAIL}. In this setting, the probability of conditional failure can change by an $\Omega(1)$ factor depending on the value of $D(1)$. 

To handle scenarios of this form, our algorithm will utilize an additional linear transformation in Step $1$ of the template. Instead of only scaling by the random coefficients $1/t_i^{1/p}$, our algorithm first \textit{duplicates} the coordinates $f_i$ to remove all heavy items from the stream. If $f$ is the vector from the example above and $F$ is the duplicated vector, then after $\poly(n)$ duplications all copies of the heavy item $f_1$ will have weight at most $|f_1|/\|F\|_1 < 1/\poly(n)$. By uniformizing the relative weight of the coordinates, this washes out the dependency of $|z_{D(2)}|$ on $D(1)$, since $\|F_{-D(1)}\|_p^p = (1 \pm n^{-\Omega(c)}) \|F_{-j} \|_p^p$ after $n^c$ duplications, for any $j \in [n^c]$. Notice that this transformation blows-up the dimension of $f$ by a $\poly(n)$ factor. However, since our space usage is always $\poly \log(n)$, the result is only a constant factor increase in the complexity. 

After duplication, we scale $F$ by the coefficients $1/t_i^{1/p}$, and the rest of the algorithm proceeds as described above. Using expressions for the order statistics $z_{D(k)}$ which separate the dependence into the anti-ranks $D(j)$ and a set of exponentials  $E_1,E_2,\dots E_n$ \textit{independent} of the anti-ranks, after duplication we can derive tight concentration of the $z_{D(k)}$'s conditioned on fixed values of the $E_i$'s. Using this concentration result, we decompose our count-max data structure $A$ into two component variables: one independent of the anti-ranks (the independent component), and a small adversarial noise of relative weight $n^{-c}$. In order to bound the effect of the adversarial noise on the outcome of our tests we must $\mathbf{1)}$ randomize the threshold for our failure condition and $\mathbf{2)}$ demonstrate the anti-concentration of the resulting distribution over the independent components of $A$. This will demonstrate that with high probability, the result of the statistical test is completely determined by the value of the independent component, which allows us to fail without affecting the conditional probability of outputting $i \in [n]$.

\paragraph*{Derandomization}
Now the correctness of our sampler crucially relies on the full independence of the $t_i$'s to show that the variable $D(1)$ is drawn from precisely the correct distribution (namely, the $L_p$ distribution $|f_i|^p/\|f\|_p^p$). Being the case, we cannot directly implement our algorithm using any method of limited independence. In order to derandomize the algorithm from requiring full-independence, we will use a combination of Nisan's pseudorandom generator \cite{nisan1992pseudorandom}, as well as an extension of the recent PRG of \cite{gopalan2015pseudorandomness} which fools certain classes of \textit{Fourier transforms}. We first use a closer analysis of the seed length Nisan's generator requires to fool the randomness required for the count-max data structure, which avoids the standard $O(\log(n))$-space blowup which would be incurred by using Nisan's as a black box. Once the count-max has been derandomized, we demonstrate how the PRG of \cite{gopalan2015pseudorandomness} can be used to fool \textit{arbitrary} functions of $d$-halfspaces, so long as these half-spaces have bounded bit-complexity. We use this result to derandomize the exponential variables $t_i$ with a seed of length $O(\log^2(n) (\log \log n)^2)$, which will allow for the total derandomization of our algorithm for $\delta = \Theta(1)$ and $p < 2$ in the same space. 

Our derandomization technique is in fact fairly general, and can be applied to streaming algorithms beyond the sampler in this work. Namely, we demonstrate that \textit{any} streaming algorithm which stores a linear sketch $A \cdot f$, where the entries of $A$ are independent and can be sampled from with $O(\log(n))$-bits, can be derandomized with only a $O((\log \log n)^2)$-factor increase in the space requirements (see Theorem \ref{thm:derandomGeneral}). This improves the $O(\log(n))$-blow up incurred from black-box usage of Nisan's PRG. As an application, we derandomize the count-sketch variant of Minton and Price \cite{minton2014improved} to use $O( \eps^{-2}\log^2(n) (\log \log n)^2)$-bits of space, which gives improved concentration results for count-sketch when the hash functions are fully-independent. The problem of improving the derandomization of \cite{minton2014improved} beyond the black-box application of Nisan's PRG was an open problem. We remark that using $O(1/\eps^2 \log^3(n))$-bits of space in the classic count sketch of \cite{charikar2002finding} has strictly better error guarantees that those obtained from derandomizing \cite{minton2014improved} with Nisan's PRG to run in the same space. Our derandomization, in contrast, demonstrates a strong improvement on this, obtaining the same bounds with an $O((\log \log n)^2)$ instead of an $O(\log(n))$ factor blowup.


 
\paragraph*{Case of $p=2$.} Recall for $p<2$, we could show that the $L_2$ norm of the level sets $I_k$ decays geometrically with $k$. More precisely, for any $\gamma$ we have $\|z_{\text{tail}(\gamma)}\|_2 = O(\|F\|_p \gamma^{-1/p +1/2})$ with probability $1-O(e^{-\gamma})$. Using this, we actually do not need the tight concentration of the $z_{D(k)}$'s, since we can show that the top $n^{c/10}$ coordinates change by at most $(1 \pm n^{-\Omega(c)})$ depending on $D(1)$, and the $L_2$ norm of the remaining coordinates is only an $O(n^{-c/10(1/p - 1/2)})$ fraction of the whole $L_2$, and can thus be absorbed into the adversarial noise. For $p=2$ however, each level set $I_k$ contributes weight $O(\|F\|_p^2)$ to $\|z\|_2^2$, so  $\|z_{\text{tail}(\gamma)}\|_2 = O(\sqrt{\log(n)}\|F\|_p )$ even for $\gamma = \poly(n)$. Therefore, for $p=2$ it is essential that we show concentration of the $z_{D(k)}$'s for \textit{nearly all} $k$. Since $\|z\|_2^2$ will now be larger than $\|F\|_2^2$ by a factor of $\log(n)$ with high probability, count-max will only succeed in outputting the largest coordinate when it is an $O(\sqrt{\log(n)})$ factor larger than expected. This event occurs with probability $1/\log(n)$, so we will need to run the algorithm $\log(n)$ times in parallel to get constant probability, for a total $O(\log^3 n)$-bits of space. Using the same $O(\log^3(n))$-bit Nisan PRG seed for all $O(\log(n))$ repititions, we show that the entire algorithm for $p=2$ can be derandomized to run in $O(\log^3 n \log 1/\delta)$-bits of space. 

\paragraph*{Optimizing the Runtime.}
In addition to our core sampling algorithm, we show how the linear transformation step to construct $z$ can be implemented via a parameterized rounding scheme to improve the update time of the algorithm without affecting the space complexity, giving a run-time/relative sampling error trade-off. By rounding the scaling variables $1/t_i^{1/p}$ to powers of $(1 + \nu)$, we discretize their support to have size $O(\nu\log(n))$. We then simulate the update procedure by sampling from the distribution over updates to our count-max data-structure $A$ of duplicating an update and hashing each duplicate independently into $A$. Our simulation utilizes results on efficient generation of binomial random variables, through which we can iteratively reconstruct the updates to $A$ bin-by-bin instead of duplicate-by-duplicate. In addition, by using an auxiliary heavy-hitter data structure, we can improve our query time from the na\"ive $O(n)$ to $O(\poly \log(n))$ without increasing the space complexity.

\paragraph*{Estimating the Frequency.}
We show that allowing an additional additive $O(\min\{\eps^{-2}, \eps^{-p}\allowbreak  \log(\frac{1}{\delta_2}) \}\ab\log n\ab\log \delta_2^{-1})$ bits of space, we can provide an estimate $\tilde{f} = (1 \pm \eps)f_i$ of the outputted frequency $f_i$ with probability $1-\delta_2$ when $p<2$. To achieve this, we use our more general analysis of the contribution of the level sets $I_k$ to $\|z\|_2$, and give concentration bounds on the tail error when the top $\eps^{-p}$ items are removed. When $p=2$, for similar reasons as described in the sampling algorithm, we require another $O(\log n)$ factor in the space complexity to obtain a $(1 \pm \eps)$ estimate.  Finally, we demonstrate an $\Omega(\eps^{-p}\log n \log \delta_2^{-1})$ lower bound for this problem, which is nearly tight when $p<2$. To do so, we adapt a communication problem introduced in \cite{jayram2013optimal}, known as \textit{Augmented-Indexing on Large Domains}. We weaken the problem so that it need only succeed with constant probability, and then show that the same lower bound still holds. Using a reduction to this problem, we show that our lower bound for $L_p$ samplers holds even if the output index is from a distribution with constant \textit{additive} error from the true $L_p$ distribution $|f_i|^p/\|f\|_p^p$. 
 
	\section{Preliminaries}
 For $a,b,\eps\in \R$, we write $a = b\pm \eps$ to denote the containment $a \in [b - \eps, b + \eps]$. For positive integer $n$, we use $[n]$ to denote the set $\{1,2,\dots,n\}$, and $\tilde{O}(\cdot)$ notation to hide $\log(n)$ terms. For any vector $v \in \R^n$, we write $v_{(k)}$ to denote the $k$-th largest coordinate of $v$ in absolute value. In other words, $|v_{(1)}| \geq |v_{(2)}| \geq \dots \geq |v_{(n)}|$. For any $\gamma \in [n]$, we define $v_{\text{tail}(\gamma)}$ to be $v$ but with the top $\gamma$ coordinates (in absolute value) set equal to $0$. For any $i \in [n]$, we define $v_{-i}$ to be $v$ with the $i$-th coordinate set to $0$. We write $|v|$ to denote the entry-wise absolute value of $v$, so $|v|_j = |v_j|$ for all $j \in [n]$. All space bounds stated will be in bits. For our runtime complexity, we assume the unit cost RAM model, where a word of $O(\log(n))$-bits can be operated on in constant time, where $n$ is the dimension of the input streaming vector. Finally, we will use $\tilde{O}$ notation to hide poly$\log(n)$ factors; in other words $O(\log^c(n)) = \tilde{O}(1)$ for any constant $c$.
	
	Formally, a data stream is given by an underlying vector $f \in \R^n$, called the \emph{frequency vector}, which is initialized to $0^n$. The frequency vector then receives a stream of $m$ updates of the form $(i_t,\Delta_t) \in [n] \times \{-M,\dots,M\}$ for some $M > 0$ and $t \in [m]$. The update $(i,\Delta)$ causes the change $f_{i_t} \leftarrow f_{i_t} + \Delta_t$. For simplicity, we make the common assumption (\cite{braverman2016beating}) that $\log(mM) = O(\log(n))$, though our results generalize naturally to arbitrary $n,m$.  	
 In this paper, we will need Khintchine's and  McDiarmid's inequality

	\begin{fact}[ Khintchine inequality \cite{haagerup1981best}]\label{fact:khintchine}
		Let $x \in \R^n$ and $Q = \sum_{i=1}^n \varphi_i x_i$ for i.i.d. random variables $\varphi_i$ uniform on $\{1,-1\}$. Then
		$\pr{|Q| > t\|x\|_2} < 2e^{-  t^2/2}$.   
	\end{fact}

	\begin{fact}[McDiarmid's inequality \cite{mcdiarmid1989}]\label{fact:McDiarmid}
		Let $X_1,X_2,\dots,X_n$ be independent random variables, and let $\psi(x_1,\dots,x_n)$ by any function that satisfies 
		\[	\sup_{x_1,\dots,x_n,\hat{x}_i} \big|\psi(x_1,x_2,\dots,x_n) - \psi(x_1,\dots,x_{i-1}, \hat{x}_i,x_{i+1},\dots,x_n)	\big| \leq c_i \; \; \; \text{for } 1 \leq i \leq n	\]
		Then for any $\eps >0$, we have
					$\bpr{\Big|	\psi(X_1,\dots,X_n) - \bex{\psi(X_1,\dots,X_n)} \Big| \geq \eps} \leq 2\exp{\Big(\frac{-2\eps^2}{\sum_{i=1}^n c_i}	\Big)}.$
	\end{fact}

	Our analysis will use stability properties of Gaussian random variables. 
	\begin{definition}\label{def:stable}
		A distribution $\mathcal{D}_p$ is said to be $p$-stable if whenever $X_1,\dots,X_n \sim \mathcal{D}_p$ are drawn independently, we have\[	\sum_{i=1}^n a_i X_i = \|a\|_p X	\]
		for any fixed vector $a \in \R^n$, where $X \sim \mathcal{D}_p$ is again distributed as a $p$-stable. In particular, the Gaussian random variables $\mathcal{N}(0,1)$ are $p$-stable for $p=2$ (i.e., $\sum_i a_i g_i = \|a\|_g$, where $g,g_1,\dots,g_n$ are Gaussian). 
	\end{definition}

	\subsection{Count-Sketch and Count-Max} \label{sec:countmax}
	Our sampling algorithm will utilize a modification of the well-known data structure known as \textit{count-sketch} (see \cite{charikar2002finding} for further details). 	
	 We now introduce the description of count-sketch which we will use for the remainder of the paper. The count-sketch data structure is a table $A$ with $d$ rows and $k$ columns. When run on a stream $f \in \R^n$, for each row $i \in [d]$, count-sketch picks a uniform random mapping $h_i: [n] \to [k]$ and $g_i:[n] \to \{1,-1\}$. Generally, $h_i$ and $g_i$ need only be $4$-wise independent hash functions, but in this paper we will use fully-independent hash functions (and later relax this condition when derandomizing). 	
	 Whenever an update $\Delta$ to item $v \in [n]$ occurs, count-sketch performs the following updates:
	\[	A_{i,h_i(v)} \leftarrow A_{i,h_i(v)} + \Delta g_i(v) \; \; \text{ for } i = 1,2,\dots,d \]
	Note that while we will not implement the $h_i$'s as explicit hash functions, and instead generate i.i.d. random variables $h_i(1),\dots,h_i(n)$, we will still use the terminology of hash functions. In other words, by \textit{hashing} the update $(v,\Delta)$ into the row $A_i$ of count-sketch, we mean that we are updating $A_{i,h_i(v)}$ by $\Delta g_i(v)$. By hashing the coordinate $f_v$ into $A$, we mean updating $A_{i,h_i(v)} $ by $g_i(v)f_v$ for each $i=1,2,\dots,d$. 	
	Using this terminology, each row of count-sketch corresponds to randomly hashing the indices in $[n]$ into $k$ buckets, and then each bucket in the row is a sum of the frequencies $f_i$ of the items which hashed to it multiplied by random $\pm 1$ signs.
	In general, count-sketch is used to obtain an estimate vector $y \in \R^n$ such that $\|y - f\|_\infty$ is small. Here the estimate $y$ is given $y_j = \text{median}_{i \in [d]} A_{i,h_i(j)} (g_{i}(j))^{-1}$ for all $j \in [n]$. This vector $y$ satisfies the following guarantee. 
	\begin{theorem}\label{thm:count-sketch}
		If $d = \Theta(\log(1/\delta))$ and $k = 6/\eps^2$, then for a fixed $i \in [n]$ we have $|y_i - f_i| < \eps\|f_{\text{tail}(1/\eps^2)}\|_2$ with probability $1- \delta$. Moreover, if $d = \Theta(\log(n))$ and $c \geq 1$ is any constant, then we have $\|y - f\|_\infty < \eps\|f_{\text{tail}(1/\eps^2)}\|_2$ with probability $1 - n^{-c}$. Furthermore, if we instead set $y_j = \text{median}_{i \in [d]} |A_{i,h_i(j)}|$, then the same two bounds above hold replacing $f$ with $|f|$.  
	\end{theorem}
	In this work, however, we are only interested in determining the index of the heaviest item in $f$, that is $i^* = \arg \max_i |f_i|$. So we utilize a simpler estimation algorithm based on the count-sketch data structure that tests whether a fixed $j \in [n]$, if $j =  \arg \max_i |f_i|$. For analysis purposes, instead of having the $g_i$'s be random signs, we draw $g_i(v) \sim \mathcal{N}(0,1)$ as i.i.d. Gaussian variables. Then for a fixed $i$, set $\alpha_j = \big| \{ i \in [d] \; | \; |A_{i, h_i(j)}| = \max_{r \in [k]} |A_{i,r} |   \}\big|$, and we declare that $j = i^*$ to be the maximizer if  $\alpha_j > \frac{4}{5}d$. The algorithm computes $\alpha_j$ for all $j \in [n]$, and outputs the first index $j$ that satisfies $\alpha_j> \frac{4}{5}d$ (there will only be one with high probability).
		 To distinguish this modified querying protocol from the classic count-sketch, we refer to this algorithm as count-max. To refer to the data structure $A$ itself, we will use the terms count-sketch and count-max interchangeably. 
		 
		 We will prove our result for the guarantee of count-max in the presence of the following generalization. Before computing the values of $\alpha$ and reporting a maximizer as above, we will scale each bucket $A_{i,j}$ of count-max by a uniform random variable $\mu_{i,j} \sim \texttt{Uniform}(\frac{99}{100},\frac{101}{100})$. This generalization will be used for technical reasons in our analysis of Lemma \ref{lem:main}. Namely, we will need it to ensure that our failure threshold of our algorithm is randomized, which will allow us to handle small adversarial error.

	\begin{lemma}\label{lem:countmax}
	Let $c\geq 1$ be an arbitrarily large constant, set $d = \Theta(\log(n))$ and $k = 2$, and let $A$ be a $d \times k$ instance of count-max run on $f \in \R^n$ using fully independent hash functions $h_i$ and Gaussian random variables $g_i \sim \mathcal{N}(0,1)$. Then then with probability $1-n^{-c}$ the following holds: for every $i \in [n]$, if $|f_i| > 20 \|f_{-i}\|_2$ then count-max declares $i$ to be the maximum, and if $|f_i| \leq \max_{j \in [n]\setminus \{i\}} |f_j|$, then count-max does not declare $i$ to be the maximum. Thus if count-max declares $|f_i|$ to be the largest coordinate of $f$, it will be correct with high probability. Moreover, this result still holds if each bucket $A_{i,j}$ is scaled by a $\mu_{i,j} \sim \texttt{Uniform}(\frac{99}{100},\frac{101}{100})$ before reporting. 


	\end{lemma}
	\begin{proof}
		
			First suppose $|f_i|^2 > 20 \|f_{-i}\|_2^2$, and consider a fixed row $j$ of $A$. WLOG $i$ hashes to $A_{j,1}$, thus $A_{j,1} = \mu_{j,1}\left(g_j(i)f_j + \sum_{t\: : \: h_j(t) = 1} g_j(t)f_t\right)$ and $A_{j,2} = \mu_{j,2}\left(   \sum_{t\: : \: h_j(t) = 2}g_j(t)f_t\right)$. By $2$-stability (Definition \ref{def:stable}),
			the probability that $|A_{j,2}| > |A_{j,1}|$ is less than probability that one $\mathcal{N}(0,1)$ Gaussian is $19$ times larger than another, which can be bounded by $15/100$ by direct computation. Thus $i$ hashes into the max bucket in a row of $A$ with probability at least $85/100$, so by Chernoff bounds, taking $d  = \Omega(c\log(n))$, with probability $1-n^{-2c}$ we have that $f_i$ is in the largest bucket at least a $4/5$ fraction of the time, which completes the first claim.
		
		
		Now suppose $i$ is not a unique max, and let $i^*$ be such that $|f_{i^*}|$ is maximal. Then conditioned on $i,i^*$ not hashing to the same bucket, the probability that $f_i$ hashes to a larger bucket than $f_{i^*}$ is at most $1/2$. To see this, note that conditioned on this, one bucket is distributed as $g_j(i^*) f_{i^*} +G$ and the other as  $g_j(i) f_{i} +G'$, where $G,G'$ are identically distributed random variables. 
		Thus the probability that $f_i$ is the in maximal bucket is at most $3/4$, and so by Chernoff bounds $f_i$ will hash to strictly less than $(4d/5)$ of the maximal buckets with probability $1-n^{-2c}$.  Union bounding over all $j \in [n]$ gives the desired result.
	\end{proof}
\begin{corollary}\label{cor:countmax}
	In the setting of Lemma \ref{lem:countmax}, with probability $1-O(n^{-c})$, count-max will never report an index $i \in [n]$ as being the maximum if $|f_i|< \frac{1}{100}\|f\|_2$. 
\end{corollary}
\begin{proof}
	 Suppose  $|f_i|< \frac{1}{100}\|f\|_2$, and in a given row WLOG $i$ hashes to $A_{j,1}$. Then we have $A_{j,1} = g_j(i)f_i + g^1\|f^1\|_2$ and $A_{j,2} = g^2\|f^2\|_2$, where $f^k$ is $f$ restricted to the coordinates that hash to bucket $k$, and $g^1,g^2 \sim \mathcal{N}(0,1)$. Since $f^1,f^2$ are i.i.d., with probability $1/2$ we have $\|f^2\|_2 > \|f^1\|_2$. Conditioned on this, we have $\|f^2\|_2 >\|f\|_2/\sqrt{2} > 70|f_i|$. So conditioned on $\|f^2\|_2 > \|f^1\|_2$, we have $|A_{j,1}| < |A_{j,2}|$ whenever one Gaussian is $(71/70)$ times larger than another in magnitude, which occurs with probability greater than $1/2 - 1/25$. So $i$ hashes into the max bucket with probability at most $79/100$, and thus by Chernoff bounds, taking $c$ sufficiently large and union bounding over all $i \in [n]$, $i$ will hash into the max bucket at most a $795/1000 < 4/5$ fraction of the time with probability $1-O(n^{-c})$, as needed.
\end{proof}

	\section{Exponential Order Statistics}\label{sec:expo}
	In this section, we discuss several useful properties of the order statistics of $n$ independent non-identically distributed exponential random variables. 
	Let $(t_1,\dots,t_n)$ be independent exponential random variables where $t_i$ has mean $1/\lambda_i$ (equivalently, $t_i$ has rate $\lambda_i$). Recall that $t_i$ is given by the cumulative distribution function $\pr{t_i < x} = 1-e^{-\lambda_i x}$. Our main $L_p$ sampling algorithm will require a careful analysis of the distribution of values $(t_1,\dots,t_n)$, which we will now describe. We begin by noting that constant factor scalings of an exponential variable result in another exponential variable.
	
	\begin{fact}[Scaling of exponentials]\label{fact:scale}
		Let $t$ be exponentially distributed with rate $\lambda$, and let $\alpha > 0$. Then $\alpha t$ is exponentially distributed with rate $\lambda/\alpha $ 
	\end{fact}
\begin{proof}
	The cdf of $\alpha t$ is given by $ \pr{t < x/\alpha} =1-e^{-\lambda x/\alpha}$, which is the cdf of an exponential with rate $\lambda/\alpha$.
\end{proof}
	
	\noindent
	We would now like to study the order statistics of the variables $(t_1,\dots,t_n)$, where $t_i$ has rate $\lambda_i$. To do so, we introduce the \textit{anti-rank vector} $(D(1),D(2),\dots,D(n))$, where for $k\in [n]$, $D(k)\in[n]$ is a random variable which gives the index of the $t$-th smallest exponential.
	\begin{definition}
		Let $(t_1,\dots,t_n)$ be independent exponentials. For $k=1,2,\dots,n$, we define the \textit{$k$-th anti-rank} $D(k)\in [n]$ of $(t_1,\dots,t_n)$ to be the values $D(k)$ such that $t_{D(1)} \leq t_{D(2)} \leq \dots \leq t_{D(n)}$. 
	\end{definition}\noindent
	Using the structure of the anti-rank vector, it has been observed \cite{nagaraja2006order} that there is a simple form for describing the distribution of $t_{D(k)}$ as a function of $(\lambda_1,\dots,\lambda_n)$ and the anti-rank vector.
	
	\begin{fact}[\cite{nagaraja2006order}]\label{fact:order}
		Let $(t_1,\dots,t_n)$ be independently distributed exponentials, where $t_i$ has rate $\lambda_i> 0$. Then for any $k=1,2,\dots,n$, we have 
		\[	t_{D(k)} = \sum_{i=1}^k \frac{E_i}{\sum_{j=i}^n \lambda_{D(j)} }	\]
		Where the $E_1,E_2,\dots,E_n$'s are i.i.d. exponential variables with mean $1$, and are independent of the anti-rank vector $(D(1),D(2),\dots,D(n))$. 
	\end{fact}
	\begin{fact}[\cite{nagaraja2006order}]\label{fact:min}
		For any $i=1,2,\dots,n$, we have
		\[	\pr{D(1) = i} = \frac{\lambda_i}{\sum_{j=1}^n \lambda_j}	\]
	\end{fact}

	We now describe how these properties will be useful to our sampler. 	
	Let $f \in \mathbb{R}^n$ be any vector presented in a general turnstile stream. We can generate i.i.d. exponentials $(t_1,\dots,t_n)$, each with rate $1$, and construct the random variable $z_i = f_i/t_i^{1/p}$, which can be obtained in a stream by scaling updates to $f_i$ by $1/t_i^{1/p}$ as they arrive. 
	By Fact \ref{fact:scale}, the variable $|z_i|^{-p} = t_i/|f_i|^p$ is exponentially distributed with rate $\lambda_i = |f_i|^p$. Now let $(D(1),\dots,D(n))$ be the anti-rank vector of the exponentials $(t_1/|f_n|^p, \dots, t_n/|f_n|^{1/p})$. By Fact \ref{fact:min}, we have $\pr{D(1) = i} = \pr{i = \arg \min\{|z_1|^{-p},\dots,|z_n|^{-p}	\}} =  \pr{i = \arg \max\{|z_1|,\dots,|z_n|	\} } = \frac{\lambda_i}{\sum_j \lambda_j} = \frac{|f_i|^p}{\|f\|_p^p}$. In other words, the probability that $|z_i| = \arg \max_j \{|z_j|\}$ is precisely $|f_i|^p / \|f\|_p^p$, so for a perfect $L_p$ sampler it suffices to return $i\in [n]$ with $|z_i|$ maximum.
		Now note $|z_{D(1)}|  \geq| z_{D(2)} |\geq \dots \geq |z_{D(n)}|$, and in this scenario the statement of Fact \ref{fact:order} becomes:
	\[	z_{D(k)}= \big(\sum_{i=1}^k \frac{E_i}{\sum_{j=i}^N \lambda_{D(j)} }\big)^{-1/p} =   \big(\sum_{i=1}^k \frac{E_i}{ \sum_{j=i}^{N} F_{D(j)}^p }\big)^{-1/p} 	\]
	Where $E_i$'s are i.i.d. exponential random variables with mean $1$, and are independent of the anti-rank vector $(D(1),\dots,D(n))$. We call the exponentials $E_i$ the \textit{hidden exponentials}, as they do not appear in the actual execution of the algorithm, and will be needed for analysis purposes only.

	\section{The Sampling Algorithm}\label{sec:main}
	 We now provide intuition for the workings of our main sampling algorithm. Our algorithm scales the input stream by inverse exponentials to obtain a new vector $z$. We have seen in the prior section that we can write the order statistics $z_{D(k)}$ as a function of the anti-rank vector $D$, where $D(k)$ gives the index of the $k$-th largest coordinate in $z$, and the hidden exponentials $E_i$, which describe the ``scale" of the order statistics. Importantly, the hidden exponentials are independent of the anti-ranks. We would like to determine the index $i$ for which $D(1) = i$, however this may not always be possible. This is the case when the largest element $|z_{D(1)}|$ is not sufficently larger than the remainig $L_2$ mass $\sum_{j > 1}\left( |z_{D(j)}|^2\right)^{1/2}$. In such a case, count-max will not declare any index to be the largest, and we would therefore like to output \texttt{FAIL}. Note that this event is more likely when there is another element $|z_{D(2)}|$ which is very close to $|z_{D(1)}|$ in size, as whenever the two elements do not collide in count-max, it is less likely that $|z_{D(1)}|$ will be in the max bucket.
	 
	 Now consider the trivial situation where $f_1 = f_2 = \dots = f_n$. Here the variables $z_{D(k)}$ have no dependence at all on the anti-rank vector $D$. In this case, the condition of failing is independent of $D(1)$, so we can safely fail whenever we cannot determine the maximum index. On the other hand, if the values $|f_i|$ vary wildly, the variables $z_{D(k)}$ will depend highly on the anti-ranks. In fact, if there exists $f_i$ with $|f_i|^p \geq \eps \|f\|_p^p$, then the probability that $|z_{D(1)}| - |z_{D(2)}|$ is above a certain threshold can change by a $(1 \pm \eps)$ factor conditioned on $D(1) = i$, as opposed to $D(1) = j$ for a smaller $|f_j|$. Given this, the probability that we fail can change by a multiplicative $(1 \pm \eps)$ conditioned on $D(1) = i$ as opposed to $D(1) = j$.  
	 In this case, we cannot output \texttt{FAIL} when count-max does not report a maximizer, lest we suffer a $(1 \pm \eps)$ error in outputting an index with the correct probability.
	 
	  To handle this, we must remove the heavy items from the stream to weaken the dependence of the values $z_{D(k)}$ on the anti-ranks, which we carry out by duplication of coordinates. For the purposes of efficiency, we carry out the duplication via a rounding scheme which will allow us to generate and quickly hash updates into our data-structures (Section \ref{sec:runtime}). We will show that, conditioned on the fixed values of the $E_i$'s, the variables $z_{D(k)}$ are highly concentrated, and therefore nearly independent of the anti-ranks ($z_{D(k)}$ depends only on $k$ and not $D(k)$). By randomizing the failure threshold to be anti-concentrated, the small adversarial dependence of $z_{D(k)}$ on $D(k)$ cannot non-trivially affect the conditional probabilities of failure, leading to small relative error in the resulting output distribution.

\paragraph{The $L_p$ Sampler.}
	We now describe our sampling algorithm, as shown in Figure \ref{fig:sampler}.  Let $f \in \R^n$ be the input vector of the stream. As the stream arrives, we duplicate updates to each coordinate $f_i$ a total of $n^{c-1}$ times to obtain a new vector $F \in \R^{n^c}$. More precisely, for $i \in [n]$ we set $i_j = (i-1)n^{c-1} + j$ for $j=1,2,\dots, n^{c-1}$, and then we will have $F_{i_j} = f_i$ for all $i \in [n]$ and $j \in [n^{c-1}]$.  We then call $F_{i_j}$ a duplicate of $f_i$. Whenever we use $i_j$ as a subscript in this way it will refer to a duplicate of $i$, whereas a single subscript $i$ will be used both to index into $[n]$ and $[n^c]$. Note that this duplication has the effect that $|F_i|^p \leq n^{-c+1} \|F\|_p^p$ for all $p > 0$ and $i \in [n^c]$.

	We then generate $i.i.d.$ exponential rate $1$ random variables $(t_1,\dots,t_n)$, and
	 define the vector $z \in \R^{n^c}$ by $z_i = F_i/t_i^{1/p}$. As shown in Section \ref{sec:expo}, we have $\pr{i_j = \arg \max_{i',j'} \{|z_{i_{j'}'}|\}} = |F_{i_j}|^p / \|F\|_p^p$. Since $\sum_{j \in [n^{c-1}]} |F_{i_j}|^p / \|F\|_p^p = |f_i|^p / \|f\|_p^p$, it will therefore suffice to find $i_j \in [n^c]$ for which $i_j = \arg \max_{i',j'} \{|z_{i_{j'}'}|\}$, and return the index $i \in [n]$. The assumption that the $t_i$'s are i.i.d. will later be relaxed in Section \ref{sec:runtime} while derandomizing the algorithm. In Section \ref{sec:runtime}, we also demonstrate that all relevant continuous distributions will be made discrete without effecting the perfect sampling guarantee. 
	   	 
	 Now fix any sufficiently large constant $c$, and fix $\nu > n^{-c}$. To speed up the update time, instead of explicitly scaling $F_i$ by $1/t_i^{1/p}$ to construct the stream $z$, our algorithm instead scales $F_i$ by $\texttt{rnd}_\nu(1/t_i^{1/p})$, where $\texttt{rnd}_\nu(x)$ rounds $x > 0$ down to the nearest value in $\{\dots, (1 + \nu)^{-1}, 1, (1+\nu), (1+\nu)^2,\dots \}$. In other words, $\texttt{rnd}_{\nu}(x)$ rounds $x$ down to the nearest power of $(1 + \nu)^j$ (for $j \in \Z$). This results in a separate stream $\zeta \in \R^{n^c}$ where $\zeta_i = F_i \cdot \texttt{rnd}_\nu(1/t_i^{1/p})$. Note $\zeta_i = (1 \pm O(\nu)) z_i$ for all $i \in [n^c]$. Importantly, note that this rounding is order preserving. Thus, if $\zeta$ has a unique largest coordinate $|\zeta_{i^*}|$, then $|z_{i^*}|$ will be the unique largest coordinate of $z$. 
	
			\begin{figure}
		\fbox{\parbox{\textwidth}{ \texttt{$L_p$ Sampler}
				\begin{enumerate}[topsep=0pt,itemsep=-1ex,partopsep=1ex,parsep=1ex] 
				\item Set $d = \Theta(\log(n))$,  instantiate a $d \times 2$ count-max table $A$, and set $\mu_{i,j} \sim \texttt{Uniform}[\frac{99}{100},\frac{101}{100}]$ for each $(i,j) \in [d] \times [2]$.  
					\item Duplicate updates to $f$ to obtain the vector $F \in \mathbb{R}^{n^c}$ so that $f_i = F_{i_j}$ for all $i \in [n]$ and $j=1,2,\dots,n^{c-1}$, for some fixed constant $c$.
					\item Choose i.i.d. exponential random variables $t = (t_1,t_2,\dots,t_{n^c})$, and construct the stream $\zeta_i = F_i\cdot \texttt{rnd}_\nu (1/t_i^{1/p})$.
					\item Run $A$ on the stream $\zeta$. Upon the end of the stream, set $A_{i,j}  \leftarrow \mu_{i,j} A_{i,j}$ for all $(i,j) \in [d] \times [2]$.
					
					\item If count-max declares that an index $i_j \in [n^c]$ is the max  for some $j \in [n^{c-1}]$ based on the data structure $A$, then output $i \in [n]$. If $A$ does not declare any index to be the max, output \texttt{FAIL}.
				\end{enumerate}
		}}\caption{Our main $L_p$ Sampling algorithm} \label{fig:sampler}
	\end{figure}

	Having constructed the transformed stream $\zeta$, we then run a $d \times 2$ instance $A \in \R^{d \times 2}$ of count-max (from Section \ref{sec:countmax}), with $d=\Theta(\log(n))$, on $\zeta$. At the end of the stream, we scale each bucket $A_{i,j}$ by a uniform random variable $\mu_{i,j}$ from the interval $[\frac{99}{100}, \frac{101}{100}]$. This step is ensures that the failure threshold is randomized, so that a small adversarial error can only effect the output of the algorithm with extremely low probability (see Lemma \ref{lem:main}). 
	Now recall that count-max will either declare an index $i_j \in [n^c]$ as being the maximum, or report nothing. If an index $i_j$ is returned, where $i_j$ is the $j$-th copy of index $i \in [n]$, then our algorithm outputs the index $i$. If count-max does not report an index, we return \texttt{FAIL}. Let $i^* = \arg \max_i |\zeta_{i}| = D(1)$ (where $D(1)$ is the first anti-rank as in Section \ref{sec:expo}). By the guarantee of Lemma \ref{lem:countmax}, we know that if $|\zeta_{i^*}| \geq 20 \|\zeta_{-i^*}\|_2$, then with probability $1-n^{-c}$ count-max will return the index $i^* \in [n^c]$. Moreover, with the same probability, count-max will never return an index which is not the unique maximizer.  To prove correctness, therefore, it suffices to analyze the conditional probability of failure given $D(1) = i$.  Let $N = |\{i \in [n^c] \: | \: F_i \neq 0 \}|$ ($N$ is the support size of $F$). We can assume that $N \neq 0$ (to check this one could run, for instance, the $O(\log^2(n))$-bit support sampler of \cite{Jowhari:2011}). Note that $n^{c-1} \leq N \leq n^c$. The following fact is straightforward.

	\begin{fact}\label{fact:avgweight}
		For $p \in (0,2]$, suppose that we choose the constant $c$ such that $mM \leq n^{c/20}$, where note we have $|F_i| \leq mM$ for all $i \in [N]$. Then if $S \subset \{i \in [n^c] \: | \: F_i \neq 0 \}$ is any subset, then $\sum_{i \in S}|F_i|^p \geq \frac{|S|}{N} n^{-c/10} \|F\|_p^p$
	\end{fact}
	\begin{proof}
		We know that $|F_i|^p \leq (mM)^p \leq n^{c/10}$ using $p \leq 2$. Then each non-zero value $|F_i|^p$ is at most an $n^{-c/10}$ fraction of any other item $|F_j|^p$, and in particular of the average item weight. It follows that $|F_i|^p \geq n^{-c/10}\frac{\|F\|_p^p}{N}$ for all $i \in [N]$, which results in the stated fact.
	\end{proof}
	As in Section \ref{sec:expo}, we now use the anti-rank vector $D(k)$ to denote the index of the $k$-th largest value of $z_i$ in absolute value. In other words, $D(k)$ is the index such that $|z_{D(k)}|$ is the $k$-th largest value in the set $\{|z_{1}|,|z_{2}|,\dots,|z_{n^c}|\}$. Note that the $D(k)$'s are also the anti-ranks of the vector $\zeta$, since rounding $z$ into $\zeta$ preserves partial ordering. For the following lemma, it suffices to consider only the exponentials $t_i$ with $F_i \neq 0$, and we thus consider only values of $k$ between $1$ and $N$.  Thus $|z_{D(1)}| \geq |z_{D(2)}| \geq \dots \geq |z_{D(N)}|$. Moreover, we have that $|z_{D(k)}|^{-p} = \frac{t_{D(k)}}{|F_{D(k)}|^p}$ is the $k$-th smallest of all the $ \frac{t_i}{|F_i|^p}$'s, and by the results of Section \ref{sec:expo} can be written as $|z_{D(k)}|^{-p}  = \sum_{\tau=1}^k \frac{E_\tau}{\sum_{j=\tau}^N |F_{D(j)}|^p }$ where the $E_\tau$ are i.i.d. exponentials and independent of the anti-rank vector $D$. We will make use of this in the following lemma.
	
	\begin{lemma}\label{lem:Zconcentration}
		For every $1 \leq k < N - n^{9c/10}$, we have 
		\[	|z_{D(k)}| = \Big[(1 \pm O(n^{-c/10}))\sum_{\tau=1}^{k} \frac{E_\tau}{\ex{\sum_{j=\tau}^{N} |F_{D(j)}|^p }} 	\Big]^{-1/p}	\]
		with probability $1 - O(e^{-n^{c/3}})$.
	\end{lemma}
	\begin{proof}
		Let $\tau <N - n^{9c/10}$. We can write $\sum_{j=\tau}^N |F_{D(j)}|^p$ as a deterministic function $\psi(t_1,\dots, t_N)$ of the random scaling exponentials $t_1,\dots,t_N$ corresponding to $F_i \neq 0$. We first argue that 
		\[|\psi(t_1,\dots,t_N) - \psi(t_1\dots,t_{i-1},t_i',t_{i+1}, \dots, t_N)| < 2 \max_j \{F_j^p\} < 2n^{-c+1}\|F\|_p^p\]
		This can be seen from the fact that changing a value of $t_i$ can only have the effect of adding (or removing) $|F_i|^p$ to the sum $\sum_{j=\tau}^N |F_{D(j)}|^p$ and removing (or adding) a different $|F_{l}|$ from the sum. The resulting change in the sum is at most $2 \max_j \{|F_j|^p\}$, which is at most $2n^{-c+1}\|F\|_p^p$ by duplication. Set $T = N - \tau + 1$. Since the $t_i$'s are independent, we apply McDiarmid's inequality (Fact \ref{fact:McDiarmid}) to obtain
		\[\bpr{|	\sum_{j=\tau}^N |F_{D(j)}|^p - \ex{\sum_{j=\tau}^N |F_{D(j)}|^p}| > \eps Tn^{-c} \|F\|_p^p} \leq 2\exp\big(\frac{-2\eps^2 T^2 n^{-2c} \|F\|_p^{2p}] }{n^c (2n^{-c+1}\|F\|_p^p)^2 }	\big)		\]
		\[	\leq 2 \exp \big(-\frac{1}{2}\eps^2 T^2 n^{-c-2}\big)	\]
		Setting $\eps = \Theta( n^{-c/5})$ and using $T > n^{9c/10}$, this is at most $2 \exp(-\frac{1}{2}n^{2c/5 -2}	)$. To show concentration up to a $(1 \pm O(n^{-c/10}))$ factor, it remains to show that $\ex{\sum_{j=\tau}^N |F_{D(j)}|^p} = \Omega(Tn^{-11c/10 }\|F\|_p^p)$. This follows from the Fact \ref{fact:avgweight}, which gives  $\sum_{j=0}^T |F_{D(-j)}|^p \geq n^{-c/10}(T n^{-c}  \|F\|_p^p)$ deterministically.
		Now recall that  $|z_{D(k)}| = [\sum_{\tau=1}^{k} \frac{E_\tau}{\sum_{j=\tau}^{N} |F_{D(-j)}|^p } 	]^{-1/p}$.
		We have just shown that $\sum_{j=\tau}^N |F_{D(j)}|^p = (1 \pm O( n^{-c/10}))\ex{\sum_{j=\tau}^N |F_{D(j)}|^p}$, so we can union bound over all $\tau = 1,2,\dots,N-n^{9c/10}$ to obtain  
		\[	|z_{D(k)}| = \Big[(1 \pm O(n^{-c/10}))\sum_{\tau=1}^{k} \frac{E_\tau}{\ex{\sum_{j=\tau}^{N} |F_{D(j)}|^p }} 	\Big]^{-1/p}	\]
	for all $k\leq N - n^{9c/10}$ with probability $1 - O(n^c e^{- n^{2c/5 - 2}}) = 1 - O(e^{n^{c/3}})$.
		
	\end{proof}
	
	We use this result to show that our failure condition is nearly-independent of the value $D(1)$. Let $\mathcal{E}_1$ be the event that Lemma \ref{lem:Zconcentration} holds. Let $\neg \ttx{FAIL}$ be the event that the algorithm \texttt{$L_p$ Sampler} does not output \ttx{FAIL}.

\begin{lemma}\label{lem:main}
	For $p \in (0,2]$ a constant bounded away from $0$ and any $\nu \geq n^{-c/60}$, $\pr{\neg \ttx{FAIL} \: | \: D(1) } =  \pr{\neg \ttx{FAIL} } \pm \tilde{O}( \nu )$ for every possible $D(1) \in [N]$.
\end{lemma} 
\begin{proof}
By Lemma \ref{lem:Zconcentration}, conditioned on $\mathcal{E}_1$, for every $k < N-n^{9c/10}$ we have $|z_{D(k)}| =  U_{D(k)}^{1/p} (1 \pm O(n^{-c/10}))^{1/p} =  U_{D(k)}^{1/p} (1 \pm O(\frac{1}{p}n^{-c/10}))$ (using the identity $(1+x) \leq e^{x}$ and the Taylor expansion of $e^x$), where $U_{D(k)} = (\sum_{\tau=1}^{k} \frac{E_\tau}{\ex{\sum_{j=\tau}^{N} |F_{D(j)}|^p }})^{-1}$ is independent of the anti-rank vector $D$ (in fact, it is totally determined by $k$ and the hidden exponentials $E_i$). Then for $c$ sufficiently large, we have $|\zeta_{D(k)}| =U_{D(k)}^{1/p} (1 \pm O(\nu))$, and so for all $p \in (0,2]$ and $k < N - n^{9c/10}$
	
	\[	|\zeta_{D(k)}| = U_{D(k)}^{1/p} + U_{D(k)}^{1/p} V_{D(k)}\]
	Where $V_{D(k)}$ is some random variable that satisfies $|V_{D(k)}| = O(\nu)$.	
		Now consider a bucket $A_{i,j}$ for $(i,j) \in [d] \times [10]$.  Let $\sigma_k = \ttx{sign}(z_k) = \ttx{sign}(\zeta_k)$ for $k \in [n^c]$. Then we write $A_{i,j}/\mu_{i,j} = \sum_{k \in B_{ij}}\sigma_{D(k)} |\zeta_{D(k)}| g_i(D(k)) + \sum_{k \in S_{ij}}\sigma_{D(k)} |\zeta_{D(k)}| g_i(D(k))$ where $B_{ij} = \{k \leq N - n^{9c/10}   \:|\: h_i(D(k)) =j \}$ and  $S_{ij} = \{n^c\geq k > N - n^{9c/10} \:| \:h_i(D(k)) =j\}$. Here we define $\{D(N+1),\dots,D(n^c)\}$ to be the set of indices $i$ with $F_i=0$ (in any ordering, as they contribute nothing to the sum). Also recall that $g_i(D(k)) \sim \mathcal{N}(0,1)$ is the i.i.d. Gaussian coefficent associated to item $D(k)$ in row $i$ of $A$. So
		\[A_{i,j}/\mu_{i,j} =\sum_{k \in B_{ij}} g_i(D(k)) \sigma_{D(k)} U_{D(k)}^{1/p} + \sum_{k \in B_{ij}} g_i(D(k))\sigma_{D(k)} U_{D(k)}^{1/p} V_{D(k)} +  \sum_{k \in S_{ij}} g_i(D(k)) \zeta_{D(k)}  \]
		Importantly, observe that since the variables $h_i(D(k))$ are fully independent, the sets $B_{i,j},S_{i,j}$ are independent of the anti-rank vector $D$. In other words, the values $h_i(D(k))$ are independent of the values $D(k)$ (and of the entire anti-rank vector), since $\{h_i(1),\dots,h_i(n^c)\} = \{h_i(D(1)),\dots, h_i(D(n^c))\}$ are i.i.d. Note that this would not necessarily be the case if $\{h_i(1),\dots,h_i(n^c)\}$ were only $\ell$-wise independent for some $\ell = o(n^c)$. So we can condition on a fixed set of values  $\{h_i(D(1)),\dots, h_i(D(n^c))\}$ now, which fixes the sets $B_{i,j},S_{i,j}$. Now let $U_{i,j}^* = |\sum_{k \in B_{ij}} g_i(D(k))\sigma_{D(k)}U_{D(k)}^{1/p}|$. 
		
		\begin{claim}\label{claim:khine}
			For all $i,j \in [d] \times [2]$ and $p \in (0,2]$, we have $$\left|\sum_{k \in B_{ij}} g_i(D(k))\sigma_{D(k)}U_{D(k)}^{1/p} V_{D(k)} + \sum_{k \in S_{ij}} g_i(D(k)) \zeta_{D(k)}\right| = O\left(\nu (|A_{i,1}|+ |A_{i,2}| ) \right)$$ with probability $1-O(\log(n) n^{-c/60})$.
		\end{claim}
		\begin{proof}
			By the $2$-stability of Gaussians (Definition \ref{def:stable}), we have $|\sum_{k \in S_{ij}} g_i(D(k)) \zeta_{D(k)}| = O(\sqrt{\log(n)}\allowbreak (\sum_{k \in S_{i,j}} \allowbreak (2z_{D(k)})^2)^{1/2})$  with probability $1 - n^{-c}$.  This is a sum over a subset of the $n^{9c/10}$ smallest items $|z_i|$, and thus $\sum_{k \in S_{i,j}} z_{D(k)}^2 < \frac{n^{9c/10}}{N} \|z\|_2^2$, giving $|\sum_{k \in S_{ij}} g_i(D(k)) \zeta_{D(k)}| =  O(\sqrt{\log(n)} n^{-c/30}\|z\|_2)$. Now WLOG $A_{i,1}$ is such that $\sum_{k \in B_{i,1} \cup S_{i,1}} \zeta_{{D(k)}}^2 > \frac{1}{2}\|\zeta\|_2^2$. Then $|A_{i,1}| \geq |g| \|z\|_2^2/3$ where $g \sim \mathcal{N}(0,1)$. Via the pdf of a Gaussian, it follows with probability $1-O(n^{-c/60})$ that $|A_{i,1}| > n^{-c/60} \|z\|_2^2 = \Omega((n^{c/60}/\sqrt{\log(n)}) |\sum_{k \in S_{ij}} g_i(D(k)) \zeta_{D(k)}|)$. Scaling $\nu$ by a $\log(n)$ factor gives. $ |\sum_{k \in S_{ij}} \allowbreak g_i(D(k)) \zeta_{D(k)}| = O( \nu |A_{i,1}|)$. 			
			 Next, using that $|V_{{D(k)}}| = O(\nu)$, we have $|\sum_{k \in B_{ij}}\allowbreak  g_i(D(k))\sigma_{D(k)} U_{D(k)}^{1/p} V_{D(k)}| = O(\nu) |\sum_{k \in B_{ij}}\allowbreak g_i(D(k))\sigma_{D(k)} U_{D(k)}^{1/p}| = O(\nu U_{i,j}^*)$. Combined with the prior paragraph, we have $U_{i,j}^* = O((|A_{i,1}| + |A_{i,2}|))$ as needed. 	
			Note there are only $O(\log(n))$ terms $i,j$ to union bound over, and from which the claim follows.
		\end{proof}

		Call the event where the Claim \ref{claim:khine} holds $\mathcal{E}_2$. Conditioned on $\mathcal{E}_2$, we can decompose $|A_{i,j}|/\mu_{i,j}$ for all $i,j$ into $U_{i,j}^* + \mathcal{V}_{ij}$ where $\mathcal{V}_{ij}$ is some random variable satisfying $|\mathcal{V}_{ij}| = O(\nu (|A_{i,1}| + |A_{i,2}|))$ and $U_{i,j}^*$ is independent of the anti-rank vector $D$ (it depends only on the hidden exponentials $E_k$, and the uniformly random gaussians $g_i(D(k))$).		
		Now fix any realization of the count-max randomness, Let $E = (E_1,\dots,E_N)$ be the hidden exponential vector, $\mu = \{\mu_{i,1}, \mu_{i,2}\}_{i \in [d]}$, $D = (D(1),D(2),\dots,D(N))$, and observe:		
			 \begin{equation*}
		\begin{split}
	\bpr{ \neg \texttt{FAIL} \; | \; D(1) }& = 	\sum_{E,\mu }  \bpr{ \neg \texttt{FAIL} \; | \; D(1) , E, \mu} \bpr{E, \mu}  \\		
		\end{split}
		\end{equation*}
		Here we have used the fact that $E,\mu$ are independent of the anti-ranks $D$. Thus, it will suffice to bound the probability of obtaining $E,\mu$ such that the event of failure can be determined by the realization of $D$. So consider any row $i$, and consider the event $\mathcal{Q}_{i}$ that $|\mu_{i,1} U^*_{i,1} - \mu_{i,2} U^*_{i,2}| < 2(|\mathcal{V}^*_{i,1}| + |\mathcal{V}^*_{i,2}|) = O(\nu (|A_{i,1}| + |A_{i,2}|)$ (where here we have conditioned on the high probability event $\mathcal{E}_2$). WLOG, $U^*_{i,1} \geq U^*_{i,2}$, giving $U^*_{i,1} = \Theta( (|A_{i,1}| + |A_{i,2}|)$. Since the $\mu_{i,j}$'s are uniform, $\pr{\mathcal{Q}_i} = O(\nu  (|A_{i,1}| + |A_{i,2}|)/ U^*_{i,1}) = O(\nu)$, and by a union bound $\pr{\cup_{i \in [d]}\mathcal{Q}_i} =  O(\log(n) \nu)$. Thus conditioned on $ \mathcal{E}_1 \cap \mathcal{E_2}$ and $\neg(\cup_{i \in [d]}\mathcal{Q}_i)$, the event of failure is completely determined by the values $E,\mu$, and in particular is independent of the anti-rank vector $D$. Thus
	
		\[ \bpr{\neg \texttt{FAIL} \; | D(1) , E, \mu , \neg (\cup_{i \in [d]}\mathcal{Q}_i),  \mathcal{E}_1 \cap \mathcal{E_2}} = \bpr{\neg \texttt{FAIL} \; |E, \mu , \neg (\cup_{i \in [d]}\mathcal{Q}_i), \mathcal{E}_1 \cap \mathcal{E_2}}  \]
		
		So averaging over all $E,\mu$:
			\begin{equation*}
		\begin{split}
		\bpr{\neg \texttt{FAIL} \; | D(1) } & = \bpr{\neg \texttt{FAIL} \; | D(1) , \neg (\cup_{i \in [d]}\mathcal{Q}_i),  \mathcal{E}_1 \cap \mathcal{E_2}} + O(\log(n) \nu) \\
	 & = \bpr{\neg \texttt{FAIL} \; | \neg (\cup_{i \in [d]}\mathcal{Q}_i),  \mathcal{E}_1 \cap \mathcal{E_2}} + O(\log(n) \nu) \\	
	  & = \bpr{\neg \texttt{FAIL} } + O(\log(n) \nu) \\	
		\end{split}
		\end{equation*}
		As needed.

	\end{proof}

 In Lemma \ref{lem:main}, we demonstrated that the probability of failure can only change by an additive $\tilde{O}(\nu)$ term given that any one value of $i \in [N]$ achieved the maximum (i.e., $D(1) = i$). This property will translate into a $(1 \pm \tilde{O}(\nu))$-\textit{relative error} in our sampler, where the space complexity is independent of $\nu$. To complete the proof of correctness of our algorithm, we now need to bound the probability that we fail at all. To do so, we first prove the following fact about $\|z_{\text{tail(s)}}\|_2$, or the $L_2$ norm of $z$ with the top $s$ largest (in absolute value) elements removed.
	
	\begin{proposition}\label{prop:1}
		For any $s = 2^j \leq n^{c-2}$ for some $j \in \mathbb{N}$, we have $\sum_{i=4s}^{N} z_{D(i)}^2 = O( \|F\|_p^{2}/s^{2/p - 1})$ if $p \in (0,2)$ is a constant bounded below $2$, and $\sum_{i=4s}^{N} z_{(i)}^2 = O(\log(n)\|F\|_p^2)$ if $p=2$, with probability $1 - 3e^{-s}$.
	\end{proposition}
	\begin{proof}
	 Let $I_k = \{i \in [N] \; |\;  z_i \in (\frac{\|F\|_p}{2^{(k+1)/p}}, \frac{\|F\|_p}{2^{k/p}})\}$ for $k=0,1,\dots, p \log(\|F\|_p)$ (where we have $\log(\|F\|_p^p) = O(\log(n))$). Note that $\pr{i \in I_k} = \pr{t_i \in (\frac{2^k1F_i^p }{\|F\|_p^p}, \frac{2^{k+1}F_i^p }{\|F\|_p^p})} < \frac{ 2^kF_i^p}{\|F\|_p^p}$, where the inequality follows from the fact that the pdf $e^{-x}$ of the exponential distribution is upper bounded by $1$. Thus $\ex{|I_k|} < 2^k$, so for every $k \geq \log(s) = j$, we have  $\pr{|I_k| > 4(2^{k})} < e^{-s 2^{k-j}}$. By a union bound, the probability that $|I_k| > 4(2^k)$ for any $k \geq \log(S)$ is at most $e^{-s} \sum_{i=0}^{O(\log(n))} e^{ 2^{i}} \leq 2e^{-s}$. 		
		Now observe $\pr{z_i > \|F\|_p/s^{1/p}} < \frac{sF_i^p}{\|F\|_p^p}$, so $\ex{| \{i | z_i > \|F\|_1/s^{1/p}\}|} < s$, and again by Chernoff bounds the number of such $i$ with $z_i > \|F\|_1/s^{1/p}$ is at most $4s$ with probability $1 -e^{-s}$. Conditioning on this, $\sum_{i=4s}^{N} z_{(i)}^2$ does not include the weight of any of these items, so
		\[	\sum_{i=4s}^{N} z_{(i)}^2 \leq \sum_{k = \log(s) }^{O(\log(n))} |I_k|(\frac{\|F\|_p}{2^{k/p}} )^2 		\leq 4\sum_{k = 0 }^{O(\log(n))}  \frac{\|F\|_p^2}{2^{(\log(s) + k)(2/p - 1)}}   \]
		First, if $p < 2$, the above sum is geometric and converges to at most $4\frac{\|F\|_p^2}{1 - 2^{-2/p + 1}} \frac{1}{s^{2/p-1}} = O(\|F\|_p^2 / s^{2/p - 1})$ for $p$ a constant bounded below by $2$. If $p=2$ or is arbitrarily close to $2$, then each term is at most $\|F\|_p^2$, and the sum is upper bounded by $O(\log(n)\|F\|_p^2)$ as stated. Altogether, the probability of failure is at most $1 - 3e^{-s}$ by a union bound. 
	\end{proof}
	\begin{lemma}\label{lem:failbound}
		For $0 < p < 2$ a constant bounded away from $0$ and $2$, the probability that \texttt{$L_p$ Sampler} outputs \texttt{FAIL} is at most $1 - \Omega(1)$, and for $p=2$ is is $1- \Omega(1/\log(n))$. 
	\end{lemma}
	\begin{proof}
		By Proposition \ref{prop:1}, with probability $1-3e^{-4} > .9$ we have $\|z_{tail(16)}\|_2 = O(|F\|_p)$ for $p<2$, and $\|z_{tail(16)}\| = O(\sqrt{\log(n)}\|F\|_p)$ when $p=2$. Observe that for $t = 2,3,\dots,16$, we have $|z_{D(t)}| < \|F\|_p (\frac{2}{\sum_{\tau=1}^t E_\tau})^{1/p}$, and with probability $99/100$ we have $E_t > 1/100$, which implies that $|z_{D(t)}| = O(\|F\|_p)$ for all $t \in [16]$.  Conditioned on this, we have $\|z_{\text{tail}(2)}\|_2 < q\|F\|_p$ where $q$ is a constant when $p<2$, and $q = \Theta(\sqrt{\log(n)})$ when $p=2$. Now $|z_{D(1)}| = \frac{\|F\|_p}{E_1^{1/p}}$, and using the fact that the pdf exponential random variables around $0$ is bounded above by a constant, we will have $|z_{D(1)}| > 20\|z_{-D(1)}\|_2$ with probability $\Omega(1)$ when $p<2$, and probability $\Omega(\frac{1}{\log(n)})$ when $p=2$. Conditioned on this, by Lemma \ref{lem:countmax}, count-max will return the index $D(1)$ with probability $1-n^{-c}$, and thus the Sampling algorithm will not fail.
	\end{proof}

Putting together the results of this section, we obtain the correctness of our algorithm as stated in Theorem \ref{thm:main}. In Section \ref{sec:runtime}, we will show that the algorithm can be implemented to have $\tilde{O}(\nu)$ update and $\tilde{O(1)}$ query time, and that the entire algorithm can be derandomized to use $O(\log^2(n))$ bits of space for $p \in (0,2)$ and $O(\log^3(n))$ bits for $p=2$. 

	\begin{theorem}\label{thm:main}
		Given any constant $c \geq 2$, $\nu \geq n^{-c}$, and $0 < p \leq 2$, there is a one-pass $L_p$ sampler which returns an index $i \in [n]$ such that $\pr{i = j} = \frac{|f_j|^p}{\|f\|_p^p}(1\pm \nu) \pm n^{-c}$ for all $j \in [n]$, and which fails with probability $\delta > 0$. The space required is $O(\log^2(n) \log(1/\delta) (\log \log n)^2)$ bits for $p < 2$,  and $O(\log^3(n) \log(1/\delta))$ bits for $p = 2$. For $p<2$ and $\delta =1/\poly(n)$, the space is $O(\log^3(n))$-bits. The update time is $\tilde{O}(\nu^{-1})$, and the query time is $\tilde{O}(1)$. 
	\end{theorem}
	
	\begin{proof}
Conditioned on not failing, by Lemma \ref{lem:countmax}, with probability $1-n^{-c}$ we have that the output $i_j\in [n^c]$ of count-max will in fact be equal to $\arg \max_i \{|\zeta_i|\}$. Recall that $\zeta_i = (1 \pm O(\nu))z_i$ for all $i \in [n^c]$ (and this rounding of $z$ to $\zeta$ is order preserving). By Lemma \ref{lem:countmax} count-max only outputs a coordinate which is the \textit{unique} maximizer of $\zeta$. Now if there was \textit{unique} maximizer of $\zeta$, there must also be a unique maximizer in $z$, from which it follows that $i_j =\arg \max_i \{|z_i|\}$ . 

			Now Lemma \ref{lem:main} states for any $i_j \in [n^c]$ that $\pr{\neg \text{FAIL} \:| \: i_j = \arg \max_{i',j'} \{|z_{i_{j'}'}|\}} = \pr{\neg \text{FAIL}} \pm \tilde{O}(\nu) = q \pm \tilde{O}(\nu)$, where $q =  \pr{\neg \text{FAIL}} = \Omega(1)$ for $p<2$, and $q = \Omega(\frac{1}{\log(n)})$ for $p=2$, both of which follow from Lemma \ref{lem:failbound}, which does not depend on any of the randomness in the algorithm. 
		Since conditioned on not failing, the output $i_j$ of count-max satisfies $i_j= \arg \max_i \{|z_i|\}$, the probability we output $i_j \in [n^c]$ is $\pr{\neg \text{FAIL} \cap i_j= \arg \max \{|z_i|\}}$, so the probability our final algorithm outputs $i \in [n]$ is 
		 \[\sum_{j \in [n^{c-1}]}\allowbreak \pr{\neg \text{FAIL} \:| \: i_j = \arg \max_{i',j'}\{|z_{i_{j'}'}|\}}\pr{ i_j = \arg \max_{i',j'}\{|z_{i_{j'}'}|\}}  = \sum_{j \in [n^{c-1}]} \frac{|f_i|^p}{\|F\|_p^p} ( q\pm \tilde{O}(\nu))\]\[ =  \frac{|f_i|^p}{\|f\|_p^p}(q\pm \tilde{O}(\nu))\]
		 Note that we can scale the $c$ value used in the algorithm by a factor of $60$, so that the statement of Lemma $\ref{lem:main}$ holds for any $\nu \geq n^{-c}$.
		  The potential of the failure of the various high probability events that we conditioned on only adds another additive $O(n^{-c})$ term to the error. Thus, conditioned on an index $i$ being returned, we have $\pr{i = j} = \frac{|f_j|^p}{\|f\|_p^p}(1\pm \tilde{O}(\nu))) \pm n^{-c}$ for all $j \in [n]$, which is the desired result after scaling $\nu$ by a $\poly(\log(n))$ term.
	  Running the algorithm $O(\log(\delta^{-1}))$ times in parallel for $p<2$ and $O(\log(n)\log(\delta^{-1}))$ for $p=2$, it follows that at least one index will be returned with probability $1-\delta$.
	  
	  For the complexity, the update time of count-max data structure $A$ follows from the routine \texttt{Fast-Update} of Lemma \ref{lem:runtime}, and the query time follows from Lemma \ref{lem:query}. 	  
	   Theorem \ref{thm:derandom} shows that the entire algorithm can be derandomized to use a random seed with $O(\log^2(n)(\log \log(n))^2)$-bits, so to complete the claim it suffices to note that using $O(\log(n))$-bit precision as required by \texttt{Fast-Update} (Lemma \ref{lem:runtime}), it follows that our whole data structure $A$ can be stored with $O(\log^2(n))$ bits, which is dominated by the cost of storing the random seed. This gives the stated space after taking $O(\log(\delta^{-1}))$ parallel repetitions for $p<2$. For $p=2$, we only need a random seed of length $O(\log^3(n))$ for all $O(\log(n)\log(\delta^{-1}))$ repetitions by Corollary \ref{cor:derandom}, which gives $O(\log^3(n)\log(\delta^{-1}) + \log^3(n)) = O(\log^3(n) \log(1/\delta))$ bits of space for $p=2$ as stated. Similarly for the case of $p<2$ and $\delta = 1/\poly(n)$, the stated space follows from Corollary \ref{cor:derandom}.
\end{proof}

In particular, it follows that perfect $L_p$ samplers exist using $O(\log^2(n)\log(1/\delta) (\log \log n)^2)$ and  $O(\log^3(n)\ab \log(1/\delta))$  bits of space for $p < 2$ and $p =2$ respectively. 

	\begin{theorem}\label{thm:main2}
	Given $0 < p \leq 2$, for any constant $c \geq 2$ there is a \textit{perfect} $L_p$ sampler which returns an index $i \in [n]$ such that $\pr{i = j} = \frac{|f_j|^p}{\|F\|_p^p} \pm O(n^{-c})$ for all $j \in [n]$, and which fails with probability $\delta > 0$. The space required is $O(\log^2(n) \log(1/\delta)(\log \log n)^2)$ bits for $p < 2$, and $O(\log^3(n) \log(1/\delta))$ bits for $p = 2$.  For $p<2$ and $\delta =1/\poly(n)$, the space is $O(\log^3(n))$-bits. 
\end{theorem}

Finally, we note that the cause of having to pay an extra $(\log \log n)^2$ factor in the space complexity for $p<2$ is only due to the derandomization. Thus, in the random oracle model where the algorithm has assess to a $\poly(n)$-length random tape which does not count against its space requirement, the space is an optimal $O(\log^2(n) \log(1/\delta))$. We remark that the $\Omega(\log^2(n) \log(1/\delta))$ of \cite{kapralov2017optimal} lower bound also holds in the random oracle model.

\begin{corollary}
	For $p \in (0,2)$, in the random oracle model,  there is a \textit{perfect} $L_p$ sampler which  fails with probability $\delta > 0$ and uses $O(\log^2(n) \log(1/\delta))$ bits of space. 
\end{corollary}

\begin{remark}
 Note that for $p$ arbitrarily close to $2$, the bound on $\|z\|_2$ of Proposition \ref{prop:1} as used in Lemma \ref{lem:failbound} degrades, as the sum of the $L_2$ norms of the level sets is no longer geometric, and must be bounded by $O(\sqrt{\log(n)} \|F\|_2)$. In this case, the failure probability from Lemma \ref{lem:failbound} goes to $\Theta(\frac{1}{\log(n)})$, and so we must use the upper bound for $p=2$. Similarly, for $p$ arbitrarily close to $0$, the bound also degrades since the values $V_{D(k)}$ in Lemma \ref{lem:main} blow-up. For such non-constant $p$ arbitrarily close to $0$, we direct the reader to the $O(\log^2(n))$-bit perfect $L_0$ sampler of \cite{Jowhari:2011}. 
\end{remark}

	\section{Time and Space Complexity} \label{sec:runtime}
	In this section, we will show that our algorithm can be implemented with the desired space and time complexity. First, in Section \ref{subsec:update}, we show how \texttt{$L_p$ Sampler} can be implemented with the update procedure \texttt{Fast-Update} to result in $\tilde{O}(\nu^{-1})$ update time. Next, in Section \ref{subsec:derandom}, we show that the algorithm \texttt{$L_p$ Sampler} with \texttt{Fast-Update} can be derandomized to use a random seed of length $O(\log^2(n) (\log\log n)^2)$-bits, which will give the desired space complexity. Finally, in Section \ref{subsec:query}, we show how using an additional heavy-hitters data structure as a subroutine, we can obtain $\tilde{O}(1)$ update time as well. This additional data structure will not increase the space or update time complexity of the entire algorithm, and does not need to be derandomized.   
	
	
	\subsection{Optimizing the Update Time}\label{subsec:update}
	In this section we prove Theorem \ref{lem:runtime}. Our algorithm utilizes a single data structure run on the stream $\zeta$, which is count-max matrix $A \in \R^{d \times 2}$ where $d = \Theta(\log(n))$. We will introduce an update procedure \texttt{Fast-Update} which updates the data structure $A$ of \texttt{$L_p$ Sampler} in $\tilde{O}(\nu^{-1})$ time. We assume the unit cost RAM model of computation, where a word of length $O(\log(n))$-bits can be operated on in $O(1)$ time (note that replacing $O(1)$ with $\poly (\log(n))$ time here would not effect our results, as the additional cost would be hidden in the $\tilde{O}$).
	 Throughout this section, we will refer to the \textit{original algorithm} as the algorithm which implements \texttt{$L_p$ sampler} by individually generating each scaling exponential $t_i$ for $i \in [n^c]$, and hashing them individually into $A$ (na\"ively taking $n^c$ update time). Our procedure will utilize the following result about efficiently sampling binomial random variables which can be found in \cite{bringmann2014internal}.

\begin{proposition}\label{prop:fastBin}
	For any constant $c>0$, there is an algorithm that can draw a sample $X \sim \texttt{Bin}(n,1/2)$ in expected $O(1)$ time in the unit cost RAM model. Moreover, it can be sampled in time $\tilde{O}(1)$ with probability $1 - n^{-c}$. The space required is $O(\log(n))$-bits. 
\end{proposition}
\begin{proof}
	The proof of the running time bounds and correctness can be found in \cite{bringmann2014internal}. Since they do not analyze the space complexity of their routine, we do so here. Their algorithm is as follows. We can assume $n$ is even, otherwise we could sample \texttt{Bin}$(n,q) \sim \texttt{Bin}(n-1,q) + \texttt{Bin}(1,q)$, where the latter can be sampled in constant time (unit cost RAM model) and $O(\log(n))$-bits of space. The algorithm first computes $\Delta \in [\sqrt{n},\sqrt{n}+3]$, which can be done via any rough approximation of the function $\sqrt{x}$, and requires only $O(\log(n))$-bits. Define the block $\mathcal{B}_k = \{km, km+1, \dots, km + m-1\}$ for $k \in \Z$, and set 
	\[	f(i) = \frac{4}{2^{\max\{k,-k-1\}}m} \; \; \text{s.t. } i \in \mathcal{B}_k	\]
	\[	p(i) = 2^{-n}\binom{n}{n/2 + i} 	\]
 Note that given $i$, $f(i)$ can be computed in constant time and $O(\log(n))$ bits of space. The algorithm then performs the following loop:
 \begin{enumerate}
 	\item Sample $i$ via the normalized probability distribution $\bar{f} =f/16$. 
 	\item Return $n/2 + i$ with probability $p(i)/f(i)$
 	\item Else, reject $i$ and return to Step 1.
 \end{enumerate}
To compute the first step, the symmetry around $n/2$ of $f$ is utilized. We flip unbiased coins $C_1,C_2,\dots$ until we obtain $C_{t+1}$ which lands tails, and pick $i$ uniformly from block $\mathcal{B}_t$ or $\mathcal{B}_{-t}$ (where the choice is decided by a single coin flip). The procedure requires at most $O(\log(n))$-bits to store the index $t$. Next, to perform the second step, we obtain $2^{-L}$ additive error approximations $\tilde{q}$ of $q=(p(i)/f(i))$ for $L=1,2,\dots,$ which (using the fact that $0 \leq q\leq 1$) can be done by obtaining a $2^{-L}$-relative error approximation of $q$. Then we flip $L$ random bits to obtain a uniform $\tilde{R} \in [0,1]$, and check if $|\tilde{R} - \tilde{q}| > 2^{-L}$. If so, we can either accept or reject $i$ based on whether $\tilde{R} > \tilde{q} + 2^{-L}$ or not, otherwise we repeat with $L \leftarrow L+1$.

To obtain $\tilde{q}$, it suffices to obtain a $2^{-L-1}$ relative error approximation of the factorial function $x!$. To do so, the $2^{-L}$ approximation 
\[x! \approx (x+L)^{x+1/2}e^{-(x+L)}\big[	\sqrt{2\pi} + \sum_{k=1}^{L-1} \frac{c_k}{x+k}	\big] \]
is used, where $c_k = \frac{(-1)^{k-1}}{(k-1)!}(L-k)^{k-1/2}e^{L-k}$. This requires estimating the functions $e^x$, $\sqrt{x}$ and $\pi$, all of which, as well as each term in the sum, need only be estimated to $O(L)$-bits of accuracy (as demonstrated in \cite{bringmann2014internal}). Thus the entire procedure is completed in $O(L) = O(\log(n))$-bits of space ($L$ can never exceed $O(\log(n))$, as $q$ is specified with at most $O(\log(n))$ bits), which completes the proof.	
\end{proof}

We now utilize a straightforward reduction from the case of sampling from $\texttt{Bin}(n,q)$ for any $q \in [0,1]$ to sampling several times from $\texttt{Bin}(n',1/2)$ where $n' \leq n$. This reduction has been observed before \cite{farach2015exact}, however we will state it here to clearly demonstrate our desired space and time bounds. 

\begin{lemma} \label{lem:bin}
	For any constant $c>0$ and $q \in [0,1]$, there is an algorithm that can draw a sample $X \sim \texttt{Bin}(n,q)$ in expected $O(1)$ time in the unit cost RAM model. Moreover, it can be sampled in time $\tilde{O}(1)$ with probability $1 - n^{-c}$, and the space required is $O(\log(n))$-bits. 
\end{lemma} 
\begin{proof}
	The reduction is as follows (for a more detailed proof of correctness, see \cite{farach2015exact}). We sample \texttt{Bin}$(n,q)$ by determining how many of the $n$ trials were successful. This can be done by generating variables $u_1,\dots,u_n$ uniform on $[0,1]$, and determining how many are less than $q$. We do this without generating all the variables $u_i$ explicitly as follows.  First write $q$ in binary as $q = (0.q_1 q_2 ,\dots )_2$. Set $b \leftarrow 0$, $j \leftarrow 1$, $n_j \leftarrow n$ and sample $b_j \sim \texttt{Bin}(n_j,1/2)$. If $q_j =1$, then set $b = b+b_j$, as these corresponding $b_j$ trials $u_i$ with the first bit set to $0$ will all be successful trials given that $q_j = 1$. Then set $n_{j+1} \leftarrow n_j - b_j$ and repeat with $j \leftarrow j +1$. Otherwise, if $q_j = 0$, then we set $n_{j+1} \leftarrow n_j - (n_j - b_j) = b_j$, since this represents the fact that $(n_j-b_j)$ of the variables $u_i$ will be larger than $q$. With probability $1 - n^{-100c}$, we reach the point where $n_j = 0$ within $O(\log(n))$ iterations, and we return the value stored in $b$ at this point. By Proposition \ref{prop:fastBin}, each iteration requires $\tilde{O}(1)$ time, and thus the entire procedure is $\tilde{O}(1)$. For space, note that we need only store $q$ to its first $O(\log(n))$ bits, since the procedure terminates with high probability within $O(\log(n))$ iterations. Then the entire procedure requires $O(\log(n))$ bits, since each sample of $\texttt{Bin}(n_j,1/2)$ requires only $O(\log(n))$ space by Proposition \ref{prop:fastBin}. 
	
\end{proof}

\paragraph{The \texttt{Fast-Update} procedure. } We are now ready to describe the implementation of our update-time algorithm. Recall that our algorithm utilizes just a single data structure on the stream $\zeta$: the $d \times 2 $ count-max matrix $A$ (where $d = \Theta(\log(n))$).  Upon receiving an update $(i,\Delta)$ to a coordinate $f_i$ for $i \in [n]$, we proceed as follows. Our goal is to compute the set $\{ \texttt{rnd}_\nu(1/t_{i_1}^{1/p}), \texttt{rnd}_\nu(1/t_{i_2}^{1/p}), \dots, \texttt{rnd}_\nu(1/t_{i_{n^{c-1}}}^{1/p}) \}$, and update each row of $A$ accordingly in $\tilde{O}(\nu^{-1})$ time. Na\"ively, this could be done by computing each value individually, and then updating each row of $A$ accordingly, however this would require $O(n^{c-1})$ time. To avoid this, we exploit the fact that the support size of $\texttt{rnd}_\nu(x)$ for $1/\poly(n) \leq x \leq \poly(n)$ is $\tilde{O}(\nu^{-1})$, so it will suffice to determine how many variables $\texttt{rnd}_\nu(1/t_{i_j}^{1/p})$ are equal to each value in the support of $\texttt{rnd}_\nu(x)$. 

Our update procedure is then as follows. Let $I_j = (1 + \nu)^j$ for $j = -\Pi, -\Pi+1,\dots, \Pi-1,\Pi$ where $\Pi = O(\log(n)\nu^{-1})$.  We utilize the c.d.f. $\psi(x) = 1-e^{-x^{-p}}$ of the $1/p$-th power of the inverse exponential distribution $t^{-1/p}$ (here $t$ is exponentially distributed). Then beginning with $j= -\Pi,-\Pi+1,\dots,\Pi$ we compute the probability  $q_j = \psi(I_{j+1}) - \psi(I_j)$ that $\texttt{rnd}_\nu(1/t^{1/p}) = I_j$, and then compute the number of values $Q_j$ in $\{ \texttt{rnd}_\nu(1/t_{i_1}^{1/p}), \texttt{rnd}_\nu(1/t_{i_2}^{1/p}), \dots, \texttt{rnd}_\nu(1/t_{i_{n^{c-1}}}^{1/p}) \}$ which are equal to  $I_j$. 
 With probability $1 - n^{100c}$, we know that $1/\poly(n) \leq t_i \leq \poly(n)$ for all $i \in [N]$, and thus conditioned on this, we will have completely determined the values of the items in $\{ \texttt{rnd}_\nu(1/t_{i_1}^{1/p}), \texttt{rnd}_\nu(1/t_{i_2}^{1/p}), \dots, \texttt{rnd}_\nu(1/t_{i_{n^{c-1}}}^{1/p}) \}$ by looking at the number equal to $I_j$ for $j =-\Pi,\dots,\Pi$. 

Now we know that there are $Q_j$ updates which we need to hash into $A$ (along with i.i.d. Gaussian scalings), each with the same value $\Delta I_j$. This is done by the procedure \texttt{Fast-Update-CS} (Figure \ref{fig:update-CS}), which computes the number $b_{k,\theta}$ that hash to each bucket $A_{k,\theta}$ by drawing binomial random variables.
Once this is done, we know that the value of $A_{k,\theta}$ should be updated by the value $\sum_{t=1}^{b_{k,\theta}} g_t \Delta I_j$, where each $g_t \sim \mathcal{N}(0,1)$.   Na\"ively, computing the value  $\sum_{t=1}^{b_{k,\theta}} g_t \Delta I_j$ would involve generating $b_{k,\theta}$ random Gaussians. To avoid this, we utilize the $2$-stability of Gaussians (Definition \ref{def:stable}), which asserts that $\sum_{t=1}^{b_{k,\theta}} g_t \Delta I_j \sim g \sqrt{b_{k,\theta}}\Delta  I_j$, where $g \sim \mathcal{N}(0,1)$. Thus we can simple generate and store the Gaussian $g$ associated with the item $i \in [n]$, rounding $I_j$, and bucket $A_{k,\theta}$, and on each update $\Delta$ to $f_i$ we can update $A_{k,\theta}$ by $g \sqrt{b_{k,\theta}}\Delta  I_j$. 


Finally, once the number of values in $\{ \texttt{rnd}_\nu(1/t_{i_1}^{1/p}), \texttt{rnd}_\nu(1/t_{i_2}^{1/p}), \dots, \texttt{rnd}_\nu(1/t_{i_{n^{c-1}}}^{1/p}) \}$ which are left to determine is less than $K$ for some $K = \Theta(\log(n))$, we simply generate and hash each of the remaining variables individually. The generation process is the same as before, except that for each of these at most $K$ remaining items we associate a fixed index $i_j$ for $j \in [n^{c-1}]$, and store the relevant random variables $h_\ell(i_j),g_{\ell}(i_j)$ for $\ell \in [d]$. Since the value of $j$ which is chosen for each of these coordinates does not affect the behavior of the algorithm -- in other words the index of the duplicate which is among the $K$ largest is irrelevant -- we can simply choose these indices to be $i_1,i_2,\dots,i_K \in [N]$ so that the first item hashed individually via step $3$ corresponds to $\zeta_{i_1}$, the second to $\zeta_{i_2}$, and so on. 

\begin{figure}[ht!]
	\fbox{\parbox{\textwidth}{ \texttt{Fast-Update} $(i,\Delta,A,B)$\\
			Set $L = n^{c-1}$, and fix $K = \Theta( \log(n))$ with a large enough constant. \\
			\texttt{For} $j =  -\Pi, -\Pi+1,\dots, \Pi-1,\Pi$:
			\begin{enumerate}[topsep=0pt,itemsep=-1ex,partopsep=1ex,parsep=1ex] 
				\item Compute $q_j = \psi(I_{j+1}) - \psi(I_j)$.		
				\item Draw $Q_j \sim \texttt{Bin}(L, q_j)$.
				\item \textbf{If} $L < K$, hash the $Q_j$ items individually into each row $A_{\ell}$ using explicitly stored uniform i.i.d. random variables $h_{\ell}:[n^c] \to [2]$ and Gaussians $g_\ell(j)$ for $\ell \in [d]$.
				\item \textbf{Else:} update count-max table $A$ by via \texttt{Fast-Update-CS}($A,Q_j, I_j,\Delta,i)$ 
				\item $L \leftarrow L - Q_j$. 
			\end{enumerate}
	}}\caption{Algorithm to Update count-max $A$} \label{fig:update}
\end{figure}
\begin{figure}[ht!]
	\fbox{\parbox{\textwidth}{\texttt{Fast-Update-CS} $(A,Q,I,\Delta,i)$\\ 
			Set $W_k =Q$ for $k = 1,\dots,d$ \\
			\texttt{For} $k =  1,\dots,d$, 
			\begin{enumerate}[topsep=0pt,itemsep=-1ex,partopsep=1ex,parsep=1ex]
				\item \texttt{For} $\theta = 1,2$:
				\begin{enumerate}[topsep=0pt,itemsep=-1ex,partopsep=1ex,parsep=1ex]
					\item Draw $b_{k,\theta} \sim \texttt{Bin}(W_k,\frac{1}{2 - \theta + 1})$.
					\item Draw and store $g_{k,\theta, I,i}  \sim \mathcal{N}(0,1)$. Reuse on every call to \texttt{Fast-Update-Cs} with the same parameters $(k,\theta, I,i)$.
					\item Set $A_{k,\theta} \leftarrow A_{k,\theta} + g_{k,\theta,I,i} \sqrt{b_{k,\theta}} \Delta I$
					\item $W_k \leftarrow W_k - b_{k,\theta}$.
				\end{enumerate}
				
			\end{enumerate}
	}}\caption{Update $A$ via updates to $Q$ coordinates, each with a value of $\Delta I$} \label{fig:update-CS}
\end{figure}


Note that the randomness used to process an update corresponding to a fixed $i\in[n]$ is stored so it can be reused to generate the same updates to $A$ whenever an update to $i$ is made. Thus, each time an update $+1$ is made to a coordinate $i \in [n]$, each bucket of count-max is updated by the same value. When an update of size $\Delta$ comes, this update to the count-max buckets is scaled by $\Delta$.
For each $i \in [n]$, let $K_i$ denote the size of $L$ when step $3$ of Figure \ref{fig:update} was first executed while processing an update to $i$. In other words, the coordinates $\zeta_{i_1},\dots,\zeta_{i_{K_i}}$ were hashed into each row $\ell\in [d]$ of $A$ using explicitly stored random variables $h_\ell(i_j),g_\ell(i_j)$. Let $\mathcal{K} = \cup_{i \in [n]} \cup_{j = 1}^{K_i} \{i_j\}$. Then on the termination of the algorithm, to find the maximizer of $\zeta$, the count-max algorithm checks for each $i \in \mathcal{K}$, whether $i$ hashed to the largest bucket (in absolute value) in a row at least a $\frac{4}{5}$ fraction of the time. Count-max then returns the first $i$ which satisfies this, or \texttt{FAIL}. In other words, the count-max algorithm decides to fail or output an index $i$ based on computing the fraction of rows for which $i$ hashes into the largest bucket, instead now it only computes these values for $i \in \mathcal{K}$ instead of $i \in [n^c]$, thus count-max can only return a value of $i \in \mathcal{K}$. We now argue that the distribution of our algorithm is not changed by using the update procedure \texttt{Fast-Update}. This will involving showing that $\arg \max \{|\zeta_i|\} \in \mathcal{K}$ if our algorithm was to return a coordinate originally.


\begin{lemma}\label{lem:runtime}
	 Running the \texttt{$L_p$ sampler} with the update procedure given by \texttt{Fast-Update} results in the same distribution over the count-max table $A$ and $L_2$ estimation vector $B$ as the original algorithm. Moreover, conditioned on a fixed realization of $A,B$, the output of the original algorithm will be the same as the output of the algorithm using \texttt{Fast-Update}. For a given $i \in [n]$, \texttt{Fast-Update} requires $\tilde{O}(\nu^{-1})$-random bits, and runs in time $\tilde{O}(\nu^{-1})$.
\end{lemma}
\begin{proof} To hash an update $\Delta$ to a coordinate $f_i$, the procedure \texttt{Fast-Update} computes the number $Q_j$ of variables in the set $\{ \texttt{rnd}_\nu(1/t_{i_1}^{1/p}),\allowbreak \texttt{rnd}_\nu(1/t_{i_2}^{1/p}),\allowbreak2 \dots, \texttt{rnd}_\nu(1/t_{i_{n^{c-1}}}^{1/p}) \}$ which are equal to $I_j$ for each $j \in \{-\Pi,\dots, \Pi\}$. Instead of computing $Q_j$ by individually generating the variables and rounding them, we utilize a binomial random variable to determine $Q_j$, which results in the same distribution over  $\{ \texttt{rnd}_\nu(1/t_{i_1}^{1/p}),\allowbreak \texttt{rnd}_\nu(1/t_{i_2}^{1/p}),\allowbreak2 \dots, \texttt{rnd}_\nu(1/t_{i_{n^{c-1}}}^{1/p}) \}$.  As noted, with probability $1 - n^{100c}$ none of the variables $\texttt{rnd}_\nu(1/t_{i_j}^{1/p})$ will be equal to $I_k$ for $|k| > \Pi$, which follows from the fact that $n^{-101c} < t_{i} < O(\log(n))$ with probability $1 - n^{-101c}$ and then union bounding over all $n^c$ exponentials variables $t_i$. So we can safely ignore this low probability event.

 Once computed, we can easily sample from the number of items of the $Q_j$ that go into each bucket $A_{k,\theta}$, which is the value $b_{k,\theta}$ in \texttt{Fast-Update-CS} (Figure \ref{fig:update-CS}). By 2-stability of Gaussians (Definition \ref{def:stable}), we can update each bucket $A_{k,\theta}$ by $g_{k,\theta, I_j,i}\sqrt{b_{k,\theta}} \Delta I_j$, which is distributed precisely the same as if we had individually generated each of the $b_{k,\theta}$ Gaussians, and taken their inner product with the vector $\Delta I_j \mathbf{1}$. Storing the explicit values $h_\ell(i_j)$ for the top $K$ largest values of $\texttt{rnd}_\nu(1/t_{i_j}^{1/p})$ does not effect the distribution, but only allows the algorithm to determine the induces of the largest coordinates $i_j$ corresponding to each $i \in [n]$ at the termination of the algorithm. Thus the distribution of updates to $A$ is unchanged by the \texttt{Fast-Update} Procedure.

 
 We now show that the output of the algorithm run with this update procedure is the same as it would have been had all the random variables been generated and hashed individually.
  First observe that for $\nu < 1/2$, no value $q_j = \psi(I_{j+1}) - \psi(I_j)$ is greater than $1/2$. Thus at any iteration, if $L > K$ then $L - \texttt{Bin}(L,q_j) > L/3 $ with probability $1 - n^{-100c}$ by Chernoff bounds (using that $K = \Omega(\log(n))$). Thus the first iteration at which $L$ drops below $K$, we will have $L > K/3$. So for each $i \in [n]$ the top $K/3$ values $\zeta{i_j}$ will be hashed into each row $A_\ell$ using stored random variables $h_\ell(i_j)$, so $K_i > K/3  = \Omega(\log(n))$ for all $i \in [n]$. In particular, $K_i > 0$ for all $i \in [n]$.
 
  Now the only difference between the output procedure of the original algorithm and that of the efficient-update time algorithm is that in the latter we only compute the values of $\alpha_{i_j} = \big| \{ t \in [d] \; | \; |A_{t, h_t(i_j)}| = \max_{r \in \{1,2\}} |A_{t,r} |   \}\big|$ for the $i_j \in [n^c]$ corresponding to the $K_i$ largest values $t_{i_j}^{-1/p}$ in the set $\{t_{i_1}^{-1/p},\dots,t_{i_{n^{c-1}}}^{-1/p}\}$, whereas in the former all values of $\alpha_{i_j}$ are computed to find a potential maximizer. In other words, count-max with \texttt{Fast-Update} only searches throguh the subset $\mathcal{K} \subset[n^c]$ for a maxmizer instead of searching through all of $[n^c]$ (here $\mathcal{K}$ is as defined earlier in this section). Since count-max never outputs a index $i_j$ that is not a unique maximizer with high probability, we know that the output of the original algorithm, if it does not fail, must be $i_j$ such that $j = \arg \max_{j'} \{t_{i_{j'}} \}$, and therefore $i_j \in \mathcal{K}$. Note the $n^{-c}$ failure probability can be safely absorbed into the additive $n^{-c}$ error of the perfect $L_p$ sampler. Thus the new algorithm will also output $i_j$. Since the new algorithm with \texttt{Fast-Update} searches over the subset $\mathcal{K} \subset [n^c]$ for a maximier, if the original algorithm fails then certainly so will \texttt{Fast-Update}. Thus the output of the algorithm using \texttt{Fast-Update} is distributed identically (up to $n^{-c}$ additive error) as the output of the original algorithm, which completes the proof.

\paragraph{Runtime \& Random Bits} For the last claim, first note that it suffices to generate all continuous random varaibles used up to $(nmM)^{-c} = 1/\poly(n)$ precision, which is $1/\poly(n)$ additive error after conditioning on the event that all random variables are all at most $\poly(n)$ (which occurs with probability $1-n^{-c})$, and recalling that the length of the stream $m$ satisfies $m < \poly(n)$ for a suitably smaller $\poly(n)$ then as in the additive error. More formally, we truncate the binary representation of every continuous random variable (both the exponentials and Gaussians) after $O(\log(n))$-bits with a sufficiently large constant. This will result in at most an additive $1/\poly(n)$ error for each bucket $A_{i,j}$ of $A$, which can be absorbed by the adversarial  error $\mathcal{V}_{i,j}$ with $|\mathcal{V}_{i,j}| = O(\nu( |A_{i,1}| + |A_{i,2}|))$ that we incur in each of these buckets already in Lemma \ref{lem:main}. Thus each random variable requires $O(\log(n))$ bits to specify. Similarly, a precision of at most $(nmM)^{-c}$ is needed in the computation of the $q_j$'s in Figure \ref{fig:update} by Lemma  \ref{lem:bin}, since the routine to compute \texttt{Bin}$(n,q_j)$ will terminate with probability $1-n^{-100c}$ after querying at most $O(\log(n))$ bits of $q_j$. 
Now there are at most $2\Pi = O(\nu^{-1}\log(n))$ iterations of the loop in \texttt{Fast-Update}. Within each, our call to sample a binomial random variable is carried out in $\tilde{O}(1)$ time with high probability by Lemma \ref{lem:bin} (and thus ue at most $\tilde{O}(1)$ random bits), and there are $\tilde{O}(1)$ entries in $A$ to update (which upper bounds the running time and randomness requirements of \texttt{Fast-Update-CS}).

 Note that since the stream has length $m = \poly(n)$, and there are at most $\tilde{O}(\nu)$ calls made to sample binomial random variables in each, we can union bound over each call to guarantee that each returns in $\tilde{O}(1)$ time with probability $1 - n^{-100c}$. Since $K = \tilde{O}(1)$, we must store an additional $\tilde{O}(1)$ random bits to store the individual random variables $h_\ell(i_j)$ for $i_j \in \{i_1,\dots,i_{K_i}\}$. Similarly, we must store $\tilde{O}(\nu)$ independent Gaussians for the procedure \texttt{Fast-Update-CS}, which also terminates in $\tilde{O}(1)$ time (noting that $r = O(\log(n))$), which completes the proof.
 
\end{proof}

\subsection{Derandomizing the Algorithm} \label{subsec:derandom}

We now show that our algorithm \texttt{$L_p$ Sampler} with \texttt{Fast-Update} can be derandomized without affecting the space or time complexity. To do this, we use a combination of  Nisan's pseudorandom generator (PRG) \cite{nisan1992pseudorandom}, and the  PRG of Goplan, Kane, and Meka \cite{gopalan2015pseudorandomness}.  We begin by introducing  Nisan's PRG, which is a deterministic map $G: \{0,1\}^\ell \to \{0,1\}^T$, where $T \gg \ell$ (here we think of $T = \poly(n)$ and $\ell = O(\log^2(n))$). Let $\sigma: \{0,1\}^T \to \{0,1\}$ be a efficiently computable tester. For the case of Nisan's PRG, $\sigma$ must be a tester which reads its random $T$-bit input in a stream, left to right, and outputs either $0$ or $1$ at the end.  Nisan's PRG can be used to fool \textit{any} such tester, which means:
\[\big|	\pr{\sigma(U_T) = 1 } - \pr{\sigma(G(U_\ell)) = 1}	\big| < \frac{1}{T^c} \]
Where $U_t$ indicates $t$ uniformly random bits for any $t$, and $c$ is a sufficiently large constant. Here the probability is taken over the choice of the random bits $U_T$ and $U_\ell$. In other words, the probability that $\sigma$ outputs $1$ is nearly the same when it is given random input as opposed to input from Nisan's generator. Nisan's theorem states if $\sigma$ has at most $\poly(T)$ states and uses a working memory tape of size at most $O(\log(T))$, then a seed length of $\ell = O(\log^2(T))$ suffices for the above result \cite{nisan1992pseudorandom}. 
Thus Nisan's PRG fools space bounded testers $\sigma$ that read their randomness in a stream.

\paragraph{Half Space Fooling PRG's.}
 Our derandomization crucially uses the PRG of Goplan, Kane, and Meka \cite{gopalan2015pseudorandomness}, which fools a certain class of fourier transforms.  Utilizing the results of \cite{gopalan2015pseudorandomness}, we will design a PRG that can fool \textit{arbitrary} functions of $\lambda = O(\log(n))$ halfspaces, using a seed of length $O(\log^2(n)(\log\log(n))^2)$. We remark that in \cite{gopalan2015pseudorandomness} it is shown how to obtain such a PRG for a function of a single half-space. Using extensions of the techniques in that paper, we demonstrate that the same PRG with a smaller precision $\eps$ can be used to fool functions of more half-spaces.   We now introduce the main result of  \cite{gopalan2015pseudorandomness}.  Let $\C_1 = \{c \in \C \; | \; |c| \leq 1\}$.
 \begin{definition}[Definition 1 \cite{gopalan2015pseudorandomness}]\label{def:fourier}
 	An $(m,n)$-Fourier shape $	f:[m]^n \to \C_1$ is a function of the form $f(x_1,\dots,x_n) = \prod_{j=1}^n f_j(x_j)$ where each $f_j:[m] \to \C_1$.
 \end{definition}
 
 \begin{theorem}[Theorem 1.1 \cite{gopalan2015pseudorandomness}]\label{thm:goplan}
 	There is a PRG $G:\{0,1\}^\ell \to [m]^n$ that fools all $(m,n)$-Fourier shapes $f$ with error $\eps$ using a seed of length $\ell = O(\log(mn/\eps) (\log\log(mn/\eps))^2)$,  meaning:
 	\[	\left|	\bex{f(x)} - \bex{f(G(y))}		\right| \leq \eps	\]
	where $x$ is uniformly chosen from $ [m]^n$ and $y$ from $\{0,1\}^\ell$.
 \end{theorem}

 For any $a^1,\dots, a^\lambda \in \Z^n$ and $\theta_1,\dots,\theta_\lambda \in \Z$, let $H_i: \R^n \to \{0,1\}$, be the function given by $H_i(X_1,\dots,X_n) = \mathbf{1}[a_1^i X_1 + a_2^i X_2 + \dots + a_n^i X_n > \theta_i ]$, where $\mathbf{1}$ is the indicator function. We now define the notion of a $\lambda$-halfspace tester, and what it means to fool one.

 \begin{definition}[$\lambda$-halfspace tester]
 	A $\lambda$-halfspace tester is any function $\sigma_H: \R^n \to \{0,1\}$ which, on input $X=(X_1,\dots,X_n)$, outputs $\sigma_H'(H_1(X),\dots,H_\lambda(X)) \in \{0,1\}$ where $\sigma_H'$ is any fixed function $\sigma_H': \{0,1\}^\lambda \to \{0,1\}$. In other words, the Boolean valued function $\sigma_H(X)$ only depends on the values $(H_1(X),\dots,H_\lambda(X))$. A $\lambda$-halfspace tester is said to be $M$ bounded if all the half-space coefficents $a_j^i$ and $\theta_i$ are integers of magnitude at most $M$, and each $X_i$ is drawn from a discrete distrubtion $\mathcal{D}$ with support contained in $\{-M,\dots,M\} \subset \Z$.  
 \end{definition}

\begin{definition}[Fooling a $\lambda$-halfspace tester]
	A PRG $G: \{0,1\}^\ell \to \R^n$ is said to $\eps$-fools the class of $\lambda$-halfspace testers under a distribution $\mathcal{D}$ over $\R^n$ if for every set of $\lambda$ halfspaces $H = (H_1,\dots,H_\lambda)$ and every $\lambda$-halfspace tester $\sigma_H: \R^n \to \{0,1\}$, we have:
	\[\big|	\mathbb{E}_{X \sim \mathcal{D}}\big[ \sigma_H(X) = 1 \big] - \mathbb{E}_{y \sim \{0,1\}^\ell} \big[ \sigma_H(G(y)) = 1 \big]\big| < \eps \]
	Here $\ell$ is the seed length of $G$. 
\end{definition}

  We will consider only product distributions $\mathcal{D}$. In other words, we assume that each coordiante $X_i$ is drawn i.i.d. from a fixed distribution $\mathcal{D}$ over $\{-M,\dots,M\} \subset \Z$. We consider PRG's $G: \{0,1\}^\ell \to \{-M,\dots,M\}^n$ which take in a random seed of length $\ell$ and output a $X' \in \{-M,\dots,M\}^n$ such that any $M$-bounded $\lambda$-halfspace tester will be unable to distinguish $X'$ from $X \sim \mathcal{D}^n$ (where $\mathcal{D}^n$ is the product distribution of $\mathcal{D}$, such that each $X_i \sim \mathcal{D}$ independently). The following Lemma demonstrates that the PRG of \cite{gopalan2015pseudorandomness} can be used to fool $M$-bounded $\lambda$-halfspace testers. The authors would like to thank Raghu Meka for providing us a proof of Lemma \ref{lemma:halfspace}.
 
 \begin{lemma}\label{lemma:halfspace}
Suppose $X_i \sim \mathcal{D}$ is a distribution on $\{-M,\dots,M\}$ that can be sampled from with $\log(M') = O(\log(M))$ random bits. Then, for any $\eps > 0$ and constant $c \geq 1$, there is a PRG $G:\{0,1\}^\ell \to \{-M,\dots,M\}^n$ which $\eps (nM)^{-c\lambda}$-fools the class of all $M$-bounded $\lambda$-halfspace testers on input $X \sim \mathcal{D}^n$ with a seed of length $\ell = O(\lambda \log(nM/\eps) (\log\log(nM/\eps))^2)$ (assuming $\lambda \leq n$). Moreover, if $G(y) = X' \in \{-M,\dots,M\}^n$ is the output $G$ on random seed $y \in \{0,1\}^\ell$, then each coordinate $X_i'$ can be computed in $O(\ell)$-space and in $\tilde{O}(1)$ time, where $\tilde{O}$ hides $\poly(\log(nM))$ factors.
 \end{lemma}
\begin{proof}
	 Let  $X = (X_1,\dots,X_n)$ be uniformly chosen from $[M']^n$ for some $M' = \poly(M)$, and let $Q:[M'] \to \{-M,\dots,M\}$ be such that $Q(X_i) \sim \mathcal{D}^n$ for each $i \in [n]$. Let $a^1,\dots,a^\lambda \in \Z^n$, $\theta_1,\dots,\theta_\lambda \in \Z$ be $\log(M)$-bit integers, where $H_i(x) = \mathbf{1}[\langle a^i, x \rangle > \theta_i]$.
Let $Y_i = \langle Q(X), a^i\rangle - \theta_i$. Note that $Y_i \in [-2M^2n,2M^2n ]$. So fix any $\alpha_i \in [-2M^2n,2M^2n ]$ for each $i \in [\lambda]$, and let $\alpha = (\alpha_1,\dots,\alpha_\lambda)$. Let $h_\alpha(x) = \textbf{1}(Y_1 = \alpha_1) \cdot \textbf{1}(Y_2 = \alpha_2) \cdots \textbf{1}(Y_\lambda = \alpha_\lambda)$, where $\textbf{1}(\cdot)$ is the indicator function. Now define $f(x) = \sum_{j=1}^\lambda (2M^2n)^{j-1} \langle a^i, x\rangle$ for any $x \in \Z^n$. Note that $f(Q(X)) \in \{-(Mn)^{O(\lambda)}, \dots, (Mn)^{O(\lambda)}\}$. We define the  Kolmogorov distance between two integer valued random variables $Z,Z'$ by $d_K(Z,Z')= \max_{k \in \Z}(\left|\pr{Z \leq k} - \pr{Z' \leq k}	\right|)$. Let $X' \in [M']^n$ be generated via the $(M',n)$-fourier shape PRG of \cite{gopalan2015pseudorandomness} with error $\eps'$ (Theorem 1.1 \cite{gopalan2015pseudorandomness}). 
  Observe $\ex{h_\alpha(Q(X))} = \pr{f(Q(X)) = \sum_{j=1}^\lambda (Mn)^{j-1} \alpha_j}$, so 
 \[ \left|	\ex{h_\alpha(Q(X))} - \ex{h_\alpha(Q(X')}	\right| \leq d_K(f(Q(X)), f(Q(X')))\]
 
 Now by Lemma 9.2 of \cite{gopalan2015pseudorandomness}, $d_K(f(Q(X)), f(Q(X'))) = O\big(\lambda \log(Mn)d_{FT}\big(f(Q(X)), f(Q(X'))\big) \big)$, where for integer valued $Z,Z'$, we define $d_{FT}(Z, Z') = \max_{\beta \in [0,1]} \left|	\ex{\exp(2\pi i \beta Z)} - \ex{\exp(2\pi i \beta Z')} 	\right|$. Now $\exp(2\pi i \beta f(Q(X))) = \prod_{i=1}^n ((\sum_{j=1}^\lambda  (2M^2n)^{j-1}  a^j_i) Q(X_i)) $, which is a $(M',n)$-Fourier shape as in Definition \ref{def:fourier}. Thus by Theorem \ref{thm:goplan} (Theorem 1.1 of \cite{gopalan2015pseudorandomness}), we have $d_{FT}(f(Q(X)), f(Q(X'))) \leq \eps'$. Thus 
 \[ \left|	\ex{h_\alpha(Q(X))} - \ex{h_\alpha(Q(X')}	\right| = O(\lambda \log(Mn) \eps') \]
 Now let $\sigma_H(x) = \sigma_H'(H_1(x),\dots,H_\lambda(x))$ be any $M$-bounded $\lambda$-halfspace tester on $x \sim \mathcal{D}^n$. Since the inputs to the halfspaces $H_i$ of $\sigma_H'$ are all integers in $\{-2M^2 n,2M^2n\}$, let $A \subset \{-2M^2 n,2M^2n\}$ be the set of $\alpha \in A$ such that $Y = (Y_1,\dots, Y_\lambda) = \alpha$ implies that $\sigma_H(Q(X)) = 1$, where $Q(X) \sim \mathcal{D}^n$ as above. Recall here that  $Y_i = \langle Q(X), a^i\rangle - \theta_i$. Then we can think of a $\sigma_H(X) = \sigma_H''(Y_1,\dots,Y_\lambda)$ for some function $\sigma_H'':\{-2M^2 n,\dots,2M^2n\}^\lambda \to \{0,1\}$, and in this case we have $A = \{\alpha \in \{-2M^2 n,2M^2n\} \; | \; \sigma_H''(\alpha) = 1\}$. Then 
 
\begin{equation*}
 \begin{split}
\left|	\ex{\sigma_H(Q(X))} - \ex{\sigma_H(Q(X'))}	\right|  & \leq	\sum_{\alpha \in A} \left|	\ex{h_\alpha(Q(X))} - \ex{h_\alpha(Q(X')}	\right|  \\
 & \leq	\sum_{\alpha \in A} O(\lambda \log(Mn) \eps')  \\
 \end{split}
 \end{equation*}
 Now note that $|A| = (nM)^{O(\lambda)}$, so setting $\eps' = \eps(nM)^{-O(\lambda)}$ with a suitably large constant, we obtain $\left|	\ex{\sigma_H(Q(X))} - \ex{\sigma_H(Q(X'))}	\right| \leq \eps (nM)^{-c\lambda}$ as needed.  By Theorem \ref{thm:goplan}, the seed required is $\ell = O(\lambda \log(nM/\eps) (\log\log(nM/\eps))^2)$ as needed. The space and time required to compute each coordinate follows from Proposition \ref{prop:PRGspacetime} below.
  
\end{proof}

\begin{proposition}\label{prop:PRGspacetime}
	In the setting of Lemma \ref{lemma:halfspace}, if $G(y) = X' \in \{-M,\dots,M\}^n$ is the output $G$ on random seed $y \in \{0,1\}^\ell$, then each coordinate $X_i'$ can be computed in $O(\ell)$-space and in $\tilde{O}(1)$ time, where $\tilde{O}$ hides $\poly(\log(nM))$ factors.
\end{proposition}
\begin{proof}
 In order to analyze the space complexity and runtime needed to compute a coordinate $X_i'$, we must describe the PRG of Theorem \ref{thm:goplan}. 
 The Goplan-Kane-Meka PRG has 3 main components, which themselves use other PRGs such as Nisan's PRG as sub-routines.
 Recall that the PRG generates a psuedo-uniform element from $X \sim [m]^n$ that fools a class of Fourier shapes $f: [m]^n \to \mathbb{C}$ on truly uniform input in $[m]^n$. Note that because of the definition of a Fourier shape, if we wish to sample from a distribution $X \sim \mathcal{D}$ over $\{-m,\dots,m\}^n$ that is not uniform, but such that $X_i$ can be sampled with $\log(m')$-bits, we can first fool Fourier shapes $f:[m']^n \to \mathbb{C}$, and then use a function $Q:[m'] \to \{-m,\dots,m\}$ which samples $X_i \sim \mathcal{D}$ given $\log(m')$ uniformly random bits. We then fool Fourier shapes $F = \prod_{i=1}^{n} f_j(x) = \prod_{i=1}^{n} f_j(Q(x))$ where $x$ is uniform, and thus $Q(x) \sim \mathcal{D}$. Thus it will suffice to fool $(m',n)$-Fourier shapes on uniform distributions.
 For simplicity, for the most part we will omit the parameter $\eps$ in this discussion.
 
 The three components of the PRG appear in Sections 5,6, and 7 of \cite{gopalan2015pseudorandomness} respectively. In this proof, when we write Section $x$ we are referring to the corresponding Section of \cite{gopalan2015pseudorandomness}. They consider two main cases: one where the function $f$ has \textit{high variance} (for some notion of variance), and one where it has low variance. The PRGs use two main pseudo-random primitives, $\delta$-biased and $k$-wise independent hash function families. Formally, a family $\mathcal{H} = \{h : [n] \to [m]\}$ is said to be $\delta$-biased if for all $r \leq n$ distinct indices $i_1,\dots,i_r \in [n]$ and $j_1,\dots,j_r \in [m]$ we have 
 \[		\textbf{Pr}_{h \sim \mathcal{H}}\left[h(i_1) = j_1 \wedge \dots \wedge h(i_r) = j_r		\right] = \frac{1}{m^r} \pm \delta	\] 
 The function is said to be $k$-wise independent if it holds with $\delta = 0$ for all $r \leq k$. It is standard that $k$-wise independent families can be generated by taking a polynomial of degree $k$ over a suitably large finite field (requiring space $O(k \log(mn))$). Furthermore, a value $h(i)$ from a $\delta$-biased family can be generated by taking products of two $O(\log(n/\delta))$-bit integers over a suitable finite field \cite{SwastikLecture} (requiring space $O(\log(n/\delta))$). So in both cases, computing a value $h(i)$ can be done in space and time that is linear in the space required to store the hash functions (or $O(\log(n/\delta))$-bit integers). Thus, any nested sequence of such hash functions used to compute a given coordinate $X_i'$ can be carried out in space linear in the size required to store all the hash functions. 
 
Now the first PRG (Section 5 \cite{gopalan2015pseudorandomness}) handles the high variance case. The PRG first subsamples the $n$ coordinates at $\log(n)$ levels using a pair-wise hash function (note that a $2$-wise permutation is used in \cite{gopalan2015pseudorandomness}, which reduces to computation of a $2$-wise hash function). In each level $S_j$ of sub-sampling,
 it uses $O(1)$-wise independent hash functions to generate the coordinates $X_i \in S_j$. So if we want to compute a value $X_i$, we can carry out one hash function computation $h(i)$ to determine $j$ such that $X_i \in S_j$, and then carry out another hash function computation $h_j(i) = X_i$.  
 Instead of using $\log(n)$ independent hash functions $h_j$, each of size $O(\log(nm))$, for each of the buckets $S_j$, they derandomize this with the PRG of Nisan and Zuckerman \cite{nisan1996randomness} to use a single seed of length $O(\log(n))$. Now the PRG of Nisan and Zuckerman can be evaluated online, in the sense that it reads its random bits in a stream and writes its pseudo-random output on a one-way tape, and runs in space linear in the seed required to store the generator itself (see Definition 4 of \cite{nisan1996randomness}). Such generators are composed to yield the final PRG of Theorem $2$ \cite{nisan1996randomness}, however by Lemma 4 of the paper, such online generators are composable. Thus the entire generator of \cite{nisan1996randomness} is online, and so any substring of the pseudorandom output can be computed in space linear in the seed of the generator by a single pass over the random input. Moreover, by Theorem $1$ of \cite{nisan1996randomness} in the setting of \cite{gopalan2015pseudorandomness}, such a substring can be computed in $\tilde{O}(1)$ time, since it is only generating $\tilde{O}(1)$ random bits to begin with.
 
On top of this, the PRG of Section $5$ \cite{gopalan2015pseudorandomness} first splits the coordinates $[n]$ via a limited independence hash function into $\poly( \log(1/\eps))$ buckets, and applies the algorithm described above on each. 
 To do this second layer of bucketing and not need fresh randomness for each bucket, they use Nisan's PRG \cite{nisan1992pseudorandom} with a seed of length $\log(n) \log\log(n)$. Now any bit of Nisan's PRG can be computed by several nested hash function computations, carried out in space linear in the seed required to store the PRG. Thus any substring of Nisan's can be computed in space linear in the seed and time $\tilde{O}(1)$. Thus to compute $X_i'$, we first determine which bucket it hashes to, which involves computing random bits from Nisan's PRG. Then we determine a second partitioning, which is done via a $2$-wise hash fucntion, and finally we compute the value of $X_i'$ via an $O(1)$-wise hash function, where the randomness for this hash function is stored in a substring output by the PRG of \cite{nisan1996randomness}.  Altogether, we conclude that the PRG of Section $5$ \cite{gopalan2015pseudorandomness} is such that value $X_i'$ can be computed in space linear in the seed length and $\tilde{O}(1)$ time.
 
 Next, in Section 6 of \cite{gopalan2015pseudorandomness}, another PRG is introduced which 
   reduces the problem to the case of $m \leq poly(n)$. Assuming a PRG $G_1$ is given which fools $(m,n)$-Fourier shapes, they design a PRG $G_2$ using $G_1$ which fools $(m^2,n)$-Fourier shapes. Applying this $O(\log \log(m))$ times reduces to the case of $m \leq n^4$.  The PRG is as follows. Let $G_1,\dots,G_t$ be the iteratively composed generators, where $t = O(\log \log(m))$. To compute the value of $(G_i)_j \in [m]$, where $(G_i)_j$ is the $j$-th coordinate of $G_i \in [m]^n$, the algorithm first implicitly generates a matrix $Z \in [m]^{\sqrt{m} \times m}$. An entry $Z_{p,q}$ is generated as follows. First one applies a $k$-wise hash function $h(q)$ (for some $k$), and uses the $O(\log(m))$-bit value of $h(q)$ as a seed for a second $2$-wise indepedent hash function $h_{h(q)}'$. Then $Z_{p,q} = h_{h(q)}'(p)$. Thus within a column $q$ of $Z$, the entries are $2$-wise independent, and separate columns of $Z$ are $k$-wise independent. This requires $O(k \log(m))$-space to store, and the nested hash functions can be computed in $O(k \log(m))$-space. Thus computing $Z_{i,j}$ is done in $\tilde{O}(1)$ time and space linear in the seed length. Then we set $(G_i)_j = Z_{(G_{i-1})_j, j}$ for each $j \in [n]$. Thus $(G_i)_j $ only depends on $(G_{i-1})_j$, and the random seeds stored for two hash functions to evaluate entries of $Z$. So altogether, the final output coordinate $(G_t)_j$ can be computed in space linear in the seed length required to store all required hash functions, and in time $\tilde{O}(1)$. Note importantly that the recursion is linear, in the sense that computing $(G_{i})_j$ involves only one query to compute $(G_i)_{j'}$ for some $j'$.
   
 Next, in Section 7 of \cite{gopalan2015pseudorandomness}, another PRG is introduced for the \textit{low-variance case}, which reduces the size of $n$ to $\sqrt{n}$, but blows up $m$ polynomially in the process. Formally, it shows given a PRG $G_1'$ that fools $(\poly(n),\sqrt{n})$ Fourier shapes, one can design a PRG $G_2'$ that fools $O(m,n)$-Fourier shapes with $m < n^4$ (here the $\poly(n)$ can be much larger than $n^4$).  To do so, the PRG first hashes the $n$ coordinates into $\sqrt{n}$ buckets $k$-wise independently, and then in each bucket uses $k$-wise independence to generate the value of the coordinate. A priori, this requires $\sqrt{n}$ independent seeds for the hash function in each of the buckets. To remove this requirement, it uses $G_1'$ to generate the $\sqrt{n}$ seeds required from a smaller seed. Thus to compute a coordinate $i$ of $G_2'$, simply evaluate a $k$-wise independent hash function on $i$ to determine which bucket $j \in [\sqrt{n}]$ a the item $i$ is hashed into. Then evaluate $G_1'(j)$ to obtain the seed required for the $k$-wise hash function $h_j$, and the final result is given by $h_j(i)$. Note that this procedure only requires one query to the prior generator $G_1'$. The space required to do so is linear in the space required to store the hash functions, and the space required to evaluate a coordinate of the output of $G_1'$, which will be linear in the size used to store $G_1'$ by induction. 
 
 Finally, the overall PRG composes the PRG from Section $6$ and $7$ to fool larger $n,m$ in the case of low variance. Suppose we are given a PRG $G_0$ which fools $(m'',\sqrt{n'})$-Fourier shapes for some $m'' < (n')^2$. We show how to construct a PRG $G_1$ which fools $(m',n')$-Fourier shapes for any $m' \leq (n')^4$. 
  Let $G^{6+7}$ be the PRG obtained by first applying the PRG from Section $6$ on $G_0$ as an initial point, which gives a PRG that fools $(\poly(n'), \sqrt{n'})$-Fourier shapes, and then applying the PRG from section $7$ on top which now fools $(m', n')$-Fourier shapes (with low variance).  
   Let $G^5$ be the generator from Section $5$ which fools $(m',n')$-Fourier shapes with high variance. The final algorithm for fooling the class of all $(m',n')$-Fourier shapes given $G_0$ computes a generator $G_1$ such that the $i$-th coordinate is $(G_1)_i = (G^{6+7})_i \oplus (G^5)_i$, where $\oplus$ is addition mod $m'$. This allows one to simultaneously fool high and low variance Fourier shapes of the desired $m',n'$.  If $m > (n')^4$, one can apply the PRG for Section $6$ one last time on top of $G_1$ to fool arbitrary $m$.  Thus if for any $i$, the $i$-th coordinate of $G_{6+7}$ and $G_5$ can be composed in $\tilde{O}(1)$ time and space linear in the size required to store the random seed, then so can $G_i$. Thus going from $G_0$ to $G_1$ takes a generator that fools $(m'',\sqrt{n'})$ to $(m',n')$-Fourier shapes, and similarly we can compose this to design a $G_2$ that fools $(m',(n')^2)$-Fourier shapes. Composing this $t = O(\log \log n)$-times, we obtain $G_t$ which fools $O(m,n)$ Fourier shapes for any $m,n$. As a base case (to define the PRG $G_0$), the PRG of \cite{nisan1996randomness} is used, which we have already discussed can be evaluated on-line in space linear in the seed required to store it and time polynomial in the length of the seed.
   
    Now we observe an important property of this recursion. At every step of the recursion, one is tasked with computing the $j$-th coordinate output by some PRG for some $j$, and the result will depend \textit{only} on a query for the $j'$-th coordinate of another PRG for some $j'$ (as well as some additional values which are computed using the portion of the random seed dedicated to this step in the recursion). Thus at every step of the recursion, only one query is made for a coordinate to a PRG at a lower level of the recursion. Thus the recursion is linear, in the sense that the computation path has only $L$ nodes instead of $2^L$ (which would occur if two queries to coordinate $j',j''$ were made to a PRG in a lower level). Since at each level of recursion, computing $G^{6+7}$ itself uses $O(\log \log(nm))$ levels of recursion, and also has the property that each level queries the lower level at only one point, it follows that the total depth of the recursion is $O(( \log \log (nm))^2)$. At each point, to store the information required for this recursion on the stack requires only $O(\log(nm))$-bits of space to store the relevant information identifying the instance of the PRG in the recursion, along with its associated portion of the random seed. Thus the total space required to compute a coordinate via these $O(\log \log (nm))^2)$ recursions is $O(\log(nm) (\log \log nm)^2)$, which is linear in the seed length. Moreover, the total time $\tilde{O}(1)$, since each step of the recursion requires $\tilde{O}(1)$.
\end{proof}

We use the prior technique to derandomize a wide class of linear sketches $A\cdot f$ such that the entries of $A$ are independent, and can be sampled using $O(\log(n))$-bit, and such that the behavior of the algorithm only depends on the sketch $Af$. It is well known that there are strong connections between turnstile streaming algorithms and linear sketches are, insofar as practically all turnstile streaming algorithms are in fact linear sketches. The equivalence of turnstile algorithms and linear sketches has even been formalized \cite{li2014turnstile}, with some restrictions. Our results show that all such sketches that use independent, efficiently sampled entries in their sketching matrix $A$ can be derandomized with our techniques. As an application, we derandomize the count-sketch variant of Minton and Price \cite{minton2014improved}, a problem which to the best of the authors knowledge was hitherto open.

\begin{lemma}\label{lemma:derandomGeneral}
	Let \texttt{ALG} be any streaming algorithm which, on stream vector $f \in \{-M,\dots,M\}^n$ for some $M = \poly(n)$, stores only a linear sketch $A\cdot f$ such that the entries of the random matrix $A \in \R^{k \times n}$ are i.i.d., and can be sampled using $O(\log(n))$-bits. Fix any constant $c \geq 1$. Then \texttt{ALG} can be implemented using a random matrix $A'$ using $O(k \log(n) (\log \log n)^2)$ bits of space, such that for every vector $y \in \R^k$ with entry-wise bit-complexity of $O(\log(n))$, 
	\[\left|	\bpr{Af = y}	- \bpr{A' f = y}		\right| < n^{-c k} \]
\end{lemma}
\begin{proof}
	We can first scale all entries of the algorithm by the bit complexity so that each entry in $A$ is a $O(\log(n))$-bit integer. Then by Lemma \ref{lemma:halfspace}, we can store the randomness needed to compute each entry of $A'$ with $O(k \log(n) (\log \log n)^2)$-bits of space, such that $A'$ $n^{-c k}$-fools the class of all $O(k)$-halfspace testers, in particular the one which checks, for each coordinate $i \in [k]$, whether both $(A'f)_i < y + 1$ and $(A'f)_i > y_i - 1$, and accepts only if both hold of all $i \in [k]$. By Proposition \ref{prop:PRGspacetime}, the entries of $A'$ can be computed in space linear in the size of the random seed required to store $A'$. Since we have scaled all values to be integers, $n^{-ck}$ fooling this tester is equivalent to the theorem statement.  Note that the test $(A'f)_i < y + 1$ can be made into a half-space test as follows. Let $X^i \in \R^{n k}$ be the vector such that $X^i_{j + (i-1)n} = f_j$ for all $j \in [n]$ and $X^i_j=0$ otherwise. Let $\ttx{vec}(A) \in \R^{n k}$ be the vectorization of $A$. Then $(Af)_i = \langle \texttt{vec}(A), X^i \rangle$, and all the entries of $\ttx{vec}(A)$ are i.i.d., which allows us to make the stated constraints into the desired half-space constraints.
\end{proof}

Observe that the above Lemma derandomized the linear sketch $Af$ by writing each coordinate $(Af)_i$ as a linear combination of the random entries of $\ttx{vec}(A)$. Note, however, that the above proof would hold if we added the values of any $O(k)$ additional linear combinations $\langle X_j, \ttx{vec}(A)$ to the Lemma, where each $X_j \in \{-M,\dots,M\}^{kn}$. This will be useful, since the behavior of some algorithms, for instance count-sketch, may depend not only on the sketch $Af$ but also on certain values or linear combinations of values within the sketch $A$. This is formalized in the following Corollary.

\begin{corollary}\label{cor:derandomGeneral}
	Let the entries of $A \in \R^{k \times n}$ be drawn i.i.d. from a distribution which can be sampled using $O(\log n)$-bits, and let $\ttx{vec}(A) \in \R^{nk}$ be the vectorization of $A$.  Let $X \in \R^{t \times nk}$ be any fixed matrix with entries contained within $\{-M,\dots,M\}$, where $M = \poly(n)$. Then there is a distribution over random matrices $A' \in \R^{k \times n}$ which can be generated and stored using $O(t \log(n) (\log \log n)^2)$ bits of space, such that for every vector $y \in \R^t$ with entry-wise bit-complexity of $O(\log(n))$, 
	\[\left|	\bpr{X\cdot \ttx{vec}(A) = y}	- \bpr{X\cdot \ttx{vec}(A') = y}		\right| < n^{-c t} \]
\end{corollary}
\begin{proof}
The proof is nearly identical to Lemma \ref{lemma:derandomGeneral}, where we first scale entries to be $O(\log(n))$-bit integers, and then apply two half-space tests to each coordinate of $X\cdot \ttx{vec}(A')$. 
\end{proof}

\begin{theorem}\label{thm:derandomGeneral}
		Let \texttt{ALG} be any streaming algorithm which, on stream vector $f \in \{-M,\dots,M\}^n$ and fixed matrix $X \in \R^{t \times nk}$ with entries contained within $\{-M,\dots,M\}$, for some $M = \poly(n)$, outputs a value that only depends on the sketches $A\cdot f$ and $X \cdot \ttx{vec}(A)$. Assume that the entries of the random matrix $A \in \R^{k \times n}$ are i.i.d. and can be sampled using $O(\log(n))$-bits.  Let $\sigma: \R^{k } \times \R^{t} \to \{0,1\}$ be any tester which measures the success of  \texttt{ALG}, namely $\sigma(A f, X \cdot \ttx{vec}(A)) = 1$ whenever \texttt{ALG} succeeds.  Fix any constant $c \geq 1$. Then \texttt{ALG} can be implemented using a random matrix $A'$ using a random seed of length $O((k+t) \log(n) (\log \log n)^2)$, such that:
	\[\left|	\bpr{\sigma(Af,X \cdot \ttx{vec}(A) ) = 1}	- \bpr{\sigma(A'f, X \cdot \ttx{vec}(A')) = 1}		\right| < n^{-c (k+t)} \]
	and such that each entry of $A'$ can be computed in time $\tilde{O}(1)$ and using working space linear in the seed length.
\end{theorem}
\begin{proof}
As in the Lemma \ref{lemma:derandomGeneral}, we first scale all entries of the algorithm by the bit complexity so that each entry in $A$ is $O(\log(n))$-bit integer. Then there is a $M' = \poly(n)$ such that each entry of $A\cdot f$ and $X \cdot \ttx{vec}(A)$ will be a integer of magnitude at most $M'$. First note that the sketch $A\cdot f$ and $X \cdot \ttx{vec}(A)$ can be written as one linear sketch $X_0 \cdot \ttx{vec}(A)$ where $X_0 \in \R^{k + t \times kn}$. Then $\sigma$ can be written as a function $\sigma:\R^{k + t} \to \{0,1\}$ evaluated on $\sigma(X_0 \cdot \ttx{vec}(A))$. Let  $S  = \{y \in \{-M',\dots,M'\}^{k+t}  \; | \; \sigma(y ) = 1 \}$. 
Then by Corollary \ref{cor:derandomGeneral},  we have $$\left|	\pr{ X_0\cdot \ttx{vec}(A)=y }	- \pr{X_0 \cdot \ttx{vec}(A')=y }		\right| < n^{-c (k+t)}$$ for all $y \in S$. Taking $c$ sufficiently large, and noting $|S| = n^{-O(k+t)}$, we have $\pr{\sigma(X_0\cdot \ttx{vec}(A)) = 1} = \sum_{y \in S} \bpr{X_0\cdot \ttx{vec}(A) = y} = \sum_{y \in S} (\pr{X_0\cdot \ttx{vec}(A') = y} \pm  n^{-c (k+t)}) = \pr{\sigma(X_0\cdot \ttx{vec}(A')) = 1}	+ n^{-O(k+t)}$ as desired. The final claim follows from Proposition \ref{prop:PRGspacetime}.
\end{proof}

We know show how this general derandomization procedure can be used to derandomize the count-sketch variant of Minton and Price \cite{minton2014improved}. Minton and Price's analysis shows improved concentration bounds for count-sketch when the random signs $g_{i}(k) \in \{1,-1\}$ are fully independent. They demonstrate that in this setting, if $y \in \R^n$ is the count-sketch estimate of a stream vector $f$ with $k$ columns and $d$ rows, then for any $t \leq d$ and index $i \in [n]$ we have:
\[ \bpr{(f_i - y_i)^2 > \frac{t  }{d} \frac{\|f_{\text{tail}(k)  )} \|_2^2 }{k}} \leq 2e^{-\Omega(t)}\]
In order to apply this algorithm in $o(n)$ space, however, one must first derandomize it from using fully independent random signs.
To the best of the authors knowledge, the best known derandomization procedure known before was a black-box application of Nisan's which results in $O(\eps^{-2}\log^3(n))$-bits of space. For the purposes of the theorem, we replace the notation $1/\eps^2$ with $k$ (the number of columns of count-sketch up to a constant).

\begin{theorem}
	The count-sketch variant of \cite{minton2014improved} can be implemented so that if $A \in \R^{d \times k}$ is a count-sketch table, then for any $t \leq d$ and index $i \in [n]$ we have:
	\[ \bpr{(f_i - y_i)^2 > \frac{t  }{d} \frac{\|f_{\text{tail}(k)  )} \|_2^2 }{k}} \leq 2e^{-\Omega(t)}\]
	and such that the total space required is $O(k d \log(n) (\log \log n )^2)$.
\end{theorem}
\begin{proof}
	We first remark that the following modifcation to the count-sketch procedure does not effect the analysis of \cite{minton2014improved}. Let $A \in \R^{d \times k}$ be a $d \times k$ count-sketch matrix. The modification is as follows: instead of each variable $h_{i}(k)$ being uniformly distributed in $\{1,2,\dots,k\}$, we replace them with variables $h_{i,j,k} \in \{0,1\}$ for $(i,j,k) \in [d] \times [k] \times [n]$, such that $h_{i,j,k}$ are all i.i.d. and equal to $1$ with probability $1/k$. We also let $h_{i,h,k} \in \{1,-1\}$ be i.i.d. Rademacher variables ($1$ with probability $1/2$). Then $A_{i,j} = \sum_{k=1}^n f_k g_{i,j,k} h_{i,j,k}$, and the estimate $y_t$ of $f_t$ for $t \in [n]$ is given by:
	\[ y_t = \text{median} \{	g_{i,j,t} A_{i,j} \; | \; h_{i,j,t = 1}	\}\]
	Thus the element $f_t$ can be hashed into multiple buckets in the same row of $A$, or even be hashed into none of the buckets in a given row. By Chernoff bounds, $|\{	g_{i,j,t} A_{i,j} \; | \; h_{i,j,t = 1}	\}| = \Theta(d)$ with high probability for all $t \in [n]$. Observe that the marginal distribution of each bucket is the same as the count-sketch used in \cite{minton2014improved}, and moreover seperate buckets are fully independent. The key property used in the analysis of \cite{minton2014improved} is that the final estimator is a median over estimators whose error is independent and symmetric, and therefore the bounds stated in the theorem still hold after this modification \cite{Pricepersonal}.
	
	Given this, the entire sketch stored by the streaming algorithm is $B \cdot f$, where 
$$B_{j} = \begin{cases}
	1 & \text{ with prob } \frac{1}{2k}\\
		-1 & \text{ with prob } \frac{1}{2k}\\
		0 & \text{ otherwise}
	\end{cases}$$
	Thus the entries of $B$ are i.i.d., and can be sampled with $O(\log(k)) \leq O(\log(n))$ bits, and $\texttt{vec}(A) = B \cdot f$, where $\texttt{vec}(A)$ is the vectorization of the count-sketch table $A$. Here $B \in \R^{dk \times n}$. 
	
	Now note that for a fixed $i$, to test the statement that $(f_i - y_i)^2 > \frac{t  }{d} \frac{\|f_{\text{tail}(k)  )} \|_2^2 }{k}$, one needs to know both the value of the sketch $Bf$, in addition to the value of the $i$-th column of $B$, since the estimate can be written as $y_i = \text{median}_{j \in [kd], B_{j,i} \neq 0} \{B_{j,i} \cdot (Bf)_j	\}$. Note that the $i$-th column of $B$ (which has $kd$ entries) can simply be written as a sketch of the form $X \cdot \ttx{vec}(B)$, where $X \in \R^{kd \times dkn}$ is a fixed matrix such that $X \cdot \ttx{vec}(B) = B_i$, so we also need to store $X \cdot \ttx{vec}(B)$ 
	Thus by Theorem \ref{thm:derandomGeneral}, the algorithm can be derandomized to use $O(k d \log(n) (\log \log n )^2)$ bits of space, and such that for any $t \leq d$ and any $i \in [n]$ we have $ \pr{(f_i - y_i)^2 > \frac{t  }{d} \frac{\|f_{\text{tail}(k)  )} \|_2^2 }{k}} \leq 2e^{-\Omega(t)} \pm n^{-\Omega(dk)}$. 
	
\end{proof}

\paragraph{The Derandomization.} 
 We now introduce the notation which will be used in our derandomization. Our $L_p$ sampler uses two sources of randomness which we must construct PRG's for. The first, $r_e$, is the randomness needed to construct the in exponential random variables $t_i$, and the second, $r_c$, is the randomness needed for the fully random hash functions and signs used in count-max. Note that $r_e,r_c$ both require $\poly(n)$ bits by Lemma \ref{lem:runtime}. From here on, we will fix any index $i \in [n]$. Our $L_p$ sampler can then be thought of as a tester $\mathcal{A}(r_e,r_c) \in  \{0,1\}$, which tests on inputs $r_e,r_c$, whether the algorithm will output $i \in [n]$. Let $G_1(x)$ be Nisan's PRG, and let $G_2(y)$ be the half-space PRG. For two values $b,c \in \R$, we write $a \sim_\eps b$ to denote $|a-b| < \eps$. Our goal is to show that 
\[	\mathbf{Pr}_{r_e,r_c}\big[ \mathcal{A}(r_e,r_c)  \big] \sim_{n^{-c}} \mathbf{Pr}_{x,y} \big[  \mathcal{A}(G_2(y),G_1(x)) \big] \]
where $x,y$ are seeds of length $O(\log^2(n))$, and $c$ is an arbitrarily large constant.  

\begin{theorem}\label{thm:derandom}
	A single instance of the algorithm \texttt{$L_p$ Sampler} using \texttt{Fast-Update} as its update procedure can be derandomized using a random seed of length $O(\log^2(n)(\log\log n)^2)$, and thus can be implemented in this space. Moreover, this does not affect the time complexity as stated in Lemma \ref{lem:runtime}.
\end{theorem}
\begin{proof}
	First note that by Lemma \ref{lem:runtime}, we require $\tilde{O}(\nu^{-1})$ random bits for each $i \in [n]$, and thus we require a total of $\tilde{O}(n\nu^{-1}) = \poly(n)$ random bits to be generated. 
	Since Nisan's PRG requires the tester to read its random input in a stream, we can use a standard reordering trick of the elements of the stream, so that all the updates to a given coordinate $i \in [n]$ occur at the same time (see \cite{indyk2006stable}). This does not effect the output distribution of our algorithm, since linear sketches do not depend on the ordering of the stream. Now let $c'$ be the constant such that the algorithm \texttt{$L_p$ Sampler} duplicates coordinates $n^{c'}$ times. In other words, the count-max is run on the stream vector $F \in \R^{n^{c'}}$, and let $N = n^{c'}$. Now, as above, we fix any index $i \in [N]$, and attempt to fool the tester which checks if, on a given random string, our algorithm would output $i$. For any fixed randomness $r_e$ for the exponentials, let $\mathcal{A}_{r_e}(r_c)$ be the tester which tests if our $L_p$ sampler would output the index $i$, where now the bits $r_e$ are hard-coded into the tester, and the random bits $r_c$ are taken as input and read in a stream. We first claim that this tester can be implemented in $O(\log(n))$-space. 
	
	To see this, note that $\mathcal{A}_{r_e}(r_c)$ must simply count the number of rows of count-max such that item $i$ is hashed into the largest bucket (in absolute value) of that row, and output $1$ if this number is at least $\frac{4d}{5}$, where $d$ is the number of rows in count-max. To do this, $\mathcal{A}_{r_e}(r_c)$ can break $r_c$ into $d$ blocks of randomness, where the $j$-th block is used only for the $j$-th row of count-max. It can then fully construct the values of the counters in a row, one row at a time, reading the bits of $r_c$ in a stream. To build a bucket, it looks at the first element of the stream, uses $r_c$ to find the bucket it hashes to and the Gaussian scaling it gets, then adds this value to that bucket, and then continues with the next element. Note that since $r_e$ is hardcoded into the tester, we can assume the entire stream vector $\zeta$ is hardcoded into the tester. Once it constructs a row of count-max, it checks if $i$ is in the largest bucket by absolute value, and increments a $O(\log(d))$-bit counter if so. Note that it can determine which bucket $i$ hashes to in this row while reading off the block of randomness corresponding to that row. 	
	Then, it throws out the values of this row and the index of the bucket $i$ hashed to in this row, and builds the next row. Since each row has $O(1)$ buckets, $\mathcal{A}_{r_e}(r_c)$ only uses $O(\log(n))$-bits of space at a time. Then using $G_1(x)$ as Nisan's generator with a random seed $x$ of length $O(\log^2(n))$-bits, we have $\pr{\mathcal{A}_{r_e}(r_c)} \sim_{n^{-c_0}} \pr{\mathcal{A}_{r_e}(G_1(x))}$, where the constant $c_0$ is chosen to be sufficiently larger than the constant $c_1$ in the $n^{-c_1}$ additive error of our perfect sampler, as well as the constant $c'$  Moreover:
	
	 \begin{equation*}
	\begin{split}
		\bpr{\mathcal{A}(r_e,r_c)}  & = 	\sum_{r_e}  \bpr{\mathcal{A}_{r_e}(r_c)}  \bpr{r_e}  \\
		& =\sum_{r_e} ( (\bpr{\mathcal{A}_{r_e}(G_1(x))} \pm n^{-c_0}) \bpr{r_e} \\
		& =\sum_{r_e}  (\bpr{\mathcal{A}_{r_e}(G_1(x))}\bpr{r_e} \pm \sum_{r_e} n^{-c_0}  \bpr{r_e}  \\ 
		& \sim_{n^{-c_0}} \bpr{\mathcal{A}(r_e,G_1(x))} \\		
		\end{split}
	\end{equation*}
	
	Now fix any Nisan seed $x$, and consider the tester $\mathcal{A}_{G_1(x)}(r_e)$, which on fixed count-max randomness $G_1(x)$, tests if the algorithm will output $i \in [n]$ on the random input $r_e$ for the exponential variables. We first observe that it seems unlikely that $\mathcal{A}_{G_1(x)}(r_e)$ can be implemented in $\log(n)$ space while reading its random bits $r_e$ in a stream. This is because each row of count-max depends on the same random bits in $r_e$ used to construct the exponentials $t_i$, thus it seems $\mathcal{A}_{G_1(x)}(r_e)$ would need to store all $\log^2(n)$ bits of count-max at once. However, we will now demonstrate that $\mathcal{A}_{G_1(x)}(r_e)$ is in fact a $\poly(n)$ bounded $O(d)$-halfspace tester (as defined earlier in this section) where $d$ is the number of rows of count-max, and therefore can be derandomized with the PRG of \cite{gopalan2015pseudorandomness}.   By the Runtime \& Random bits analysis in Lemma \ref{lem:runtime}, it suffices to take all random variables in the algorithm to be $O(\log(n))$-bit rational numbers. Scaling by a sufficiently large $\poly(n)$, we can assume that $1/t_j^{1/p}$ is a discrete distribution supported on $\{-T,\dots,T\}$ where $T \leq \poly(n)$ for a sufficiently large $\poly(n)$. We can then remove all values in the support which occur with probability less than $\poly(n)$, which only adds an a $n^{-c_0}$ additive error to our sampler. After this, the distribution can be sampled from with $\poly(T) = \poly(n)$ random bits, which is as needed for the setting of Lemma \ref{lemma:halfspace}. 	
	 Note that we can also apply this scaling the Gaussians in count-max, so that they too are integers of magnitude at most $\poly(n)$.

	 Given this, the distribution of the variables $1/t_j^{1/p}$ satisfy the conditions of Lemma $\ref{lemma:halfspace}$, in particular being $\poly(n)$-bounded, thus we must now show that 
	$\mathcal{A}_{G_1(x)}(r_e)$ is indeed a $O(d)$-halfspace tester, with integer valued half-spaces bounded by $\poly(n)$. First consider a given row of count-max, and let the buckets be $B_1,\dots,B_{10}$. WLOG $i$ hashs into $B_1$, and we must check if $|B_1| > |B_t|$ for $t =2,3,\dots,10$. Let $g_j$ be the random count-max signs (as specified by $G_1(x)$), and let $S,S_t$ be the set of indices which hash to $B_1$ and $B_t$ respectively. We can run the following $6$ half-space tests to test if $|B_1| > |B_t|$:
	\begin{equation}\label{test1}
	\sum_{j \in S} g_j f_j (\frac{1}{t_j^{1/p}})  > 0 
	\end{equation}
		\begin{equation}\label{test2}
	\sum_{j \in S_t} g_j f_j (\frac{1}{t_j^{1/p}})  > 0 
	\end{equation}	
	\begin{equation}\label{test3}
	a_1 \sum_{j \in S} g_j f_j (\frac{1}{t_j^{1/p}}) + a_2 \sum_{j \in S_t} g_j f_j (\frac{1}{t_j^{1/p}}) > 0
	\end{equation}	
 	where $a_1,a_2$ range over all values in $\{1,-1\}^2$. The tester can decide whether $|B_1| > |B_t|$ by letting $a_1$ be the truth value (where $-1$ is taken as fail) of the first test \ref{test1} and $a_2$ the truth value of \ref{test2}. It then lets $b_t \in \{0,1\}$ be the truth value of \ref{test3} on the resulting $a_1,a_2$ values, and it can correctly declare $|B_1| > |B_t|$ iff $b_t = 1$. Thus for each of the $9$ pairs $B,B_t$, the tester uses $6$ halfspace testers to determine if $|B_1| > |B_t|$, and so can determine if $i$ hashed to the max bucket with $O(1)$ halfspace tests. So $\mathcal{A}_{G_1(x)}(r_e)$ can test if the algorithm will output $i$ by testing if $i$ hashed to the max bucket in a $4/5$ fraction of the $d$ rows of count-max, using $O(d) = O(\log(n))$ halfspace tests. Note that by the scaling performed in the prior paragraphs, all coefficents of these half-spaces are integers of magnitude at most $\poly(n)$. So by Lemma \ref{lemma:halfspace}, the PRG $G_2(y)$ of \cite{gopalan2015pseudorandomness} fools $\mathcal{A}_{G_1(x)}(r_e)$ with a seed $y$ of $O(\log^2(n) (\log\log n)^2)$-bits. So $\pr{\mathcal{A}_{G_1(x)}(r_e)} \sim_{n^{-c_0}}\pr{ \mathcal{A}_{G_1(x)}(G_2(y))}$, and so by the same averaging argument as used in for the Nisan PRG above, we have $\pr{\mathcal{A}(r_e,G_1(x)) }\sim_{n^{-c_0}} \pr{ \mathcal{A}(G_2(y),G_1(x))}$, and so $\pr{\mathcal{A}(r_e, r_c) }\sim_{n^{-c_0}} \pr{ \mathcal{A}(G_2(y),G_1(x))}$ as desired. Now fixing any $i \in [n]$, 
 	let $\mathcal{A}'( r_e,r_c)$ be the event that the overall algorithm outputs the index $i$. In other words, $\mathcal{A}_i'( r_e,r_c) = 1$ if $\mathcal{A}_{i_j}( r_e,r_c) = 1$ for some $j \in [n^{c'-1}]$, where $\mathcal{A}_{i_j}( r_e,r_c) = 1$ is the event that count-max declares that $i_j$ is the maximum in Algorithm \texttt{$L_p$ Sampler}. Thus, 
 	the probability that the algorithm outputs a non-duplicated coordinate $i \in [n]$ is given by:
 	\begin{equation}
 	\begin{split}
 	\bpr{\mathcal{A}_i'( r_e,r_c) } &= \sum_{j=1}^{n^{c'}} \bpr{\mathcal{A}_{i_j}( r_e,r_c)}\\
 	&= \sum_{j=1}^{n^{c'}} \bpr{\mathcal{A}_{i_j}(G_2(y),G_1(x))} \pm n^{-c_0}\\
 		&=  \bpr{\mathcal{A}_{i}'(G_2(y),G_1(x))} \pm n^{-c_1}\\
 	\end{split}
 	\end{equation}
 	where in the last line we set $c_0 > c' + c_1$, where recall $c_1$ is the desired additive error in our main sampler. In conclusion, replacing the count-max randomness with Nisan's PRG and the exponential random variable randomness with the half-space PRG $G_2(y)$, we can fool the algorithm which tests the output of our algorithm with a total seed length of $O(\log^2(n)(\log\log n)^2)$.
 	
 	 To show that the stated update time of Lemma \ref{lem:runtime} is not affected, we first remark that Nisan's PRG simply involves performing $O(\log(n))$ nested hash computations on a string of length $O(\log(n))$ in order to obtain any arbitrary substring of $O(\log(n))$ bits. Thus the runtime of such a procedure is $\tilde{O}(1)$ to obtain the randomness needed in each update of a coordinate $i \in [n^c]$. By Lemma \ref{lemma:halfspace}, the PRG of \cite{gopalan2015pseudorandomness} requires $\tilde{O}(1)$ time to sample the $O(\log(n))$-bit string needed to generate an exponential, and moreover can be computed with working space linear in the size of the random seed (note that this is also true of Nisan's PRG, which just involves $O(\log(n))$-nested hash function computations). Thus the update time is only blown up by a $\tilde{O}(1)$ factor, which completes the proof. 
\end{proof}

\begin{corollary}\label{cor:derandom}
	For $p=2$, the entire algorithm can be derandomized to run using $O(\log^3(n) \log(1/\delta))$-bits of space with failure probability of $\delta$. For $p < 2$, the algorithm can be derandomized to run using $O(\log^3(n))$-bits of space with $\delta = 1/\poly(n)$.
\end{corollary}
\begin{proof}
	We can simply derandomize a single instance of our sampling algorithm using Nisan's PRG as in Theorem \ref{thm:derandom}, except that we derandomize all the randomness in the algorithm at once. Since such an instance requires $O(\log^2(n))$-bits of space, using Nisan's blows up the complexity to $O(\log^3(n))$ (the tester can simply simulate our entire algorithm in $O(\log^2(n))$-bits of space, reading the randomness in a stream by the reordering trick of \cite{indyk2006stable}). Since the randomness for separate parallel instances of the main sampling algorithm is disjoint and independent, this same $O(\log^2(n))$-bit tester can test the entire output of the algorithm by testing each parallel instance one by one, and terminating on the first instance that returns an index $i \in [n]$. Thus the same $O(\log^3(n))$-bit random seed can be used to randomize all parallel instances of our algorithm. For $p<2$, we can run $O(\log(n))$ parallel instances to get $1/\poly(n)$ failure probability in $O(\log^3(n))$-bits of space as stated. For $p=2$, we can run $O(\log(n)\log(1/\delta))$ parallel repetitions needed to get $\delta$ failure probability using the same random string, for a total space of $O(\log^3(n) \log(1/\delta) + \log^3(n)) = O(\log^3(n) \log(1/\delta))$ as stated. As noted in the proof of Theorem \ref{thm:derandom}, computing a substring of $O(\log(n))$-bits from Nisan's PRG can be done in $\tilde{O}(1)$ time and using space linear in the seed length, which completes the proof. 
\end{proof}

\subsection{Query Time} \label{subsec:query}
We will now show the modifications to our algorithm necessary to obtain $\tilde{O}(1)$ query time. Recall that our algorithm maintains a count-max matrix $A$. Our algorithm then searches over all indices $i \in \mathcal{K}$ to check if $i$ hashed into the maximum bucket in a row of $A$ at least a $4/5$ fraction of the time. Since $|\mathcal{K}| = \tilde{O}(n)$, running this procedure requires $ \tilde{O}(n)$ time to produce an output on a given query. To avoid this and obtain $\tilde{O}(1)$ running time, we will utilize the heavy hitters algorithm of \cite{larsen2016heavy}, which has $\tilde{O}(1)$ update and query time, and which does not increase the complexity of our algorithm.  

\begin{theorem}[\cite{larsen2016heavy}]\label{thm:expander}
	For any precision parameter $0< \eps < 1/2$, given a general turnstile stream $x \in \R^n$ there is an algorithm, \texttt{ExpanderSketch}, which with probability $1 - n^{-c}$ for any constant $c$, returns a set $S \subset [n]$ of size $S = O(\eps^{-2})$ which contains all indices $i$ such that $|x_i| \geq \eps \|x\|_2$. The update time is $O(\log(n))$, the query time is $\tilde{O}(\eps^{-2})$, and the space required is $O(\eps^{-2}\log^2(n))$-bits.
\end{theorem}

\paragraph{Using \texttt{ExpanderSketch} to speed up query time.} The modifications to our main algorithm \texttt{$L_p$ Sampler} with \texttt{Fast-Update} are as follows. We run our main algorithm as before, maintaining the same count-max data structures $A$. Upon initialization of our algorithm, we also initialize an instance \texttt{ExSk} of \texttt{ExpanderSketch} as in Theorem \ref{thm:expander}, with the precision parameter $\eps = 1/100$.

Now recall in our \texttt{Fast-Update} procedure, for each $i \in [n]$ we hash the top $K_i = O(\log(n))$ largest duplicates $\zeta_{i_j}$ corresponding to $f_i$ individually, and store the random variables $h_\ell(i_j)$ that determine which buckets in $A$ they hash to. While processing updates to our algorithm at this point, we make the modification of \textit{additionally} sending these top $K_i$ items to \texttt{ExSk} to be sketched. More formally, we run \texttt{ExSk} on the stream $\zeta_{\mathcal{K}}$, where $\zeta_{\mathcal{K}}$ is the vector $\zeta$ projected onto the coordinates of $\mathcal{K}$. Since $K_i = \tilde{O}(1)$, this requires making $\tilde{O}(1)$ calls to update \texttt{ExSk} on different coordinates, which only increases our update time by an  $\tilde{O}(1)$ additive term.

On termination, we obtain the set $S$ containing all items $\zeta_i$ such that $i \in \mathcal{K}$ and $\zeta_i \geq (1/100) \|\zeta_{\mathcal{K}}\|_2$. Instead of searching through all coordinates of $\mathcal{K}$ to find a maximizer, we simply search through the coordinates in $S$, which takes $\tilde{O}(|S|) = \tilde{O}(1)$ time. 
 We now argue that the output of our algorithm does not change with these new modifications.  We refer collectively to the new algorithm with these modifications as \texttt{$L_p$ Sampler} with \texttt{Fast-Update} and \texttt{ExSk}, and the algorithm of Section \ref{subsec:update} as simply \texttt{$L_p$ Sampler} with \texttt{Fast-Update}.

\begin{lemma}\label{lem:query}
For any constant $c > 0$, with probability $1 - n^{-100c}$ the algorithm \texttt{$L_p$ Sampler} with \texttt{Fast-Update} and \texttt{ExSk} as described in this section returns the same output (an index $i \in [n]$ or \ttx{FAIL}) as \texttt{$L_p$ Sampler} using \texttt{Fast-Update} but without \texttt{ExSk}. The space and update time are not increased by using \texttt{ExSk}, and the query time is now $\tilde{O}(1)$. 
\end{lemma}
\begin{proof}
	We condition on the event that $S$ contains all items $i$ such that $i \in \mathcal{K}$ and $|\zeta_i| \geq 1/100 \|\zeta_{\mathcal{K}}\|_2$, which occurs with probability $1 - n^{-100c}$ by Theorem \ref{thm:expander}. Since \texttt{$L_p$ Sampler} already uses at least $O(\log^2(n))$ bits of space, the additional $O(\log^2(n))$ bits of overhead required to run an instance \texttt{ExSk} of \texttt{ExpanderSketch} with sensitivity parameter $\eps = 1/100$ does not increase the space complexity. Furthermore, as mentioned above, the update time is blown-up by a factor of $\tilde{O}(1)$, since we make $K_i = \tilde{O}(1)$ calls to update \texttt{ExSk}, which has an update time of $\tilde{O}(1)$ by Theorem \ref{thm:expander}. Furthermore, our algorithm does not require any more random bits, as it only uses \texttt{ExpanderSketch} as a subroutine, and thus no further derandomization is required. Thus the complexity guarantees of Lemma \ref{lem:runtime} are unchanged. For the query time, we note that obtaining $S$ requires $\tilde{O}(1)$ time (again by Theorem \ref{thm:expander}), and querying each of the $|S| = O(1)$ items in our count-max $A$ requires $\tilde{O}(1)$ time. To complete the proof, we now consider the output of our algorithm. Since we are searching through a strict subset $S \subset [n^c]$, it suffices to show that if the original algorithm output an $i_j \in [n^c]$, then so will we. As argued in Lemma \ref{lem:runtime}, such a coordinate must be contained in $\mathcal{K}$. By Corollary \ref{cor:countmax}, we must have $|\zeta_{i_j}| > \frac{1}{100} \|\zeta\|_2 \geq  \frac{1}{100} \|\zeta_{\mathcal{K}}\|_2$ with probability $1-n^{-100c}$ (scaling $c$ by $100$ here), thus $i_j \in S$, which completes the proof.

\end{proof}

\section{Estimating the Frequency of the Coordinate Output}\label{sec:estimate}	
In this section, we will show how, conditioned on our algorithm \texttt{$L_p$ Sampler} returning an index $i \in [n]$, we can obtain an estimate $\tilde{f}_i = (1\pm \eps)f_i$ with probability $1-\delta_2$. We now describe how to do this. Our algorithm, in addition to the count-max matrix $A$ used by \texttt{$L_p$ Sampler}, stores a count-sketch matrix $A'$ with, $d' = O(\log(1/\delta_2))$ rows and $O(\gamma) = O(\min\big\{\eps^{-2}, \:  \eps^{-p}\log\big(\frac{1}{\delta_2}\big) \} \big\} )$ columns. Recall in our \texttt{Fast-Update} procedure, for each $i \in [n]$ we hash the top $K_i = O(\log(n))$ largest duplicates $\zeta_{i_j}$ corresponding to $f_i$ individually into $A$, and store the random variables $h_\ell(i_j)$ that determine which buckets in $A$ they hash to. Thus if count-max outputs an $i_j \in [n^c]$ we know that $i_j \in \mathcal{K}$, where  $\mathcal{K} = \cup_{i \in [n]} \cup_{j = 1}^{K_i} \{i_j\}$  as in Section \ref{sec:runtime} (since our algorithm only searches through $\mathcal{K}$ to find a maximizer). Thus it suffices to run the count-sketch instance $A'$ on the stream $\zeta_{\mathcal{K}}$, where $\zeta_{\mathcal{K}}$ is the vector $\zeta$ with the coordinates not in $\mathcal{K}$ set to $0$. Since $K_i = \tilde{O}(1)$, we perform at most $\tilde{O}(1)$ updates to count-sketch at every step in the stream. This requires making $\tilde{O}(1)$ calls to update count-sketch on each stream update, which only increases our update time by an $\tilde{O}(1)$ additive term.

Now if \texttt{$L_p$ Sampler} returns $i_j \in [n^c]$ (corresponding to some duplicate $i_j$ of $i$), then we must have $i_j \in \mathcal{K}$. Thus we can query $A'$ for a value $\tilde{y}_{i_j}$ such that $|\tilde{y}_{i_j} - \zeta_{i_j}| < \sqrt{1/\gamma} \|\zeta_{\text{tail}(1/\gamma) } \|_2$ with probability $1 - \delta_2$ by Theorem \ref{thm:count-sketch}. Furthermore, since  $i_j \in \mathcal{K}$, we can compute the value $I_k$ such that $I_k = (\texttt{rnd}_\nu(1/t_{i_j}))$ by simulating the \texttt{Fast-Update} procedure on an update to $i$.  We will argue that the estimate  $\tilde{f}= \tilde{y}_{i_j} (\texttt{rnd}_\nu(1/t_{i_j}))^{-1}$ satisfies $\tilde{f}=  (1 \pm \eps) f_{i}$.
Putting this together with Theorem \ref{thm:main}, we will obtain the following result.
\begin{theorem}
	There is an algorithm $\mathcal{A}$ which, on a general turnstile stream $f$, outputs $i \in [n]$ with probability $|f_i|^p/\|f\|_p^p(1 \pm \nu) +O(n^{-c})$, and outputs \ttx{FAIL} with probability at most $\delta_1$. Conditioned on outputting some $i \in [n]$, $\mathcal{A}$ will then output $\tilde{f}$ such that $\tilde{f} = (1 \pm \eps)f_i$ with probability $1 - \delta_2$. The space required is $O\big(\big(\log^2(n)(\log\log n)^2 + \beta \log(n) \log(1/\delta_2)\big) \log(1/\delta_1) \big)$ for $p \in (0,2)$, and $O\big(\big(\log^3(n) + \eps^{-2}\log^2(n) \log(1/\delta_2)\big) \log(1/\delta_1) \big)$ for $p=2$, where $\beta =\min\big\{\eps^{-2}, \: \eps^{-p}\log\big(\frac{1}{\delta_2} \big) \} \big\}$ . The update time is $\tilde{O}(\nu^{-1})$ and the query time is $\tilde{O}(1)$
\end{theorem}
\begin{proof}
We first consider the complexity. The first term in each of the upper bounds follows from Theorem \ref{thm:main}, as well as the $\log(1/\delta_1)$ term which comes from repeating the entire algorithm $\log(1/\delta_1)$ times for $p<2$, and $\log(n)\log(1/\delta_1)$ times for $p=2$. The second term in the space bound results from storing the $d' \times \gamma$ count-sketch table $A'$, which is $O(\gamma \log(n) \log(1/\delta_2))$ as stated. Moreover, the update time for the new data structure is at most $\tilde{O}(1)$, since the only additional work we do on each update is to hash $K_i = O(\log(n))$ items into $d' = O(\log(n))$ rows of $A'$. Furthermore, the query time just requires computing a median of $O(\log(n))$ entries of $A'$. Each of these actions is $\tilde{O}(1)$ time in the unit cost RAM model, so the additional update and query time is $\tilde{O}(1)$. The remaining $\tilde{O}(\nu)$ update time follows from Lemma \ref{lem:runtime}. 
	
	 For correctness, note that if \texttt{$L_p$ Sampler} does not fail and instead outputs $i_j \in [n^c]$, we know that $|\zeta_{i_j}| > 1/100 \|\zeta\|_2$. Furthermore, we have $|\tilde{y}_{i_j} - \zeta_{i_j}| < \sqrt{1/\gamma} \|\zeta_{\text{tail}(1/\gamma)} \|_2 \leq\sqrt{1/\gamma} \|\zeta\|_2$ with probability $1-\delta_2$, so setting $\gamma = \Theta(1/\eps^2)$ sufficiently large, it follows that $\tilde{y}_{i_j} = (1 \pm O(\eps)) \zeta_{i_j}$.  Then $\tilde{y}_{i_j} (\texttt{rnd}_\nu(1/t_{i_j}))^{-1} =(1 \pm \eps) f_{i}$ follows immediately from the fact that $f_i = \zeta_{i_j} (\texttt{rnd}_\nu(1/t_{i_j}))^{-1}$ (and a rescaling of $\eps$ by a constant). This shows that $O(\eps^{-2})$ bits is always an upper bound for the value of $\gamma = \Theta(\beta)$ needed for $p \in (0,2]$.

	To show the other upper bound in the definition of $\beta$ (for cases when $p<2$), 
	first define $T_\gamma \subset [n^c]$ as the set of $n^c - \gamma$ smallest coordinates (in absolute value) of $z$. In other words $z_{T_\gamma} = z_{\text{tail}(\gamma)}$, where for any set $S \subseteq [n^c]$ $z_S$ denotes $z$ projected onto the coordinates of $S$. Note that if $S$ is any set of size $n^c - s$ and $v \in \R^{n^c}$ any vector, we have $\|v_{\text{tail}(s)} \|_2 \leq \|v_{S}\|_2 $. 
	Then by Proposition \ref{prop:1}, using the fact that $\zeta_i = (1 \pm O(\nu))z_i$ for all $i \in [n^c]$, we have $\|\zeta_{\text{tail}(\gamma)} \|_2 \leq \|\zeta_{T_\gamma}\|_2 \leq  2 \|z_{\text{tail}(\gamma)}\|_2 = O(\|F\|_p (\gamma)^{-1/p + 1/2})$ for $p < 2$ with probability $1 - O(e^{-\gamma}) > 1 - \delta_2$, where now we are setting $\gamma = \Theta(\max\{\eps^{-p}, \log( 1/\delta_2)	\})$. Condition on this now. Then we obtain error error $|\tilde{y}_{i_j} - \zeta_{i_j}| < \sqrt{1/\gamma} \|\zeta_{\text{tail}(1/\gamma)} \|_2  = O(\|F\|_p \gamma^{-1/p}) = O(\eps (\log(1/\delta_2))^{-1/p} \|F\|_p )$ from our second count-sketch $A'$. Now $z_{D(1)} = \|F\|_p/E_1^{1/p}$, which is at least $\Omega(\|F\|_p/(\log(1/\delta_2))^{1/p})$ with probability greater than $1-\delta_2$ using the pdf of an exponential. Conditioned on this, the error from our second count-sketch $A'$ gives, in fact, a $(1 \pm O(\eps))$ relative error approximation  of $\zeta_{i_j}$, which is the desired result. Note that we conditioned only on our count-sketch giving the desired $|\tilde{y}_{i_j} - \zeta_{i_j}| < \sqrt{1/\gamma} \|\zeta_{\text{tail}(1/\gamma)} \|_2 $ error, on 
	$\|z_{\text{tail}(\gamma)}\|_2 = O(\|F\|_p (\gamma)^{-1/p + 1/2})$, and on $E_1 = O(\log(1/\delta))$, each of which holds with probability at least $1-O(\delta_2)$, so the Theorem follows after a union bound.

\end{proof}

	\section{Lower bounds}
	In this section, we obtain a lower bound for providing relative error approximations of the frequency of a sampled item. 
		Our lower bound is derived from one-way two-party communication complexity. Let $\mathcal{X},\mathcal{Y}$ be input domains to a two party communication complexity problem. Alice is given $x \in \mathcal{X}$ and Bob is given $y \in \mathcal{Y}$. Their goal is to solve some relational problem $Q \subseteq \mathcal{X} \times \mathcal{Y} \times \mathcal{O}$, where for each $(x,y) \in \mathcal{X}\times \mathcal{Y}$ the set $Q_{xy} = \{z | (x,y,z ) \in Q	\}$ represents the set of \textit{correct} solutions to the communication problem. 
	
	In the one-way \textit{communication protocol} $\mathcal{P}$, Alice must send a single message $M$ to Bob (depending on her input $X$), from which Bob must output an answer in $o \in \mathcal{O}$ depending on his input $Y$ and the message $M$. The maximum possible length (in bits) of $M$ over all inputs $(x,y) \in \mathcal{X} \times \mathcal{Y}$ is the \textit{communication cost} of the protocol $\mathcal{P}$. Communication protocols are allowed to be randomized, where each player has \textit{private} access to an unlimited supply of random bits. The protocol $\mathcal{P}$ is said to \textit{solve} the communication problem $Q$ if Bob's output $o$ belongs to $Q_{xy}$ with failure probability at most $\delta < 1/2$. The \textit{one-way communication complexity} of $Q$, denoted $R^{\to}_\delta(Q)$, is the minimum communication cost of a protocol which solves the protocol $Q$ with failure probability $\delta$.

	Now a similar measure of complexity is the \textit{distributional complexity} $D^\to_{\mu ,\delta}(Q)$, where $\mu$ is a distribution over $\mathcal{X} \times \mathcal{Y}$, which denotes the minimum communication cost of the best deterministic protocol of $Q$ with failure probability at most $\delta$ when the inputs $(x,y) \sim \mu$. By Yao's Lemma, we have that $R^{\to}_\delta(Q) = \max_\mu D^\to_{\mu,\delta}(Q)$. We first review some basic facts about entropy and mutual information (see Chapter 2 of \cite{cover2012elements} for proofs of these facts). 
	\begin{proposition} \label{prop:info}$ $\\
		\begin{enumerate}			
			\item Entropy Span: if $X$ takes on at most $s$ values, then $0 \leq H(X) \leq \log s$
			\item $I(X:Y) := H(X) - H(X|Y) \geq 0$, that is $H(X|Y) \leq H(X)$
			\item Chain Rule: $I(X_1,X_2,\dots,X_n :Y|Z) = \sum_{i=1}^n I(X_i:Y|X_1,\dots,X_{i-1},Z)$
			\item Subadditivity: $H(X,Y|Z) \leq H(X|Z) + H(Y |Z)$ and equality holds if and only if $X$ and $Y$ are independent conditioned on $Z$
			\item Fano's Inequality: Let $M$ be a predictor of $X$. In other words, there exists a function $g$ such that $\pr{g(M) = X} > 1-\delta$ where $\delta < 1/2$. Let $\mathcal{U}$ denote the support of $X$, where $\mathcal{U} \geq 2$. Then $H(X|M) \leq \delta \log(|\mathcal{U}| - 1) + h_2(\delta)$, where $h_2(\delta) :=\delta \log(\delta^{-1}) + (1-\delta)\log(\frac{1}{1-\delta})$ is the binary entropy function.
		\end{enumerate}
	\end{proposition}

We now define the information cost of a protocol $\mathcal{P}$:
	
	\begin{definition}
		Let $\mu$ be a distribution of the input domain $\mathcal{X} \times \mathcal{Y}$ to a communication problem $Q$. Suppose the unputs $(X,Y)$ are choosen according to $\mu$, and let $M$ be Alice's message to Bob, interpreted as a random variable which is a function of $X$ and Alice's private coins. Then the \textit{information cost }of a protocol $\mathcal{P}$ for $Q$ is defined as $I(X:M)$. 
		
		The \textit{one-way information complexity} of $Q$ with respect to $\mu$ and $\delta$, denoted by $\text{IC}^\to_{\mu,\delta}(Q)$, is the minimum information cost of a one-way protocol under $\mu$ that solves $Q$ with failure probability at most $\delta$.
	\end{definition}

	Note that by Proposition \ref{prop:info}, we have
	\[	I(X:M) = H(M) - H(M|X)\leq H(M) \leq |M| 	\]
	where $|M|$ is the length of the message $M$ in bits. This results in the following proposition.
	\begin{proposition}
		For every probability distribution $\mu$ on inputs,
		\[	R^\to_\delta(Q) \geq IC^\to_{\mu, \delta}(Q)	\]
	\end{proposition}
	
	\subsection{Augmented Indexing on Large Domains}

	We now introduce the following communication problem, known as Augmented Index on Large Domains. Our communication problem is derived from the communication problem (of the same name) introduced in \cite{jayram2013optimal}, but we modify the guarantee of the output required so that constant probability of error is allowed. The problem is as follows.
	
	\begin{definition}
			  Let $\mathcal{U}$ be an alphabet with $|\mathcal{U}| = k \geq 2$. Alice is given a string $x \in \mathcal{U}^d$, and Bob is given $i\in [d]$ along with the values $x_{i+1}, x_{i+2},\dots, x_d$. Alice must send a message $M$ to Bob, and then Bob must output the value $x_{i} \in \mathcal{U}$ with probability $3/4$.  
			  We refer to this problem as the \textit{augmented-index problem on large domains}, and denote it by $\textsc{ind}_\mathcal{U}^d$.
	\end{definition}
	Note that in \cite{jayram2013optimal}, a correct protocol is only required to determine whether $x_{i} =  a$ for some fixed input $a \in \mathcal{U}$ given only to Bob, however such a protocol must succeed with probability $1-\delta$. For the purposes of both problems, it is taken that $|\mathcal{U}| = \Theta(1/\delta)$. In this scenario, we note that the guarantee of our communication problem is strictly weaker, since if one had a protocol that determined whether $x_{i} =  a$ for a given $a \in \mathcal{U}$ with probability $1-\delta$, one could run it on all $a \in \mathcal{U}$ and union bound over all $|\mathcal{U}|$ trails, from which the exact value of $x_i$ could be determined with probability $3/4$, thereby solving the form of the communication problem we have described. We show, nevertheless, that the same lower bound on the communication cost of our protocol holds as the lower bound in \cite{jayram2013optimal}.

	Let $\mathcal{X}$ be the set of all  $x \in \mathcal{U}^d$, let $\mathcal{Y} = [d]$, and define $\mu$ to be the uniform distribution over $\mathcal{X} \times \mathcal{Y}$.
	
	\begin{lemma}
	Suppose $|\mathcal{U}| \geq c$ for some sufficiently large constant $c$.	We have $IC^\to_{\mu,3/4}(\textsc{ind}_{\mathcal{U}}^d) \geq d \log(|\mathcal{U})/2$. 
	\end{lemma}
\begin{proof}
	Fix any protocol $\mathcal{P}$ for $\textsc{ind}_{\mathcal{U}}^d$ which fails with probability at most $1/4$.
	Let $X = (X_1,X_2,\dots,X_d)$ denote Alice's input as chosen via $\mu$, and let $M$ be Alice's message to Bob given $X$. By Proposition \ref{prop:info}
	\[	I(X:M) = \sum_{i=1}^d I(X_i : M | X_1,\dots,X_{i-1})	\]
	\[=  \sum_{i=1}^d \Big(H(X_i | X_1,\dots,X_{i-1})   -  H(X_i |M, X_1,\dots,X_{i-1})\Big)	\]
	First note that since $X_i$ is independent of $X_j$ for all $j \neq i$, we have $H(X_i | X_1,\dots,X_{i-1})   = H(X_i) = \log(|\mathcal{U}|)$. 
	Now since the protocol $\mathcal{P}$ is correct on $\textsc{ind}_{\mathcal{U}}^d$, then the variables $M, X_1,\dots,X_{i-1}$ must be a predictor for $X_i$ with failure probability $1/4$ (since Bob outputs $X_i$ with probability $3/4$ given only $M, X_1,\dots,X_{i-1}$ and his private, independent randomness). So by Fano's inequality (Proposition \ref{prop:info}), we have 
	\[	H(X_i |M, X_1,\dots,X_{i-1}) \leq \frac{1}{4}\log(|\mathcal{U}| - 1) + h_2(\frac{1}{4})	\]
	\[ \leq \frac{1}{2}\log(|\mathcal{U}|) \]
	which holds when $|\mathcal{U}|$ is sufficiently large. Putting this together, we obtain 
	\[	I(X:M) \geq \frac{d \log(|\mathcal{U}|)}{2}	\]
\end{proof}
\begin{corollary}\label{cor:ind}
	We have $R^\to_{3/4}(\textsc{ind}_\mathcal{U}^d) = \Omega(d\log(|\mathcal{U}|))$.
\end{corollary}
	We now use this lower bound on $\textsc{ind}_\mathcal{U}^d$ to show that, even when the index output is from a distribution with constant \textit{additive} error from the true $L_p$ distribution, returning an estimate with probability $1-\delta_2$ still requires $\Omega(\eps^{-p} \log(n) \log(1/\delta_2))$ bits of space.
\begin{theorem}\label{thm:lb}
Fix any $p > 0$ constant bounded away from $0$, and let $\eps < 1/3$ with $\eps^{-p} = o(n)$. Then any $L_p$ sampling algorithm that outputs FAIL with probability at most $1/100$, and otherwise returns an item $\ell \in [n]$ such that $\pr{\ell = l} = |f_l|^p/\|f\|_p^p \pm 1/50$ for all $l \in [n]$, along with an estimate $\tilde{f}_\ell$ such that $\tilde{f}_\ell=  (1 \pm \eps)f_\ell$ with probability $1 - \delta_2$, requires $\Omega(\eps^{-p} \log(n) \log(1/\delta_2))$ bits of space.
\end{theorem}
\begin{proof}
	We reduce via $\textsc{ind}_\mathcal{U}^d$. Suppose we have a streaming algorithm $\mathcal{A}$ which satisfies all the properties stated in the theorem.
	Set $|\mathcal{U}| = 1/(10\delta_2)$, and let $X \in \mathcal{U}^d$ be Alice's input, where $d = rs$ where $r =  \frac{1}{10^{p+1} \eps^p}$ and $s = \log(n)$. Alice conceptually divides $X$ into $s$ blocks $X^1,\dots,X^s$, each containing $r$ items $X^i = X^i_{1},X^i_{2},\dots,X^i_{r} \in \mathcal{U}$. Fix some labeling $\mathcal{U} = \{\sigma_1,\dots,\sigma_k\}$, and let $\pi(X^i_j) \in [k]$ be such that $X^i_j = \sigma_{\pi(X^i_j)}$. Then each $X^i_j$ can be thought of naturally as a binary vector in $\R^{rsk}$ with support $1$, where $(X^i_j)_t = 1$ when $t = (i-1)r + (j-1)k + \pi(X^i_j)$, and $(X^i_j)_t = 0$ otherwise. Set $n' = rsk < n$ for $\eps^{-p} = o(n)$. Using this interpretation of $X^i_j \in \R^{rsk}$, we define the vector $f \in \R^{rsk}$ by
	\[	f = \sum_{i=1}^s \sum_{j=1}^{r}B^i X^i_j 	\]
	Where $B = 10^{1/p}$. 
	Alice can construct a stream with the frequency vector $f$ by making the necessary insertions, and then send the state of the streaming algorithm $\mathcal{A}$ to Bob. Now Bob has some index $i^* \in [d] = [rs]$, and his goal is to output the value of $X^{i'}_{j'}= X_{i^*}$ such that $i^* = (i'-1)r + j'$. Since Bob knows $X_i^j$ for all $(i,j)$ with $i > i'$, he can delete off the corresponding values of $B^i X_j^i$ from the stream, leaving the vector $f$ with the value 
		\[	f = \sum_{i=1}^{i'} \sum_{j=1}^{r}B^i X^i_j 	\]
	For $j \in [k]$, let $\gamma^j \in \R^{rsk}$ be the binary vector  with $\gamma^j_{(i'-1)r + (j'-1)k + j} = B^{i'}/(10\eps)$ and $\gamma^j_t = 0$ at all other coordinates $t \neq (i'-1)r + (j'-1)k + j$. Bob then constructs the streams $f^j = f + \gamma^j$ for $j=1,\dots,k$ sequentially. After he constructs $f^j$, he runs $\mathcal{A}$ on $f^j$ to obtain an output $(\ell_j,\tilde{f}_{\ell_j}^j) \in ([n'] \times \R)  \cup (\{\ttx{FAIL}\} \times \{\ttx{FAIL}\} )$ from the streaming algorithm, where if the algorithm did not fail we have that $\ell_j \in [n']$ is the index output and $\tilde{f}_{\ell_j}^j$ is the estimate of $f_{\ell_j}^j$. By union bounding over the guarantee of $\mathcal{A}$ we have that if $\ell_j \neq \ttx{FAIL}$ then $\tilde{f}_{\ell_j} = (1 \pm \eps)f_{\ell_j}^j$ for all $j=1,2,\dots,k$ with probability $1 - k \delta_2 > 9/10$. Call this event $\mathcal{E}_1$. Conditioned on $\mathcal{E}_1$, it follows that if for each $\ell_j$ with $\ell_j = (i'-1)r + (j'-1)k + j$, if $X^{i'}_{j'} = \sigma_j$ then 
	\[ \tilde{f}_{\ell_j}^j > B^{i'}( 1 + \frac{1}{10\eps})(1 - \eps) > \frac{B^{i'}}{10\eps} + \frac{9}{10}B^{i'}  - \eps B^{i'}  \]
	On the other hand, if $X^{i'}_{j'} \neq \sigma_j$, then we will have 
	\[\tilde{f}_{\ell_j}^j < (B^{i'}/(10\eps))(1 + \eps) = \frac{B^{i'}}{10\eps}  + \frac{B^{i'}}{10} \]\[<  \frac{B^{i'}}{10\eps} + \frac{9}{10}B^{i'}  - \eps B^{i'} \]
	 using that $\eps < 1/3$. Thus if $\ell_j = (i'-1)r + (j'-1)k + j$, Bob can correctly determine whether or not $X^{i'}_{j'} = \sigma_j$. Now suppose that, in actuality, Alice's item was $X^{i'}_{j'} = \sigma_{\tau} \in \mathcal{U}$ for some $\tau \in [k]$. Set $\lambda = (i'-1)r + (j'-1)k + \tau$. To complete the proof, it suffices to lower bound the probability that $\ell_{\tau} \neq\lambda$. 
	
	Thus we consider only the event of running $\mathcal{A}$ on $f^\tau$. We know that with probability $99/100$, $\ell_\tau \neq \ttx{FAIL}$. We write $\mathcal{E}_2$ to denote the event that $\ell_\tau \neq \ttx{FAIL}$. Let $f_{-\lambda}$ be equal to $f$ everywhere except with the coordinate $\lambda$ set equal to $0$. Then
	\[ \|f^\tau_{-\lambda}\|_p^p <	 \sum_{i=1}^{i'} \sum_{j=1}^{r}(B^p)^i 	\]
	\[	\leq   r \sum_{i=1}^{i'}10^i 	\leq (\frac{1}{10^{p+1}\eps^p})\frac{10^{i'+1}}{9} 	\]
	So
	\[	\frac{|f_{\lambda}^\tau|^p }{\|f^\tau_{-\lambda}\|_p^p} \geq \frac{10^{i'}( \frac{1}{10\eps})^{p}}{  (\frac{1}{10^{p+1}\eps^p})\frac{10^{i'+1}}{9}}		\geq \frac{ 9( \frac{1}{10\eps})^{p}}{ (\frac{1}{10\eps})^p} \geq 9\]
	Since $\mathcal{A}$ has $1/50$-additive error, we conclude $\pr{\ell_\tau = \lambda } > 9/10 - 1/50 = 22/25$, and call the event that this occurs $\mathcal{E}_3$. Then conditioned on $\mathcal{E} = \mathcal{E}_1\cap \mathcal{E}_2 \cap \mathcal{E}_3$ Bob sucsessfully recovers the value of $X_{j'}^{i'} = X_{i^*}$, and thus solves the communication problem. Note that the probability of success is $\pr{\mathcal{E}} > 1 - (1/10 + 1/100 + 3/25) > 3/4$, and thus this protocol solves $\textsc{ind}_{\mathcal{U}}^d$.
	So by Corollary \ref{cor:ind}, it follows that any such streaming algorithm $\mathcal{A}$ requires $\Omega(rs \log(|\mathcal{U}|)) = \Omega(\eps^{-p} \log(n) \log(1/\delta_2))$ bits of space. Note that the stream $f$ in question had length $n' < n$ for $p$ constant bounded from $0$, and no coordinate in the stream ever had a value greater than $\poly(n)$, and thus the stream in question is valid in the given streaming model.
\end{proof}

		\section{Conclusion}
	This work demonstrates the existence of perfect $L_p$ samplers for $p \in (0,2)$ using $O(\log^2(n)\log(1/\delta))$ bits of space in the random oracle model. This bound is tight in terms of both $n$ and $\delta$. However, to derandomize our algorithm for $p<2$, our space increases by a $O((\log \log n)^2)$-factor, which is perhaps unnecessary. There are also several other open problems for $L_p$ samplers which this work does not close. Notably, there is still a $\log(n)$ factor gap between the upper and lower bounds for $L_2$ samplers, as the best known lower bound for any $p \geq 0$ is $\Omega(\log^2 n)$, compared to our upper bound of $O(\log^3 n)$. While perfect $L_2$ samplers using polylogarithmic space were not known before this work, our upper bound matches the best upper bounds of prior approximate $L_2$ samplers with constant $\nu = \Omega(1)$. It is therefore an open question whether this additional factor of $\log n$ is required in the space complexity of an $L_2$ sampler, perfect or otherwise.  
	
	Secondly, one notable shortcoming of the perfect sampler presented in this paper is the large update time. To obtain a perfect sampler as defined in the introduction, the algorithm in this paper takes polynomial (in $n$) time to update its data structures after each entry in the stream. This is clearly non-ideal, since most streaming applications demand constant or polylogarithmic update time. Using our rounding procedure, we can obtain a $(1 \pm 1/\poly(\log n))$ relative error sampler with polylogarithmic update time (and the same space as the perfect sampler), but it is still an open problem to design a perfect $L_p$ sampler with optimal space dependency as well as polylogarithmic update time.

	Finally, there are several gaps in the dependency on $\eps,\delta_2$ in our procedure which, in addition to outputting an index $i \in [n]$, also outputs a $(1 \pm \eps)$ estimate of the frequency $f_i$. Taking Theorem \ref{thm:lb} along with the known lower bounds for $L_p$ sampling, our best lower bound for the problem is $\Omega(\log^2(n)\log(1/\delta_1) + \eps^{-p} \log(n) \log(1/\delta_2))$, where $\delta_1$ is the probability that the sampler fails to output an index $i$. On the other hand, our best upper bound is $O\big(\big(\log^2(n) + \beta \log(n) \log(1/\delta_2)\big) \log(1/\delta_1) \big)$ for $p \in (0,2)$, and $O\big(\big(\log^3(n) + \eps^{-2}\log^2(n) \log(1/\delta_2)\big) \log(1/\delta_1) \big)$ for $p=2$, where $\beta =\min\big\{\eps^{-2}, \eps^{-p} \log\big(\frac{1}{\delta_2}\big)  \big\}$. Notably, the $\log(1/\delta_1)$ multiplies the $\log(1/\delta_2)$ term in the upper bound but not in the lower bound. We leave it as an open problem to determine precisely the right dependencies of such an algorithm on $\eps,\delta_1,\delta_2$.

	\section*{Acknowledgments}
	The authors would like to thank Raghu Meka for a helpful explanation of the \cite{gopalan2015pseudorandomness} PRG, and for pointing out how the arguments could be extended to fooling functions of multiple half-spaces (Lemma \ref{lemma:halfspace}). The authors would also like to thank Ryan O'Donnell for a useful discussion on pseudo-random generators in general.
	
	\bibliography{cluster}
	\appendix

	\section{Original $L_p$ Sampling via Count-Sketch}\label{sec:appendix}

	In a previous version of this work, we used a slightly different testing Algorithm for the $L_p$ Sampler. Namely, we used the classic count-sketch estimation procedure of Theorem \ref{thm:count-sketch} to obtain a $y$ such that $\|y - \zeta\|_\infty$ is small. We then take the largest coordinate of $y$ as our guess of the maximizer in $\zeta$. The algorithm presented in the current version has the advantage of being slightly simpler, and does not incur the $(\log \log n)^2$ blow-up in space for $p=2$ from the derandomization. In this section, we show how the algorithm in the original version can be derandomized using the general derandomization results for linear sketches of Theorem \ref{thm:derandomGeneral}.  First, we introduce a few preliminary tools that we will need.

	\subsection{Preliminaries}
We first introduce the $L_2$ estimation algorithm of \cite{indyk2006stable}.
	To estimate $\|f\|_2$ for $f \in \R^n$, we generate i.i.d. Gaussians $\varphi_{i,j} \sim \mathcal{N}(0,1)$ for $i \in [n]$ and $j \in [r]$ where $r = \Theta(\log(n))$. We will later derandomize this assumption. We then store the vector $B \in \R^r$ where $B_j = \sum_{i=1}^n f_i \varphi_{i,j}$ for $j=1,\dots,r$,
	which can be computed update by update throughout the stream. We return the estimate $R = \text{median}_j \frac{5|B_j|}{4}$. 
	\begin{lemma}\label{lem:L2}
		For any constant $c > 0$, The value of $R$ as computed in the above algorithm satisfies $\frac{1}{2}\|f\|_2 \leq R \leq 2\|f\|_2$ with probability $1-n^{-c}$.
	\end{lemma}
	\begin{proof}
		Each coordinate $B_j$ is distributed as $|B_j| =| g_j| \|f\|_2$, where $g_j$ are i.i.d. Gaussian random variables. A simple computation shows that $\pr{|g_j| \in [2/5,8/5]} > .55$, and thus $\pr{(5/4)|B_j| \in [1/2 \|f\|_2, 2 \|f\|_2]} > .55$. Then by Chernoff-Hoeffding bounds, the median of $O(\log(n))$ repetitions satisfies this bound with probability $1- n^{-c}$ as stated. 
	\end{proof}
	
	

Finally, we remark that making a simple modification to the classic count-sketch algorithm (see Theorem \ref{thm:count-sketch}), still results in the same error guarantee. Let $A \in \R^{d \times k}$ be a $d \times k$ count-sketch matrix. The modification is as follows: instead of each variable $h_{i}(k)$ being uniformly distributed in $\{1,2,\dots,k\}$, we replace them with variables $h_{i,j,k} \in \{0,1\}$ for $(i,j,k) \in [d] \times [k] \times [n]$, such that $h_{i,j,k}$ are all i.i.d. and equal to $1$ with probability $1/k$. We also let $h_{i,h,k} \in \{1,-1\}$ be i.i.d. Rademacher variables ($1$ with probability $1/2$). Then $A_{i,j} = \sum_{k=1}^n f_k g_{i,j,k} h_{i,j,k}$, and the estimate $y_k$ of $f_k$ is given by:
\[ y_k = \text{median} \{	g_{i,j,k} A_{i,j} \; | \; h_{i,j,k = 1}	\}\]
Thus the element $f_k$ can be hashed into multiple buckets in the same row of $A$, or even be hashed into none of the buckets in a given row. By Chernoff bounds, $|\{	g_{i,j,k} A_{i,j} \; | \; h_{i,j,k = 1}	\}| = \Theta(d)$ with high probability for all $k \in [n]$. Observe, the marginal distribution of each bucket is the same as before, and the thus the original analysis of count-sketch  (\cite{charikar2002finding}) is unchanged, as it only relies on taking the median of $\Theta(d)$ buckets, each of which independently succeed in giving a good estimate with probability at least $2/3$, as is the case here. Thus the bounds of Theorem \ref{thm:count-sketch} apply as usual.

\begin{theorem}\label{thm:count-sketch2}
	Let $A \in \R^{d \times k}$ be the modified count-sketch as described above. 
	If $d = \Theta(\log(n))$, $k = 6/\eps^2$, and $c \geq 1$ is any constant, then we have $\|y - f\|_\infty < \eps\|f_{\text{tail}(1/\eps^2)}\|_2$ with probability $1 - n^{-c}$.
\end{theorem}

	\begin{figure}
		\fbox{\parbox{\textwidth}{ \texttt{$L_p$ Sampler}
				\begin{enumerate}[topsep=0pt,itemsep=-1ex,partopsep=1ex,parsep=1ex] 
					\item For $0 < p < 2$, set $\eps = \Theta(1)$, and for $p=2$ set $\eps = \Theta(\sqrt{1/\log(n)})$. Let $d = \Theta(\log(n))$, and instantiate a $d \times 6/\eps^2$ count-sketch table $A$, and set $\mu \sim \text{Uniform}[\frac{1}{2},\frac{3}{2}]$.  
					\item Duplicate updates to $f$ to obtain the vector $F \in \mathbb{R}^{n^c}$ so that $f_i = F_{i_j}$ for all $i \in [n]$ and $j=1,2,\dots,n^{c-1}$, for some fixed constant $c$.
					\item Choose i.i.d. exponential random variables $t = (t_1,t_2,\dots,t_{n^c})$, and construct the stream $\zeta_i = F_i\cdot \texttt{rnd}_\nu (1/t_i^{1/p})$.
					\item Run $A$ on $\zeta$ to obtain estimate $y$ with $|y - |\zeta|\|_\infty < \eps \|\zeta_{\text{tail}(1/\eps^2)}\|_2$ as in Theorem \ref{thm:count-sketch2}.
					\item Run $L_2$ estimator on $\zeta$ to obtain $R \in [\frac{1}{2}\|\zeta\|_2, 2 \|\zeta\|_2 ]$ with high probability.
					\item If $y_{(1)} - y_{(2)} < 100 \mu \eps R$ or if $y_{(2)} < 50 \eps \mu R$,  report \ttx{FAIL}, else return $i \in [n]$ such that $y_{i_j} = y_{(1)}$ for some $j \in [n^{c-1}]$.
				\end{enumerate}
		}}\caption{Our main $L_p$ Sampling algorithm} \label{fig:samplerOld}
	\end{figure}
	\subsection{The $L_p$ Sampler} \label{sec:samplerOld}
	
	We begin by describing the original sampling algorithm, as shown in Figure \ref{fig:samplerOld}. The algorithm duplicates coordinates just as the Sampler of Figure \ref{fig:sampler}, and scales it by inverse $1/p$-th powers of i.i.d. exponentials $1/t_i^{1/p}$. We also perform the same rounding procedure, turning $z$ into $\zeta$. 
Having constructed the transformed stream $\zeta$, we then run a $\Theta(\log(n)) \times 6/\eps^2$ instance $A$ of count-sketch on $\zeta$ to obtain an estimate vector $y$ with $\|y - |\zeta|\|_\infty < \eps \|\zeta_{\text{tail}(1/\eps^2)}\|_2$ with probability $1-n^{-c}$ (as in Theorem \ref{thm:count-sketch2}). Here, for a vector $v \in \R^n$, $|v| \in \R^n$ is the vector such that $(|v|)_i = |v_i|$ for all $i \in [n]$. Thus $y_j$ is an estimate of the \textit{absolute value} $\zeta_j$, and is always positive. This is simply accomplished by taking the absolute value of the usual estimate $y$ obtained from count-sketch.

Then for $0 < p < 2$, we set $\eps = \Theta(1)$, and for $p=2$ we set $\eps = \Theta(1/\sqrt{\log(n)})$. Next, we obtain estimates $R \in [\frac{1}{2}\|\zeta\|_2, 2\|\zeta\|_2]$ via the algorithm of Lemma \ref{lem:L2} with high probability. The algorithm then finds $y_{(1)},y_{(2)}$ (the two largest coordinates of $y$), and samples $\mu \sim \text{Uniform}[1/2,3/2]$.  It then checks if $y_{(1)}-y_{(2)} < 100 \mu \eps R$ or if  $y_{(2)} < 50 \eps \mu R$, and reports \ttx{FAIL} if either occur, otherwise it returns $i \in [n]$ with $y_{i_j} = y_{(1)}$ for some $j \in [n^{c-1}]$.

	Let $i^* \in [n^c]$ be the index of the maximizer in $y$, so $y_{i^*} = y_{(1)}$. By checking that $y_{(1)}-y_{(2)} >100 \mu \eps R$, noting that $100 \mu \eps R \geq 25 \|y - |\zeta|\|_\infty$ and $z_k = (1 \pm \nu)\zeta_k$ for all $k \in [n^{c}]$, for $\nu < \eps$ sufficiently small we ensure that $|z_{i^*}|$ is also the maximum element in $z$. The necessity for the test $y_{(2)} \geq 50 \eps \mu R$ is less straightforward (see Remark \ref{rem:runtimeold} for justification).
	To prove correctness, we need to analyze the conditional probability of failure given $D(1) = i$.  Let $N = |\{i \in [n^c] \: | \: F_i \neq 0 \}|$ ($N$ is the support size of $F$). We can assume that $N \neq 0$ (to check this one could run, for instance, the $O(\log^2(n))$-bit support sampler of \cite{Jowhari:2011}). Note that $n^{c-1} \leq N \leq n^c$. We now will prove the propositions and lemmas needed to demosntrate correctness of this sampler. Lemmas \ref{lem:main2} and \ref{lem:failbound2} are the analogous results to Lemmas \ref{lem:main} and \ref{lem:failbound} in Section \ref{sec:main}, and will follow nearly the same proofs.

	\begin{proposition}\label{prop:inde}
	Let $X,Y \in \R^d$ be random variables where $Z = X + Y$. Suppose $X$ is independent of some event $E$, and let $M > 0$ be such that for every $i \in [d]$ and every $a < b$ we have $\pr{a \leq X_i \leq b} \leq M(b-a)$.  Suppose further that $|Y|_\infty \leq \eps$.  Then if $I = I_1 \times I_2 \times \dots \times I_d \subset \R^n$, where each $I_j = [a_j,b_j] \subset \R$, $- \infty \leq a_j < b_j \leq \infty$ is a (possibly unbounded) interval, then
	\[\pr{Z \in I | E} = \pr{Z \in I } + O(\eps d M)\]
\end{proposition}
\begin{proof}
	For $j \in [d]$, let $\overline{I}_j = [a_j- \eps ,b_j + \eps],\underline{I}_j = [a_j+ \eps ,b_j - \eps]$, and let $\overline{I} = \overline{I}_1 \times \dots \times \overline{I}_d$, and $\underline{I} = \underline{I}_1 \times \dots \times \underline{I}_d$. If one of the endpoints is unbounded we simply use the convention $\infty \pm c = \infty$, $-\infty \pm c = -\infty$ for any real $c$. Then 
	\[\pr{Z \in I | E} \leq \pr{X \in \overline{I} | E} =  \pr{ 
		X \in \overline{I} } \]\[
	\leq  \pr{ X \in \underline{I} } + \pr{\bigcup_{i=1}^d X_i \in \overline{I}_j \setminus \underline{I}_j  }\]
	By the union bound, this is at most $\pr{ X \in \underline{I} } + 4d\eps M \leq \pr{Z \in I} + 4d\eps M$. 	
	Similarly, $\pr{Z \in I | E } \geq \pr{X \in \underline{I}} \geq \pr{X \in \overline{I}} - 4 d \eps M \geq \pr{Z \in I} - 4 d \eps M$.
\end{proof}

\begin{lemma}\label{lem:main2}
	For $p \in (0,2]$ a constant bounded away from $0$ and any $\nu \geq n^{-c}$, $\pr{\neg \ttx{FAIL} \: | \: D(1) } =  \pr{\neg \ttx{FAIL} } \pm O(\log(n) \nu )$ for every possible $D(1) \in [N]$.
\end{lemma} 
\begin{proof}
By Lemma \ref{lem:Zconcentration}, conditioned on $\mathcal{E}_1$, for every $k < N-n^{9c/10}$ we have $|z_{D(k)}| =  U_{D(k)}^{1/p} (1 \pm O(n^{-c/10}))^{1/p} =  U_{D(k)}^{1/p} (1 \pm O(\frac{1}{p}n^{-c/10}))$ (using the identity $(1+x) \leq e^{x}$ and the Taylor expansion of $e^x$), where $U_{D(k)} = (\sum_{\tau=1}^{k} \frac{E_\tau}{\ex{\sum_{j=\tau}^{N} |F_{D(j)}|^p }})^{-1}$ is independent of the anti-rank vector $D$ (in fact, it is totally determined by $k$ and the hidden exponentials $E_i$). Then for $c$ sufficiently large, we have $|\zeta_{D(k)}| =U_{D(k)}^{1/p} (1 \pm O(\nu))$, and so for all $p \in (0,2]$ and $k < N - n^{9c/10}$
	
	\[	|\zeta_{D(k)}| = U_{D(k)}^{1/p} + U_{D(k)}^{1/p} V_{D(k)}\]
	Where $V_{D(k)}$ is some random variable that satisfies $|V_{D(k)}| = O(\nu)$.	
		Now consider a bucket $A_{i,j}$ for $(i,j) \in [d] \times [6/\eps^2]$.  Let $\sigma_k = \ttx{sign}(z_k) = \ttx{sign}(\zeta_k)$ for $k \in [n^c]$. Then we write $A_{i,j} =\sum_{k \in B_{ij}}\sigma_{D(k)} |\zeta_{D(k)}| g_{i,j,D(k)} + \sum_{k \in S_{ij}}\sigma_{D(k)} |\zeta_{D(k)}| g_{i,j,D(k)}$ where $B_{ij} = \{k \leq N - n^{9c/10}   \:|\: h_{i,j,D(k)} =1 \}$ and  $S_{ij} = \{n^c\geq k > N - n^{9c/10} \:| h_{i,j,D(k)} =1 \}$ (see notation above Theorem \ref{thm:count-sketch2}). Here we define $\{D(N+1),\dots,D(n^c)\}$ to be the set of indices $i$ with $F_i=0$ (in any ordering, as they contribute nothing to the sum).  So 
		\[A_{i,j} =\sum_{k \in B_{ij}} g_{i,j,D(k)} \sigma_{D(k)} U_{D(k)}^{1/p} + \sum_{k \in B_{ij}} g_{i,j,D(k)}\sigma_{D(k)} U_{D(k)}^{1/p} V_{D(k)} +  \sum_{k \in S_{ij}} g_{i,j,D(k)} \zeta_{D(k)}  \]
		Importantly, observe that since the variables $h_{i,j,D(k)}$ are fully independent, the sets $B_{i,j},S_{i,j}$ are independent of the anti-rank vector $D$. In other words, the values $h_{i,j,D(k)}$ are independent of the values $D(k)$ (and of the entire anti-rank vector). Note that this would not necessarily be the case if these were only $\ell$-wise independent for some $\ell = o(n^c)$. So we can condition on a fixed set of values  $\{h_{i,j,D(1)},\dots, h_{i,j,D(n^c)}\}$ now, which fixes the sets $B_{i,j},S_{i,j}$.
		
		\begin{claim}\label{claim:khine2}
			For all $i,j$ and $p \in (0,2]$, we have $|\sum_{k \in B_{ij}} g_{i,j,D(k)}\sigma_{D(k)}U_{D(k)}^{1/p} V_{D(k)}| +|\sum_{k \in S_{ij}} g_{i,j,D(k)} \zeta_{D(k)}| = O(\sqrt{\log(n)}\nu \|z\|_2 )$ with probability $1-O(\log^2(n) n^{-c})$.
		\end{claim}
		\begin{proof}
			By Khintchines's inequality (Fact \ref{fact:khintchine}), we have $|\sum_{k \in S_{ij}} g_{i,j,D(k)} \zeta_{D(k)}| = O(\sqrt{\log(n)}) (\sum_{k \in S_{i,j}} \allowbreak (2z_{D(k)})^2)^{1/2}$  with probability $1 - n^{-c}$.  This is a sum over a subset of the $n^{9c/10}$ smallest items $|z_i|$, and thus $\sum_{k \in S_{i,j}} z_{D(k)}^2 < \frac{n^{9c/10}}{N} \|z\|_2^2$, giving $|\sum_{k \in S_{ij}} g_{i,j,D(k)} \zeta_{D(k)}| =  O(\sqrt{\log(n)} n^{-c/30}\|z\|_2)$. Furthermore, using the fact that for $k \leq N - n^{9c/10}$ we have $|\zeta_{{D(k)}}| < 2U_{{D(k)}}^{1/p}$ and $|V_{{D(k)}}| = O(\nu)$, we have $|\sum_{k \in B_{ij}} g_{i,j,D(k)}\sigma_{D(k)} U_{D(k)}^{1/p} V_{D(k)}| = O(\sqrt{\log(n)} \nu \|z\|_2)$ with probability $1 - n^{-c}$ again by Khintchine's inequality, as needed.				
			Note there are only $O(\eps^{-2}\log(n)) = O(\log^2(n))$ (for $p< 2$ this is $O(\log(n))$) terms $|\sum_{k \in B_{ij}} g_{i,j,D(k)} \sigma_k U_{D(k)}^{1/p} V_{D(k)}| + |\sum_{k \in S_{ij}} g_{i,j,D(k)} \zeta_{D(k)}|$ which ever occur in all of the $A_{i,j}$'s, since the count-sketch has size $O(\eps^{-2} \log(n))$. Union bounding over these buckets, and taking $c$ sufficiently large, the claim follows.					
		\end{proof}

		Call the event where the Claim \ref{claim:khine2} holds $\mathcal{E}_2$. Conditioned on $\mathcal{E}_2$, we can decompose $|A_{i,j}|$ for all $i,j$ into $|\sum_{k \in B_{ij}} g_{i,j,D(k)}\sigma_{D(k)}U_{D(k)}^{1/p}| + \mathcal{V}_{ij}$ where $\mathcal{V}_{ij}$ is some random variable satisfying $|\mathcal{V}_{ij}| = O(\sqrt{\log(n)}\nu \|z\|_2)$ and $\sum_{k \in B_{ij}} g_{i,j,D(k)}\sigma_{D(k)}U_{D(k)}^{1/p}$ is independent of the anti-rank vector $D$ (it depends only on the hidden exponentials $E_k$, and the uniformly random signs $g_{i,j,D(k)}\sigma_{D(k)}$). Let $U_{ij}^* = |\sum_{k \in B_{ij}} g_{i,j,D(k)}\sigma_{D(k)}U_{D(k)}^{1/p}|$. Let $\Gamma(k) = \{(i,j) \in  [d] \times [k] \; | \; h_{i,j,D(k)} =1\}$. Then our estimate for $|\zeta_{D(k)}|$ is $y_{D(k)} = \text{median}_{(i,j) \in \Gamma(k)}\{ U_{i,j}^* + \mathcal{V}_{i,j} \} = \text{median}_{(i,j) \in \Gamma(k)}\{U_{i,j}^*\} + \mathcal{V}_{D(k)}^*$ where $|\mathcal{V}_{D(k)}^*| = O(\sqrt{\log(n)}\nu \|z\|_2)$ for all $k \in [n^c]$.
				
		We now consider our $L_2$ estimate, which is given by $R = \frac{5}{4}\text{median}_j\{|\sum_{k \in [n^c]} \varphi_{kj} \zeta_k|\}$ where the $\varphi_{kj}$'s are i.i.d. normal Gaussians. We can write this as 
		\[R = \frac{5}{4}\text{median}_j\Big\{\Big|\sum_{k \in B} \varphi_{D(k)j} \sigma_{D(k)} U_{D(k)}^{1/p} + (\sum_{k \in B} \varphi_{{D(k)}j} \sigma_{D(k)} U_{D(k)}^{1/p} V_{D(k)} + \sum_{k \in S} \varphi_{{D(k)}j} \zeta_{D(k)})\Big|\Big\}\]
		 where $B = \cup_{ij} B_{ij}$ and $S = [n^c] \setminus B$. Now the $\varphi_{{D(k)}j}$'s are not $\pm 1$ random variables, so we cannot apply Khintchine's inequality. However, by the $2$-stability of Gaussians (Definition \ref{def:stable}), if $\varphi_1,\dots,\varphi_n$ are i.i.d. Gaussian, then $\pr{|\sum_{i} \varphi_i a_i| > O(\sqrt{\log(n)}) \|a\|_2} =\pr{|\varphi| \|a\|_2> O(\sqrt{\log(n)}) \|a\|_2}$, where $\varphi$ is again Gaussian. This latter probability can be bounded by $n^{-c}$ via the pdf of a Gaussian, which is the same bound as Khintchine's inequality. So applying the same argument as in Claim \ref{claim:khine2}, we have $R = \frac{5}{4}\text{median}_j \{(|\sum_{k \in B} \varphi_{{D(k)}j} \sigma_{D(k)}\allowbreak U_{D(k)}^{1/p}|\} + \mathcal{V}_R$ with probability $1 - O(n^{-c})$ where $|\mathcal{V}_R| = O(\sqrt{\log(n)}\nu\|z\|_2)$. Call this event $\mathcal{E}_3$. By the symmetry of Gaussians, the value $\varphi_{{D(k)}j} \sigma_{D(k)}$ is just another i.i.d. Gaussian, so $|\sum_{k \in B} \varphi_{{D(k)}j} \sigma_{D(k)} U_{D(k)}^{1/p}|$ is independent of the anti-rank vector.
		
		Let $U_{D(k)}^* = \text{median}_{(i,j) \in \Gamma(k)}\{U_{i,h_i(D(k))}^* \}$ for $k \in [n^c]$, and $U_R^* =\frac{5}{4}\text{median}_j(|\sum_{k \in B} \varphi_{{D(k)}j}\sigma_{D(k)} U_{D(k)}^{1/p}|)$. Then both $U_{D(k)}^*,U_R^*$ are independent of the anti-ranks $D(k)$ (the former does, however, depend on $k$), and $y_{D(k)} = U_{D(k)}^* + V_{D(k)}^*$. 
		Now to analyze our failure condition, we define a deterministic function $\Lambda(x,v) \in \R^2$. For vector $x$ and a scalar $v$, set $\Lambda(x,v)_1 = x_{(1)} - x_{(2)} - 100 \eps v $, and $\Lambda(x,v)_2 = x_{(2)} - 50\eps v$. Note $ \Lambda(y,  \mu R) \geq 0$ (coordinate-wise) if and only if $\neg$\ttx{FAIL}.
		
		\begin{claim}
			Conditioned on $\mathcal{E}_1 \cap \mathcal{E}_2 \cap\mathcal{E}_3$, we have the decomposition $\Lambda(y, \mu R) = \Lambda(\vec{U}^*,  \mu U_R^*) + \overline{V}$ where the former term is independent of the max index and $\|\overline{V}\|_\infty = O(\sqrt{\log(n)}\nu\|z\|_2)$.
		\end{claim}
		\begin{proof}
			We have shown that $ |\mathcal{V_R}|$ and $|\mathcal{V}_{D(k)}^*|$ are both  $O(\sqrt{\log(n)}\nu\|z\|_2)$ for all $k \in [n^c]$ conditioned on $\mathcal{E}_1 \cap \mathcal{E}_2 \cap \mathcal{E}_3$. We have $y = \vec{U}^* + \vec{\mathcal{V}}^*$, where $\vec{U}^*_{D(k)} = U^*_{D(k)}$ and $\vec{V}^*_{D(k)} = V^*_{D(k)}$, so $\vec{\mathcal{V}}^*$ can change the value of the two largest coordinates in $y$ by at most $\|\vec{\mathcal{V}}^*\|_\infty =O(\sqrt{\log(n)}\nu\|z\|_2)$. Similarly $|\mathcal{V_R}|$ can change the value of $R$ by at most $O(\sqrt{\log(n)}\nu\|z\|_2)$, which completes the proof of the decomposition. To see the claim of independence, note that $\Lambda(\vec{U}^*,  \mu U_R^* ) $ is a deterministic function of the hidden exponentials $E_1,\dots,E_{N}$, the random signs $g$, and the uniform random variable $\mu$, the joint distribution of all of which is marginally independent of the anti-rank vector $D$, which completes the claim.  
		\end{proof}
		
		To complete the proof of the Lemma, it suffices to show the anti-concentration of $\Lambda(\vec{U}^*, \mu U_R^*)$. Now for any interval $I$
		\[	\pr{\Lambda(\vec{U}^*,  \mu U_R^* )_1 \in I} = \pr{ \mu \in I' / (100 \eps U_R^*)}	\]
		\[	= O(|I|/ (\eps U_R^*))	\]
		and 
		\[	\pr{\Lambda(\vec{U}^*, \mu U_R^*)_2 \in I} = \pr{ \mu \in I'' / (50 \eps U_R^*)}	\]
	\[	= O(|I|/ (\eps U_R^*))	\]
		where $I'$ and $I''$ are the result of shifting the interval $I$ by a term which is independent of $\mu$. Here $|I| \in [0,\infty]$ denotes the size of the interval $I$. Thus it suffices to lower bound $U_R^* $. We have $ 2U_R^*>R>\frac{1}{2}\|z\|_2$ after conditioning on the success of our $L_2$ estimator, an event we call $\mathcal{E}_4$, which holds with probability $1 - n^{-c}$ by Lemma \ref{lem:L2}. Thus $	\pr{\Lambda(\vec{U}^*,  \mu U_R^*)_1 \in I} = O(\eps^{-1}|I|/\|z\|_2)$ and  $	\pr{\Lambda(\vec{U}^*,  \mu U_R^*)_2 \in I} = O(\eps^{-1}|I|/\|z\|_2)$ for any interval $I$.
	So by Proposition \ref{prop:inde}, conditioned on $\mathcal{E}_1 \cap \mathcal{E}_2 \cap \mathcal{E}_3\cap \mathcal{E}_4$ we have
	
		\begin{equation}\label{eqn:prob}
		\bpr{\Lambda(y, \mu R) \geq  \vec{0} \in \R^2 \:\big|\: D(1)  } = \bpr{\Lambda(y,  \mu R) \geq  \vec{0} } \pm  O(\log(n)\nu)
		\end{equation}		
		Note that $\mathcal{E}_1 \cap \mathcal{E}_2 \cap \mathcal{E}_3\cap \mathcal{E}_4$ holds with probability $1 - O(n^{-c+1})$, so choosing $c$ such that $n^{-c} < \log(n) \nu$, Equation \ref{eqn:prob} holds without conditioning on $\mathcal{E}_1 \cap \mathcal{E}_2 \cap \mathcal{E}_3\cap \mathcal{E}_4$, which completes the proof of the lemma. 
	\end{proof}

	
	\begin{lemma}\label{lem:failbound2}
		If $y$ is the vector obtained via count-sketch as in the algorithm \texttt{$L_p$ Sampler}, and $0 < p \leq 2$ a constant, then we have $\pr{y_{(1)} - y_{(2)} > 100 \eps \mu R , \; y_{(2)} > 50 \eps \mu R } \geq 1/2$, where $\eps = \Theta(1)$ when $p < 2$, and $\eps = \Theta(1/\sqrt{\log(n)})$ when $p = 2$.
	\end{lemma}
	\begin{proof}
		By Proposition \ref{prop:1}, with probability $1-3e^{-4} > .9$ we have $\|z_{tail(16)}\|_2 = O(|F\|_p)$ for $p<2$, and $\|z_{tail(16)}\| = O(\sqrt{\log(n)}\|F\|_p)$ when $p=2$. Observe that for $t \in [16]$ we have $|z_{D(t)}| < \|F\|_p (\frac{2}{\sum_{\tau=1}^t E_\tau})^{1/p}$, and with probability $99/100$ we have $E_1 > 1/100$, which implies that $|z_{D(t)}| = O(\|F\|_p)$ for all $t \in [16]$.  Conditioned on this, we have $\|z\|_2 < q\|F\|_p$ where $q$ is a constant when $p<2$, and $q = \Theta(\sqrt{\log(n)})$ when $p=2$. In either case, we know that the estimate $y$ from count sketch satisfies $\| y - |\zeta|\|_\infty < \eps \|\zeta \|_2 < 2 \eps \|z\|_2 = O(\|F\|_p)$. Thus conditioning on the high probability event that $R = \Theta(\|\zeta\|_2)$, we have that 
		$100 \eps \mu R  = O(\|F\|_p)$, where we can rescale the quantity down by any constant by a suitable rescaling of $\eps$.

		Now note that $|z_{D(1)}| = \|F\|_p / E_1^{1/p}$ and $|z_{D(1)}| = \|F\|_p / (E_1 + E_2 (1 \pm n^{-c+1}))^{1/p}$ where $E_1,E_2$ are independent exponentials. So with probability $7/8$, we have all of $|z_{D(1)}|  = \Theta(\|F\|_p)$, $|z_{D(2)}| = \Theta(\|F\|_p)$ and  $|z_{D(1)}| - |z_{D(2)}|  = \Theta(\|F\|_p)$ with sufficiently scaled constants, so scaling $\nu$ by a sufficiently small constant we have $|\zeta_{D(1)}| = \Theta(\|F\|_p)$, $|\zeta_{D(1)}| - |\zeta_{D(2)}|  = \Theta(\|F\|_p)$ and $|\zeta_{D(2)}| = \Theta(\|F\|_p)$. Conditioned on the event in the prior paragraph and on the high probability success of our $L_2$ estimation algorithm and our count-sketch error, our estimates of $|\zeta_{D(1)}|, |\zeta_{D(2)}|$ via $y$ are $\Theta(1)$-relative error estimates, so for $\eps$ small enough the maximum indices in $y$ and $\zeta$ will coincide, and we will have both $y_{(1)} - y_{(2)} > 100 \eps \mu R  = O(\|F\|_p)$ and $y_{(2)} > 50 \eps \mu R =O(\|F\|_p)$. By a union bound, it follows that this condition holds with probability at least $1 - (1/10 +99/100 + 1/8 + O(n^{-c}))> 1/2$ as desired.

	\end{proof}

Putting together the results of this section, we obtain the correctness of our algorithm as stated in Theorem \ref{thm:mainold}.
	\begin{theorem}\label{thm:mainold}
		Given any constant $c \geq 2$, $\nu \geq n^{-c}$, and $0 < p \leq 2$, there is a one-pass $L_p$ sampler which returns an index $i \in [n]$ such that $\pr{i = j} = \frac{|f_j|^p}{\|F\|_p^p}(1\pm \nu)+n^{-c}$ for all $j \in [n]$, and which fails with probability $\delta > 0$. The space required is $O(\log^2(n) \log(1/\delta)(\log \log n)^2)$ bits for $p < 2$,  and $O(\log^3(n) \log(1/\delta)(\log \log n)^2)$ bits for $p = 2$. 
	\end{theorem}
	
	\begin{proof}
 We claim that conditioned on not failing, we have that $i^* = \arg \max_i \{y_i\} = \arg \max_i \{|z_i|\}$. First, condition on the success of our count-sketch estimator, and on the guarantees of our estimate $R$, which occur with probability $1 - n^{-c}$ together. Since the gap between the two largest coordinates in $y$ is at least $100 \eps \mu R > 20 \eps \|\zeta \|_2 \geq 20  \|y - |\zeta|\|_\infty$ ($20$ times the additive error in estimating $|\zeta|$), it cannot be the case that the index of the maximum coordinate in $y$ is different from the index of the maximum coordinate (in absolute value) in $\zeta$, and moreover both $y$ and $\zeta$ must have a unique maximizer. Then we have $|\zeta_{i^*}| - |\zeta_{(2)}| = |\zeta_{(1)}| - |\zeta_{(2)}| > 18 \eps \|\zeta\|_2$, and since $z_{i} = (1 \pm  O(\nu))\zeta_i$ for all $i$, we have $\| |\zeta| - |z| \|_\infty \leq O( \nu) \|\zeta\|_2$. Scaling $\nu$ down by a factor of $\eps = \Omega(\sqrt{1/\log(n)})$ (which is absorbed into the $\tilde{O}(\nu^{-1})$ update time), the gap between the top two items in $\zeta$ is $18$ times large than the additive error in estimating $z$ via $\zeta$. Thus we must have $i^* = \arg \max_i  \{|\zeta_i\|\} = \arg \max_i \{|z_i|\}$, which completes the proof of the claim.		
	
			Now Lemma \ref{lem:main2} states for any $i_j \in [n^c]$ that $\pr{\neg \text{FAIL} \:| \: i_j = \arg \max_{i',j'} \{|z_{i_{j'}'}|\}} = \pr{\neg \text{FAIL}} \pm O(\log(n) \nu) = q \pm O(\log(n)\nu)$, where $q =  \pr{\neg \text{FAIL}} = \Omega(1)$ is a fixed constant, by Lemma \ref{lem:failbound2}, which does not depend on any of the randomness in the algorithm.
		Since conditioned on not failing we have  $\arg \max_i \{y_i\} = \arg \max_i \{|z_i|\}$, the probability we output $i_j \in [n^c]$ is $\pr{\neg \text{FAIL} \cap y_{i_j} \text{ is the max in } y } = \pr{\neg \text{FAIL} \cap |z_{i_j}| \text{ is the max in } |z|}$ (conditioned on the high probability events in the prior paragraph), so the probability our final algorithm outputs $i \in [n]$ is 
		 \[\sum_{j \in [n^{c-1}]}\allowbreak \pr{\neg \text{FAIL} \:| \: i_j = \arg \max_{i',j'}\{|z_{i_{j'}'}|\}}\pr{ i_j = \arg \max_{i',j'}\{|z_{i_{j'}'}|\}}  = \sum_{j \in [n^{c-1}]} \frac{|f_i|^p}{\|F\|_p^p} ( q\pm O(\log(n)\nu))\]\[ =  \frac{|f_i|^p}{\|f\|_p^p}(q\pm O(\log(n)\nu))\]
		  The potential of the failure of the various high probability events that we conditioned on only adds another additive $O(n^{-c})$ term to the error. Thus, conditioned on an index $i$ being returned, we have $\pr{i = j} = \frac{|f_j|^p}{\|f\|_p^p}(1\pm O(\log(n)\nu)) \pm n^{-c}$ for all $j \in [n]$, which is the desired result after rescaling $\nu$ down by a factor of $\Omega(1/\log(n))$ (we need only scale down by $\Omega(1/\sqrt{\log(n)})$ after already rescaling by $\eps = \Theta(1/\sqrt{\log(n)}$ when $p=2$).
	  Running the algorithm $O(\log(\delta^{-1}))$ times in parallel, it follows that at least one index will be returned with probability $1-\delta$.

	   Theorem \ref{thm:derandomold} shows that the entire algorithm can be derandomized to use a random seed with $O(\log^2(n)(\log \log n)^2)$-bits for $p<2$ and  $O(\log^3(n)(\log \log n)^2)$-bits for $p=2$, which dominates the space required to store the sketches of the sampling algorithm themselves. Repeating $O(\log(1/\delta))$ times to obtain $\delta$ failure probability gives the stated space bounds.

\end{proof}

\begin{remark}\label{rem:runtimeold}
	Using roughly the same update-procedures and a similar analysis as in Section \ref{sec:runtime}, one can implement the above $L_p$ sampling algorithm to have $\tilde{O}(\nu)$ update time and $\tilde{O}(1)$ report time, just as in Theorem \ref{thm:main}. The only difference is the use of Rademacher $\{1,-1\}$ variables in the count-sketch instead of Gaussians, and the change to make the variables $g_{i,j,k}$ independent. These Rademacher variables are easier to handle, as one can just compute, for a given bucket $A_{i,j}$ of count-sketch, the number of items which hash into this bucket with a $1$ and $-1$ sign, and add the corresponding value to that bucket. This is simply another computation of a binomial random variable. The variables $g_{i,j,k}$ can be handled in \texttt{Fast-Update} by a modifying the procedure to draw a binomial to determine how many items hash to each bucket $A_{i,j}$ independently for each $j \in [k]$. This is as opposed to the \texttt{Fast-Update} of Figure \ref{fig:update-CS}, which only allows an item to be hashed into a single bucket in each row of $A$. In other words, we change Figure \ref{fig:update-CS} to deal with the modified variables $g_{i,j,k}$ by simply removing step $1(d)$ which decrements the value of $W_k$, which is the counter of items left to be hashed in a row $k$ of $A$.
	
	 To show that the output of this algorithm is the same when only searching through a subset $\mathcal{K}$ of the coordinates (where $\mathcal{K}$ is as in Section \ref{sec:runtime}) for the maximizers $y_{(1)},y_{(2)}$, observe that the test $y_{(2)} \geq 50 \eps \mu R$ enforces that, conditioned on not failing, both $y_{(1)}$ and $y_{(2)}$ will be large enough to be contained in the set $\mathcal{K}$. Thus we can safely implement the \texttt{Fast-Update} procedure to give improved update time, and the \texttt{ExpanderSketch} of Theorem \ref{thm:expander} to obtain the improved query time.
\end{remark}

\section{Derandomizing the Original Algorithm}
We now show how our original algorithm can be derandomized using the same techniques as in Section \ref{sec:runtime}. For this section, we let $B \in \R^{O(\log(n))}$ be the sketch stored for the high probability $L_2$ estimation used in the $L_p$ sampler as in Lemma \ref{lem:L2}. Note that $B = G \cdot \zeta$, where $G$ is a matrix of i.i.d. Gaussian variables.

\begin{theorem}\label{thm:derandomold}
	The algorithm of Section \ref{sec:samplerOld} can be derandomized to run in $O(\log^2(n) \log(1/\delta) \allowbreak (\log \log(n))^2)$ space for $p<2$, and $O(\log^3(n) \log(1/\delta) (\log \log(n))^2)$ space for $p=2$.
\end{theorem}
\begin{proof}
	We use the same notation $\mathcal{A}(r_e,r_c)$ as in Theorem \ref{thm:derandom}. Recall here that $r_e$ is the randomness required for the exponentials, and $r_c$ is the randomness required for count-sketch (and now $r_c$ must also include the randomness required for the $L_2$ estimation sketch $B$).  For any fixed randomness $r_c$, let $\mathcal{A}_{r_c}(r_e)$ be the tester which tests if our $L_p$ sampler would output the index $i$, where now the bits $r_c$ are hard-coded into the tester, and the random bits $r_e$ for the exponentials are taken as input. 
	
	Now note that the entire sketch stored by our algorithm can be written as $Z \cdot \zeta$, where $Z \in \R^{O(\log(n) / \eps^2) \times n}$ is a fixed matrix defined by the count-sketch randomness $r_c$, and $\zeta$ is the scaled (by inverse exponentials) and rounded stream vector of the algorithm. Here $Z \cdot \zeta = [\texttt{vec}(A); \texttt{vec}(B)]$ where $\texttt{vec}(A)$ denotes the vectorization of the count-sketch matrix $A$ (and resp. $B$), and $[x;y]$ is vector which stacks $x$ on top of $y$. Note that we can pull the scalings by $F$ into the matrix $Z$ (making it into a new fixed matrix $Z'$), so our sketch can be written as $Z' \cdot t$, where for $j \in [n^c]$ we have $t_j = \texttt{rnd}_\nu (1/t_j^{1/p})$ and $t_j$'s are the i.i.d. exponentials.

	Since we are rounding the exponentials to powers of $(1 + \nu)$ anyway, we can restrict the support of the coordinates in $t$ to a discrete support of size $O(\poly(n))$ such that each value occurs with probability at least $1/\poly(n)$ for a suitably larger $\poly(n)$. This allows us to sample the variables $\texttt{rnd}_\nu(1/t_i^{1/p})$  using $O(\log(n))$-bits of space as needed for Lemma \ref{lemma:halfspace}.  Thus our entire algorithm requires $\poly(n)$ random bits to be generated for the exponentials. Similarly, for the random Gaussians used to estimate the $L_2$ in the sketch $B$, one can truncate to $O(\log(n))$-bits, incurring only an additive $n^{-c}$ error in these buckets, which can be absorbed in to the adversarial error which is already handled in Lemma \ref{lem:main2}. Restricting the support of the Gaussians so that each value occurs with probability at least $1/\poly(n)$, it follows that these Gaussians can also be sampled using $O(\log(n))$-bits each. The only remaining randomness are the random signs and $h_{i,j,k}$ in count sketch, each of which have a support of size $2$ and can be sampled with $O(\log(n))$-bits.  So using Lemma \ref{lemma:halfspace}, we can fool the tester which tests if $Z' \cdot t = y$ for any $y$ with $O(\log(n))$ bounded bit-complexity, using a seed of $O(\log^2(n) (\log \log n)^2)$ bits (and $O(\log^3(n) (\log \log n)^2)$ for $p=2$). Then as in Theorem \ref{thm:derandomGeneral}, since we can fool $\pr{Z' \cdot t = y}$, we can also fool any tester which takes as input $y=Z' \cdot t$ and outputs whether or not on input $y$ our algorithm would output $i \in [n]$. Thus if $G(x)$ is one instance of the PRG from Lemma \ref{lemma:halfspace}, we have
	$\pr{\mathcal{A}_{r_c}(r_e)} \sim_{n^{-O(\log(n))}} \pr{\mathcal{A}_{r_c}(G(x))}$, and similarly, as in Theorem \ref{thm:derandom}
	\begin{equation*}
	\begin{split}
	\bpr{\mathcal{A}(r_e,r_c)}  & = 	\sum_{r_c}  \bpr{\mathcal{A}_{r_c}(r_e)}  \bpr{r_c}  \\
	& =\sum_{r_c} ( (\bpr{\mathcal{A}_{r_c}(G(x))} \pm n^{-O(\log(n))}) \bpr{r_c} \\
	& =\sum_{r_c}  (\bpr{\mathcal{A}_{r_c}(G(x))}\bpr{r_c} \pm \sum_{r_c} n^{-O(\log(n))}  \bpr{r_c}  \\ 
	& \sim_{n^{-O(\log(n))}} \bpr{\mathcal{A}(G(x),r_c)} \\		
	\end{split}
	\end{equation*}
	Now fix any seed $G(x)$, and consider $\mathcal{A}_{G(x)}(r_c)$ which on fixed exponential randomness $G(x)$ and fresh count-sketch randomness $r_c$, tests whether out algorithm would output $i \in [n]$. Note that this algorithm simply maintains the same sketch $Z \cdot \zeta = [\texttt{vec}(A), \texttt{vec}(B)]$ as above. Note that the entries of $Z$ are of two forms: the i.i.d. count-sketch randomness and the i.i.d. Gaussians needed for the sketch $B$. By Theorem \ref{thm:derandomGeneral}, we can derandomize both of these separately by two more instances $G(x_2),G(x_3)$ of the PRG of Lemma \ref{lemma:halfspace}, each using seeds $x_2,x_3$ of $O(\log^2(n) (\log\log n)^2)$ bits of space for $p<2$ and $O(\log^3(n) (\log\log n)^2)$ bits of space for $p=2$. So if $Z_1$ is the first set of rows of $Z$ which correspond to the count-sketch randomness, and $Z_2$ is the rest of the rows which contain i.i.d. Gaussians, we have that for all $y,y'$ with $O(\log(n))$-entrywise bounded bit complexity: $\pr{Z_1\cdot  \zeta = y} \sim_{n^{-O(\log(n))}} \pr{G(x_2) \cdot \zeta = y}$ and $\pr{Z_2\cdot  \zeta = y'} \sim_{n^{-O(\log(n))}} \pr{G(x_3) \cdot \zeta = y'}$. Here we are abusing notation and thinking of the PRG randomness $G(x_2)$ as being formed into the matrix which it defines.
	
	 Since $G(x_2)$ is independent of $G(x_3)$, for any $y$ of $O(\log(n))$-entrywise bounded bit complexity, we have $\pr{Z \cdot \zeta = y} \sim_{n^{-O(\log(n))}} \pr{[G(x_2) ; G(x_3)] \cdot \zeta = y}$. Thus we fool the entire tester $\mathcal{A}_{G(x)}(r_c)$ with $\mathcal{A}_{G(x)}(G(x_2) \cup G(x_3))$, meaning $\pr{\mathcal{A}_{G(x)}(r_c)} \sim_{n^{-O(\log(n))}} \pr{\mathcal{A}_{G(x)}(G(x_2) \cup G(x_3))}$, and by a similar averaging arguement as above, we have $\pr{\mathcal{A}(G(x), r_c)} \sim_{n^{-O(\log(n))}} \pr{\mathcal{A}(G(x), G(x_2) \cup G(x_3))}$, thus $	\pr{\mathcal{A}(r_e,r_c)} \sim_{n^{-O(\log(n))}} \pr{ \mathcal{A}(G(x), G(x_2) \cup G(x_3))}$, which completes the proof. We note that any coordinate output by the PRG of Lemma \ref{lemma:halfspace} (and thus Theorem \ref{thm:derandomGeneral}) can be computed in space linear in the seed length required by Proposition \ref{prop:PRGspacetime}, thus the space required to evaluate the generator is linear in the seed length.
\end{proof}

\end{document}